\documentclass[conference,compsoc]{IEEEtran}

\usepackage{ascmac}
\usepackage{amsmath}
\usepackage{amssymb}
\usepackage{stmaryrd}
\usepackage{graphicx}
\usepackage{amsthm}
\usepackage{thmtools}
\usepackage{color}
\usepackage{enumitem}
\usepackage[hyphens]{url}
\usepackage{dsfont}
\usepackage{multirow}
\usepackage[linesnumbered, lined, boxed]{algorithm2e}
\usepackage{enumitem}
\usepackage{comment}
\newcommand{\nats}{\mathbb{N}}
\newcommand{\ints}{\mathbb{Z}}
\newcommand{\nnints}{\mathbb{Z}_{\ge0}}
\newcommand{\reals}{\mathbb{R}}
\newcommand{\nnreals}{\mathbb{R}_{\ge0}}
\newcommand{\mathone}{\mathds{1}} % dsfont package

\renewcommand{\epsilon}{\varepsilon}

\newcommand{\calA}{\mathcal{A}}

\newcommand{\calD}{\mathcal{D}}

\newcommand{\calG}{\mathcal{G}}

\newcommand{\calM}{\mathcal{M}}
\newcommand{\calR}{\mathcal{R}}
\newcommand{\calS}{\mathcal{S}}
\newcommand{\calT}{\mathcal{T}}

\newcommand{\calY}{\mathcal{Y}}

\newcommand{\E}{\operatorname{\mathbb{E}}}
\newcommand{\V}{\operatorname{\mathbb{V}}}

\newcommand{\bmf}{\mathbf{f}}

\newcommand{\bmx}{\mathbf{x}}
\newcommand{\bmy}{\mathbf{y}}
\newcommand{\bmz}{\mathbf{z}}

\newcommand{\hf}{\hat{f}}

\newcommand{\tf}{\tilde{f}}
\newcommand{\tih}{\tilde{h}}
\newcommand{\tn}{\tilde{n}}
\newcommand{\tx}{\tilde{x}}

\newcommand{\bmhf}{\mathbf{\hf}}
\newcommand{\bmth}{\mathbf{\tilde{h}}}

\newcommand{\GMGA}{G_{\text{MGA}}}

\def\SBin{\textsf{SBin-Shuffle}}

\def\SAGeo{\textsf{SAGeo-Shuffle}}
\def\SOGeo{\textsf{S1Geo-Shuffle}}
\def\GRRS{\textsf{GRR-Shuffle}}
\def\OUES{\textsf{OUE-Shuffle}}
\def\OLHS{\textsf{OLH-Shuffle}}
\def\RAPS{\textsf{RAPPOR-Shuffle}}
\def\Shuffle{\textsf{-Shuffle}}
\def\BC{\textsf{BC20}}
\def\CM{\textsf{CM22}}
\def\LWY{\textsf{LWY22}}

\def\Bin{\textsf{Bin}}
\def\AGeo{\textsf{AGeo}}
\def\OGeo{\textsf{1Geo}}
\def\ageo{\textsf{ageo}}
\def\nonageo{\textsf{non-ageo}}

\def\onegeo{\textsf{1geo}}

\newcommand{\colorB}[1]{\textcolor{black}{#1}}

\newtheorem{definition}{Definition}
\newtheorem{theorem}{Theorem}
\newtheorem{lemma}{Lemma}
\newtheorem{proposition}{Proposition}

\SetKwInput{KwData}{Input}
\SetKwInput{KwResult}{Output}

\usepackage{hyperref}

\newif\ifconferenceon\conferenceontrue
\ifconferenceon
\newcommand{\conference}[1]{#1}
\newcommand{\arxiv}[1]{}
\else
\newcommand{\conference}[1]{}
\newcommand{\arxiv}[1]{#1}
\usepackage{balance}
\fi

\ifCLASSOPTIONcompsoc
  \usepackage[nocompress]{cite}
\else
  \usepackage{cite}
\fi

\ifCLASSINFOpdf
\else
\fi

\begin{document}
\title{Augmented Shuffle Protocols for Accurate and Robust Frequency Estimation\\ under Differential Privacy}

\author{\IEEEauthorblockN{Takao Murakami}
\IEEEauthorblockA{ISM/AIST\\
Email: tmura@ism.ac.jp}
\and
\IEEEauthorblockN{Yuichi Sei}
\IEEEauthorblockA{UEC\\
Email: seiuny@uec.ac.jp}
\and
\IEEEauthorblockN{Reo Eriguchi}
\IEEEauthorblockA{AIST\\
eriguchi-reo@aist.go.jp}}

\maketitle

\begin{abstract}
The shuffle model of DP (Differential Privacy) provides high utility by introducing a shuffler that randomly shuffles noisy data sent from users. 
However, recent studies show that existing shuffle protocols suffer from the following two major drawbacks. 
First, they are vulnerable to local data poisoning attacks, which manipulate the statistics about input data by sending crafted data, especially when the privacy budget $\epsilon$ is small. 
Second, the actual value of $\epsilon$ is increased by collusion attacks by the data collector and users. 

In this paper, we 
address these two issues by 
thoroughly exploring the potential of 
the \textit{augmented shuffle model}, which allows the shuffler to perform additional operations, such as random sampling and dummy data addition. 
Specifically, we propose a generalized framework for \textit{local-noise-free protocols} in which users send (encrypted) input data to the shuffler without adding noise. 
We show that this generalized protocol provides DP and is robust to 
the above two attacks 
if a simpler mechanism that performs the same process on binary input data provides DP. 
Based on this framework, 
we propose three concrete protocols providing DP and robustness against the two attacks. 
Our first protocol 
\colorB{generates the number of dummy values for each item from a binomial distribution}
and provides higher utility than several state-of-the-art existing shuffle protocols. 
Our second protocol significantly improves the utility of our first protocol by introducing a novel \colorB{dummy-count} distribution: \textit{asymmetric two-sided geometric distribution}. 
Our third protocol is a special case of our second protocol and provides pure $\epsilon$-DP. 
We show the effectiveness of our protocols through theoretical analysis and comprehensive experiments. 
\end{abstract}

\IEEEpeerreviewmaketitle

\section{Introduction}
\label{sec:intro}
DP (Differential Privacy)~\cite{Dwork_ICALP06,DP} is well known as the gold standard for private data analysis. 
It offers strong privacy guarantees when a parameter (privacy budget) $\epsilon$ is small. 
To date, numerous studies have been made on 
central DP or LDP (Local DP)~\cite{Kasiviswanathan_FOCS08,Duchi_FOCS13}. 
Central DP assumes a model in which a single central server has personal data of all users and obfuscates some statistics (e.g., mean, frequency distribution) about the data. 
The major drawback of this model is that 
all personal data are held by a single server and, therefore, 
might be leaked from the server by cyber-attacks~\cite{data_breach}. 
LDP addresses this issue by introducing a model in which users obfuscate their personal data by themselves before sending them to a data collector. 
Since LDP does not assume any trusted third party, it does not suffer from data leakage issues. 
However, LDP destroys data utility in many practical scenarios, as a large amount of noise is added to each user's data. 

Since Google implemented Prochlo~\cite{Bittau_SOSP17}, shuffle DP~\cite{Balcer_ITC20,Balcer_SODA21,Balle_CRYPTO19,Balle_CCS20,Balle_NeurIPS20,Cheu_EUROCRYPT19,Cheu_SP22,Erlingsson_SODA19,Feldman_FOCS21,Ghazi_EUROCRYPT21,Girgis_AISTATS21,Girgis_CCS21,Girgis_NeurIPS21,Imola_CCS22,Li_TDSC23,Liu_AAAI21,Luo_CCS22,Meehan_ICLR22,Wang_PVLDB20} has been 
extensively 
studied to bridge the gap between central DP and LDP. 
Shuffle DP introduces an intermediate server called the \textit{shuffler}, and works as follows. 
First, each user adds noise to her input value, typically using an LDP mechanism. 
Then, she encrypts her noisy value and sends it to the shuffler. 
The shuffler randomly shuffles the (encrypted) noisy values of all users and sends them to the data collector. Finally, the data collector decrypts them. 
Although most shuffle protocols only allow shuffling operations on the shuffler side, some shuffle protocols~\cite{Bittau_SOSP17,Girgis_NeurIPS21,Wang_PVLDB20} (including 
Prochlo~\cite{Bittau_SOSP17}) allow additional operations, such as random sampling and dummy data addition, to the shuffler. 
To differentiate between the two, we refer to the former model as a \textit{pure shuffle model} and the latter as an \textit{augmented shuffle model}. 

Since the shuffling 
removes the linkage between users and noisy values, it amplifies the privacy guarantees. 
Specifically, 
it can significantly reduce the privacy budget 
$\epsilon$ under the assumption that the data collector does not collude with the shuffler. 
Thus, the shuffle model can significantly improve the utility of LDP at the same value of $\epsilon$.

However, the existing shuffle protocols suffer 
from two major drawbacks. 
First, 
they are 
vulnerable to \textit{local data poisoning attacks}~\cite{Cao_USENIX21}, 
which inject some fake users and craft 
data sent from the fake users 
to increase the estimated frequencies of some target items (i.e., to promote 
these items). 
In particular, 
Cao \textit{et al.}~\cite{Cao_USENIX21} show 
that LDP 
protocols 
are less secure against the poisoning attacks 
when the privacy budget $\epsilon$ is smaller. 
This is called a fundamental security-privacy trade-off~\cite{Cao_USENIX21} 
-- LDP 
protocols 
cannot provide high privacy and high robustness against data poisoning simultaneously. 
This is caused by the fact that genuine users need to add LDP noise to their input values, whereas fake users do not. 
Since most existing shuffle protocols apply LDP mechanisms on the user side, they also suffer from the fundamental security-privacy trade-off. 
Some studies~\cite{Balcer_ITC20,Luo_CCS22} propose multi-message protocols, where each user sends multiple noisy values that do not provide LDP. 
However, as this paper shows, these protocols are also vulnerable to the poisoning attacks for the same reason -- 
genuine users need to add noise, whereas fake users do not. 

Second, 
the existing shuffle protocols are 
vulnerable to \textit{collusion attacks by the data collector and users} 
(referred to as \textit{collusion with users} for shorthand). 
Specifically, 
Wang \textit{et al.}~\cite{Wang_PVLDB20} point out 
that the data collector may collude with some users (or compromise some user accounts) to obtain their 
noisy values. 
In this case, the privacy budget $\epsilon$ 
increases with increase in the number of colluding users (as formally proved in 
Section~\ref{sub:poisoning}), 
as the data collector can reduce the number of shuffled values. 
This issue is crucial because it is difficult to know the number of colluding users in practice. 
In other words, it is difficult to know the actual value of $\epsilon$ in the existing shuffle model. 

Note that these two issues are inevitable 
in the pure shuffle model. 
Specifically, the shuffler only performs random shuffling in the pure shuffle model. 
Thus, users need to add noise to their input values to ensure privacy. 
Then, this model is vulnerable to the poisoning and collusion attacks, as (i) fake users do not need to add noise, and (ii) the data collector can increase $\epsilon$ by obtaining some noisy values. 

In this work, we make the first attempt to address the above two issues by 
thoroughly exploring the potential of the \textit{augmented shuffle model}. 
In particular, we propose new protocols in the augmented shuffle model: \textbf{local-noise-free protocols}. 
In our 
protocols, users encrypt their input values \textit{without adding any noise} and send them to the shuffler. 
The shuffler randomly shuffles the encrypted values of all users and sends them to the data collector, who then decrypts them. 
Since the random shuffling alone cannot provide DP in this case, 
we introduce two additional operations on the shuffler side before shuffling: 
random sampling and dummy data addition. 
The former randomly selects each encrypted value with some probability, whereas the latter adds encrypted dummy values. 
The shuffler can easily perform these two operations because they can be done without decrypting data sent from users. 
Our protocols are 
robust to 
both data poisoning and collusion with users, 
as 
they do not 
add any noise on the user side -- 
they add 
noise on the shuffler side without seeing the input values of users. 
To our knowledge, we are the first to introduce such 
protocols (see Section~\ref{sec:related} for more details). 

We focus on frequency estimation~\cite{Erlingsson_CCS14,Fanti_PoPETs16,Kairouz_ICML16,Wang_USENIX17,Wang_PVLDB20}, in which the data collector estimates a frequency of input values for each item \colorB{(i.e., histogram)}, and propose 
a generalized framework for local-noise-free protocols. 
In our framework, 
the shuffler performs random sampling and dummy data addition, where the number of dummy values for each item follows a pre-determined distribution \colorB{called a \textit{dummy-count distribution}}. 
\colorB{Our key insight is that random sampling and adding dummies, when followed by shuffling, are equivalent to adding discrete noise to each bin of the histogram.} 
We reduce DP of this generalized protocol to that of a simpler mechanism called a \textit{binary input mechanism}, which performs 
the same random sampling and dummy data addition process on binary input data. 
Specifically, we prove that if the 
binary input mechanism 
provides DP, our generalized protocol also provides DP. 
Moreover, we prove that, in this case, our generalized protocol is robust to both the poisoning and collusion attacks. 
In particular, we prove that our generalized protocol 
does not suffer from the security-privacy trade-off (i.e., the robustness 
to 
poisoning does \textit{not} depend on $\epsilon$) and that the value of $\epsilon$ is \textit{not} increased by collusion with users. 
We also analyze the estimation error and communication cost of our generalized protocol. 

Based on our theoretical analysis, we propose three concrete local-noise-free protocols in our framework. 
In each of the three cases, the corresponding binary input mechanism provides DP. 
Thus, all of our three protocols provide DP and are robust to the poisoning and collusion attacks. 

Our first protocol, 
\SBin{} (Sample, Binomial Dummies, and Shuffle), is a simple local-noise-free protocol 
\colorB{that uses a binomial distribution as a dummy-count distribution}. 
For this protocol, 
we present a tighter bound on the number of trials in the binomial distribution 
than the previous results~\cite{Dwork_EUROCRYPT06,Agarwal_NeurIPS18}.
Using our bound, 
we show that 
\SBin{} achieves a smaller estimation error than 
state-of-the-art shuffle protocols, including the ones based on the GRR (Generalized Randomized Response)~\cite{Kairouz_ICML16,Wang_USENIX17}, 
OUE (Optimized Unary Encoding)~\cite{Wang_USENIX17}, OLH (Optimized Local Hashing)~\cite{Wang_USENIX17}, and RAPPOR~\cite{Erlingsson_CCS14}, and three multi-message protocols in~\cite{Balcer_ITC20,Cheu_SP22,Luo_CCS22}. 
This means that \SBin{} outperforms all of these seven existing shuffle protocols in terms of both robustness and accuracy. 

However, \SBin{} has room for improvement in two aspects. 
First, the accuracy can still be significantly improved, as shown in this paper. 
Second, \SBin{} provides only $(\epsilon,\delta)$-DP (i.e., approximate DP)~\cite{DP} and cannot provide $\epsilon$-DP (i.e., pure DP). 
Note that all of the existing shuffle protocols explained above have the same issue. 
That is, how to provide $\epsilon$-DP in the shuffle model remains open.

We address these two issues by introducing two additional protocols. 
Specifically, our second protocol, \SAGeo{} (Sample, Asymmetric Two-Sided Geometric Dummies, and Shuffle), 
significantly improves the accuracy of 
\SBin{} by introducing a novel \colorB{dummy-count} distribution: \textit{asymmetric two-sided geometric distribution}. 
To our knowledge, we are the first to introduce this distribution to provide DP. 
As the sampling probability in random sampling decreases, the left-hand curve in this distribution becomes steeper, reducing the variance of the \colorB{dummy-count} distribution. 
We show that \SAGeo{} is very promising in terms of accuracy. 
For example, our experimental results show that \SAGeo{} outperforms \SBin{} (resp.~the existing shuffle protocols) by one or two (resp.~two to four) orders of magnitude in terms of squared error.

Our third protocol, \SOGeo{} (Sample, One-Sided Geometric Dummies, and Shuffle), is a special case of \SAGeo{} where the sampling probability is $1 - e^{-\frac{\epsilon}{2}}$. 
\colorB{In this case, a one-sided geometric distribution is used as a dummy-count distribution}. 
A notable feature of \SOGeo{} is that it provides pure $\epsilon$-DP, offering one solution to the open problem explained above.

\smallskip{}
\noindent{\textbf{Our Contributions.}}~~Our contributions are as follows: 
\begin{itemize}
\item We propose a generalized framework for local-noise-free protocols in the augmented shuffle model. 
We 
rigorously 
analyze the privacy, robustness, estimation error, and communication cost of our generalized protocol. 
In particular, we prove that if 
a simpler binary input mechanism 
provides DP, 
our generalized protocol also provides DP and, moreover, is robust to both data poisoning and collusion with users. 
\item We propose three DP protocols in our framework: \SBin{}, \SAGeo{}, and \SOGeo{}. 
We prove that 
\SBin{} provides higher accuracy than seven state-of-the-art shuffle protocols: four single-message protocols based on the GRR, OUE, OLH, and RAPPOR, and three multi-message protocols in~\cite{Balcer_ITC20,Cheu_SP22,Luo_CCS22}. 
Furthermore, we prove that \SAGeo{} significantly improves the accuracy of \SBin{} and that \SOGeo{} provides pure $\epsilon$-DP. 
\item We show the effectiveness of our protocols through comprehensive experiments, which compare our protocols with seven existing protocols explained above and three existing defenses against the poisoning and collusion attacks~\cite{Cao_USENIX21,Sun_ICDE24,Wang_PVLDB20} using four real datasets. 
\end{itemize}
We published our source code on GitHub~\cite{SDPwoLN}.

\section{Related Work}
\label{sec:related}

\noindent{\textbf{Pure/Augmented Shuffle DP.}}~~Most 
existing shuffle protocols assume the pure shuffle model, where the shuffler only shuffles data. 
For example, privacy amplification bounds in this model are analyzed in~\cite{Erlingsson_SODA19,Balle_CRYPTO19,Cheu_EUROCRYPT19,Feldman_FOCS21,Girgis_CCS21}. 
The pure shuffle model is applied to tabular data in~\cite{Li_TDSC23,Meehan_ICLR22,Wang_PVLDB20}, 
federated learning in~\cite{Balle_NeurIPS20,Girgis_AISTATS21,Girgis_NeurIPS21,Liu_AAAI21}, and graph data in~\cite{Imola_CCS22}. 
Multi-message protocols, in which a user sends multiple noisy values to the shuffler, 
are studied for real summation~\cite{Balle_CCS20} and frequency estimation~\cite{Balcer_ITC20,Cheu_SP22,Ghazi_EUROCRYPT21,Luo_CCS22} in the pure shuffle model. 
The connection between the pure shuffle model and pan-privacy is studied in~\cite{Balcer_SODA21}. 

A handful of studies assume the augmented shuffle model, which allows additional operations, such as random sampling and dummy data addition, to the shuffler. 
For example, 
Prochlo~\cite{Bittau_SOSP17} introduces randomized thresholding, where the shuffler discards reports from users if the number of reports is below a randomized threshold. 
Girgis \textit{et al.}~\cite{Girgis_NeurIPS21} introduce a subsampled shuffle protocol in which the shuffler randomly samples users to improve privacy amplification bounds. 
Wang \textit{et al.}~\cite{Wang_PVLDB20} propose an augmented shuffle protocol 
that injects uniformly sampled dummy values on the shuffler side as a defense against collusion with users 
(see ``Collusion Attacks'' in this section for more details). 

All of the above studies on the pure/augmented shuffle model 
assume that users add 
noise to their input values. 
Thus, they are vulnerable to 
poisoning and collusion attacks, 
as 
explained in Section~\ref{sec:intro}. 
Our local-noise-free protocols address this issue by not requiring users to add noise.

\smallskip{}
\noindent{\textbf{Data Poisoning Attacks.}}~~Data poisoning attacks against LDP 
protocols 
have been recently studied~\cite{Cao_USENIX21,Wu_USENIX22,Cheu_SP21,Imola_arXiv22}. 
Their results can be applied to most existing shuffle protocols, as they use LDP mechanisms on the user side. 
For example, targeted and untargeted poisoning attacks for categorical data are proposed in~\cite{Cao_USENIX21} and~\cite{Cheu_SP21}, respectively. 
Targeted poisoning attacks for key-value data and graph data are studied in~\cite{Wu_USENIX22} and~\cite{Imola_arXiv22}, respectively. 
We focus on targeted attacks for categorical data in the same way as~\cite{Cao_USENIX21}. 

Cao \textit{et al.}~\cite{Cao_USENIX21} explore some defenses against poisoning attacks. 
Among them, a normalization technique, which normalizes a minimum estimate to $0$, performs the best when the number of targeted items is one or two. 
In addition, only the normalization technique is applicable to the GRR. 
Sun \textit{et al.}~\cite{Sun_ICDE24} propose LDPRecover, which recovers the frequencies of genuine users by eliminating the impact of fake users. 
Although they also propose another defense called LDPRecover*, it assumes that the server perfectly knows the set of target items selected by the adversary in advance, which does not hold in practice. 
A detection method in~\cite{Huang_TKDE24} has the same issue.

In our experiments, we evaluate the normalization technique~\cite{Cao_USENIX21} and LDPRecover~\cite{Sun_ICDE24} as defenses against poisoning attacks. 
Our results show that 
their effectiveness is limited, especially when $\epsilon$ is small.

\smallskip{}
\noindent{\textbf{Collusion Attacks.}}~~Wang \textit{et al.}~\cite{Wang_PVLDB20} point out that the data collector can collude with users to increase $\epsilon$ in the shuffle model. 
They also propose an augmented shuffle protocol, where the shuffler injects dummy values uniformly sampled from the range of LDP mechanisms, 
as a defense against this collusion attack. 
However, their protocol is still vulnerable to 
this collusion, 
as it adds LDP noise on the user side. 

In our experiments, we evaluate the defense in~\cite{Wang_PVLDB20} and show that $\epsilon$ in their defense is still increased (and $\epsilon$ in
our protocols is not increased) by collusion with users. 

\smallskip{}
\noindent{\textbf{MPC-DP.}}~~Finally, we note that our local-noise-free protocols are related to the 
MPC (Multi-Party Computation)-DP model based on two servers~\cite{Bell_CCS22}. 
Specifically, in their MPC-DP protocol, users do not add noise to their input values, and two servers calculate sparse histograms from their encrypted input values by using homomorphic encryption. 

Our local-noise-free protocols differ from their MPC-DP protocol in two ways. 
First, the protocol in~\cite{Bell_CCS22} requires as many as four rounds of interaction between the two servers and is, therefore, inefficient. 
In contrast, our protocols need only one round between the shuffler and the data collector. 
Second, the protocol in~\cite{Bell_CCS22} assumes that an underlying public-key encryption scheme is \textit{homomorphic}, i.e., it needs to support computation over encrypted data.
It limits the class of encryption schemes that the protocol can be based on.
In contrast, our protocols can be based on any public-key encryption scheme \colorB{(e.g., RSA, 
ECIES, EC-ElGamal).}

\section{Preliminaries}
\label{sec:preliminaries}

\subsection{Notations}
\label{sub:notations}
Let $\nats$, $\ints$, $\nnints$, $\reals$, $\nnreals$ be the sets of natural numbers, integers, non-negative integers, real numbers, and non-negative real numbers, respectively. 
For $a \in \nats$, let $[a] = \{1,2,\ldots,a\}$. 
We denote the expectation and the variance of a random variable $X$ by $\E[X]$ and $\V[X]$, respectively. 
We denote the natural logarithm by $\log$ throughout the paper. 

We focus on frequency estimation as a task of the data collector. 
Let $n \in \nats$ be the number of users. 
For $i \in [n]$, let $u_i$ be the $i$-th user. 
Let $d \in \nats$ be the number of items. 
We express a user's input value as an index of the corresponding item, i.e., natural number from $1$ to $d$. 
Let $x_i \in [d]$ be an input value of user $u_i$. 
For $i \in [d]$, let 
$f_i \in [0,1]$ 
be the relative frequency\footnote{Following~\cite{Cao_USENIX21,Wang_PVLDB20}, we consider relative frequency. 
We can easily calculate absolute frequency and its estimate by multiplying $f_i$ and $\tf_i$, respectively, by $n$.} of the $i$-th item; i.e., 
$f_i = \frac{1}{n} \sum_{j=1}^n \mathone_{x_j = i}$, 
where $\mathone_{x_j = i}$ 
takes $1$ if $x_j = i$ and $0$ otherwise. 
Note that 
$\sum_{i=1}^d f_i = 1$. 
We denote the frequency distribution by $\bmf = (f_1, \cdots, f_d)$. 
Let $\hf_i \in \reals$ be an estimate of $f_i$. 
We denote the estimate of $\bmf$ by $\bmhf = (\hf_1, \cdots, \hf_d)$.

\subsection{DP (Differential Privacy)}
\label{sub:DP}

\noindent{\textbf{Neighboring Databases and DP.}}~~In this paper, we use $(\epsilon,\delta)$-DP~\cite{DP} as a privacy notion. 
We first introduce the notion of $\Omega$-neighboring databases~\cite{Beimel_CRYPTO08}: 

\begin{definition} [$\Omega$-neighboring databases~\cite{Beimel_CRYPTO08}] 
Let $\Omega \subsetneq [n]$ be a 
strict subset of $[n]$. 
We say that two databases $D = (x_1, \cdots, x_n) \in 
[d]^n$ and $D' = (x'_1, \cdots, x'_n) \in [d]^n$ are \emph{$\Omega$-neighboring} if they differ on a single entry whose index $i$ is \textit{not} in $\Omega$; i.e., $x_i \ne x'_i$ for some $i \in [n] \setminus \Omega$ and $x_j = x'_j$ for any $j \in [n] \setminus \{i\}$. 
We simply say that $D$ and $D'$ are \emph{neighboring} 
if $\Omega = \emptyset$. 
\end{definition}
The notion of $\Omega$-neighboring databases is useful when the adversary 
colludes with users $\{u_i | i \in \Omega\}$ 
and attempts to infer input values of other users. 
Balcer \textit{et al.}~\cite{Balcer_SODA21} define robust shuffle DP, which expresses DP guarantees as a function of honest users. 
However, robust shuffle DP focuses on the pure shuffle model, where the honest users execute LDP mechanisms and the shuffler only shuffles their noisy data. 
We use the notion of $\Omega$-neighboring databases in~\cite{Beimel_CRYPTO08} because it is more flexible and can be applied to both the pure and augmented shuffle models. 
Using $\Omega$-neighboring databases, we formally prove the robustness of the existing protocols and our protocols against the adversary colluding with users (Proposition~\ref{prop:shuffle_collusion} and Theorem~\ref{thm:generalization}, respectively). 

We can define $(\epsilon,\delta)$-DP using $\Omega$-neighboring databases: 

\begin{definition} [$(\epsilon,\delta)$-DP] \label{def:DP} 
Let $\epsilon \in \nnreals$ and $\delta \in [0,1]$. 
A randomized algorithm $\calM$ with domain $[d]^n$ provides \emph{$(\epsilon,\delta)$-DP for $\Omega$-neighboring databases} if for any $\Omega$-neighboring databases $D,D' \in [d]^n$ 
and any $S \subseteq \mathrm{Range}(\calM)$, 
\begin{align}
\Pr[\calM(D) \in S] \leq e^\epsilon \Pr[\calM(D') \in S] + \delta.
\label{eq:DP_inequality}
\end{align}
We simply say that $\calM$ provides \emph{$(\epsilon,\delta)$-DP} if $\Omega = \emptyset$. 
\end{definition}

The parameter $\epsilon$ is called the privacy budget. 
Both $\epsilon$ and $\delta$ should be small. 
For example, $\epsilon \geq 5$ is considered unsuitable in most use cases~\cite{DP_Li}. 
$\delta$ should be much smaller than $\frac{1}{n}$ to 
rule out 
the release-one-at-random mechanism~\cite{Cummings_arXiv23}. 

Note that Definition~\ref{def:DP} considers bounded DP~\cite{Kifer_SIGMOD11}, where $D'$ is obtained from $D$ by changing one record. 
The existing shuffle protocols cannot provide unbounded DP~\cite{Kifer_SIGMOD11}, where $D'$ is obtained from $D$ by adding or removing one record, as the shuffler knows the number $n$ of users. 
The same applies to our proposed protocols.

\smallskip{}
\noindent{\textbf{LDP.}}~~LDP~\cite{Kasiviswanathan_FOCS08} is a special case of DP where $n=1$. 
In LDP, a user adds 
local noise to her input value using an obfuscation mechanism:\begin{definition} 
[$(\epsilon,\delta)$-LDP] 
\label{def:LDP} 
Let $\epsilon \in \nnreals$. 
An obfuscation mechanism 
$\calR$ with domain $[d]$ provides 
\emph{$(\epsilon,\delta)$-LDP} 
if for any input values $x,x' \in [d]$ and any $S \subseteq \mathrm{Range}(\calR)$, 
\begin{align*}
\Pr[\calR(x) \in S] \leq e^\epsilon \Pr[\calR(x') \in S] + \delta.
\end{align*}
We simply say that $\calR$ provides \emph{$\epsilon$-LDP} if it provides $(\epsilon,0)$-LDP. 
\end{definition}

Examples of 
obfuscation mechanisms 
$\calR$ providing 
LDP 
include 
the GRR~\cite{Kairouz_ICML16,Wang_USENIX17}, OUE~\cite{Wang_USENIX17}, OLH~\cite{Wang_USENIX17}, 
RAPPOR~\cite{Erlingsson_CCS14}, and the mechanism in~\cite{Cheu_SP22}. 

\subsection{Pure Shuffle DP Model}
\label{sub:existing_shuffle}

\noindent{\textbf{Assumptions and Privacy Amplification.}}~~Below, we explain the pure shuffle protocols. 
Assume that the data collector has a secret key and publishes the corresponding public key. 
Each user $u_i$ perturbs her input value $x_i$ using an obfuscation mechanism $\calR$ 
common to all users. 
Then, user $u_i$ encrypts the noisy value $\calR(x_i)$ using the public key and sends it to the shuffler. 
The shuffler randomly shuffles the noisy values $\calR(x_1),\ldots,\calR(x_n)$ and sends them to the data collector. 
Finally, the data collector decrypts the shuffled noisy values using the secret key. 

Assume 
that the data collector does not collude with the shuffler and does not have access to any non-shuffled noisy value 
$\calR(x_i)$. 
Then, the data collector cannot link the shuffled noisy values to the users. 
In addition, the shuffler cannot access the noisy values $\calR(x_i)$, as she cannot decrypt them. 
Under this assumption, the shuffled values are protected with $(\epsilon,\delta)$-DP, where $\epsilon$ can be expressed as a function of $n$ and $\delta$, i.e., 
\begin{align}
\epsilon = g(n, \delta). 
\label{eq:shuffle_epsilon_f}
\end{align}
The function $g$ is \textit{monotonically decreasing} with respect to $n$ and $\delta$. 
The definition of $g$ depends on the protocol. 
For example, 
if $\calR$ provides $\epsilon_L$-LDP, 
the following guarantees hold: 

\begin{theorem} [Privacy amplification 
by shuffling~\cite{Feldman_FOCS21}] \label{thm:shuffle}
Let 
$\epsilon_L \in \nnreals$. 
Let $D = (x_1, \cdots, x_n) \in [d]^n$. 
Let $\calR: [d] \rightarrow \calY$ be an obfuscation mechanism. 
Let $\calM_S: [d]^n \rightarrow \calY^n$ be a pure shuffle algorithm that given a dataset $D$, % computes $y_i = \calR(x_i)$ for $i \in [n]$, samples a uniform random permutation $\pi$ over $[n]$, and outputs shuffled values $\calM_S(D) = (y_{\pi(1)}, \ldots, y_{\pi(n)})$. 
outputs shuffled values $\calM_S(D) = (\calR(x_{\pi(1)}), \ldots, \calR(x_{\pi(n)}))$, where $\pi$ is a uniform random permutation over $[n]$. 
If $\calR$ provides $\epsilon_L$-LDP, then 
for any $\delta \in [0,1]$, 
$\calM_S$ provides $(\epsilon, \delta)$-DP with 
$\epsilon = g(n, \delta)$, where 
\begin{align}
g(n, 
\delta) = 
\textstyle{
\log \left( 1 + \frac{e^{\epsilon_L}-1}{e^{\epsilon_L}+1} \left( \frac{8\sqrt{e^{\epsilon_L} \log(4/\delta)}}{\sqrt{n}} + \frac{8 e^{\epsilon_L}}{n} \right) \right) 
}
\label{eq:shuffle_epsilon}
\end{align}
if $\epsilon_L \leq \log (\frac{n}{16 \log (2/\delta)})$ and $g(n,\delta) = \epsilon_L$ otherwise. 
\end{theorem}
It follows from (\ref{eq:shuffle_epsilon}) that $\epsilon \ll \epsilon_L$ for a large value of $n$. 
Feldman \textit{et al.}~\cite{Feldman_FOCS21} also propose a method for numerically computing a tighter upper bound than the closed-form upper bound in Theorem~\ref{thm:shuffle}. 
We use the numerical upper bound in our experiments. 
Moreover, Feldman \textit{et al.}~\cite{Feldman_FOCS21} propose a closed-form upper bound specific to the GRR~\cite{Kairouz_ICML16,Wang_USENIX17}, which is tighter than the general bound in Theorem~\ref{thm:shuffle} and the numerical bound. 
We use this tighter bound for the GRR. 
We use these bounds because they are tighter than other bounds, such as~\cite{Erlingsson_SODA19,Balle_CRYPTO19,Cheu_EUROCRYPT19} (and~\cite{Girgis_CCS21} when 
the mechanism $\calR$ is applied to each input value $x_i$ once). 

Note that multi-message protocols \cite{Balcer_ITC20,Cheu_SP22,Luo_CCS22} can provide privacy guarantees different from Theorem~\ref{thm:shuffle} 
($\calR$ does not even provide LDP in \cite{Balcer_ITC20,Luo_CCS22}). 
Thus, the function $g$ is different for these protocols. 
See Appendix~\ref{sec:existing_mechanisms} for details. 

\subsection{Data Poisoning and Collusion with Users}
\label{sub:poisoning}

\noindent{\textbf{Data Poisoning.}}~~For data poisoning attacks, we consider the same threat model as~\cite{Cao_USENIX21}. 
We assume that $n$ users are genuine and that the adversary injects $n' \in \nats$ fake users. 
Thus, there are $n+n'$ users in total after data poisoning. 
The fake users can send arbitrary messages to the shuffler. 
This is called a general attack model in~\cite{Cheu_SP21}. 
The adversary's goal is to promote some target items (e.g., products of a specific company). 
To achieve this goal, the adversary attempts to increase the estimated frequencies of the target items. 

Formally, let 
$\calT \subseteq [d]$ be the set of target items. 
Let $\hf'_i \in \reals$ be an estimate of the relative frequency after data poisoning, and $\Delta \hf_i = \hf'_i - \hf_i$ be the frequency gain for the $i$-th item. 
Let $\bmy' = (y'_1, \cdots, y'_{n'})$ be messages sent from $n'$ fake users to the shuffler. 
Let $G(\bmy')$ be the adversary's overall gain given by the sum of the expected frequency gains over all target items, i.e.,
$G(\bmy') = \sum_{i \in \calT} \E[\Delta \hf_i]$. 
Note that $G(\bmy')$ depends on $\bmy'$, as the estimate $\hf'_i$ depends on the value of $\bmy'$. 

For an attack algorithm, we consider the MGA (Maximal Gain Attack)~\cite{Cao_USENIX21} because it is optimal. 
The MGA determines $\bmy'$ so that it maximizes the overall gain. 
That is, it solves the following optimization problem: 
\begin{align*}
\max_{\bmy'} G(\bmy').
\end{align*}
We denote the overall gain of the MGA by $\GMGA$ ($=\max_{\bmy'} G(\bmy')$).

\smallskip{}
\noindent{\textbf{Collusion with Users.}}~~Section~\ref{sub:existing_shuffle} 
assumes that the data collector does not have access to non-shuffled noisy values $\calR(x_i)$. 
However, this assumption may not hold in practice, as pointed out in~\cite{Wang_PVLDB20}. 
Specifically, when the data collector colludes with some users (or compromises their user accounts), she may obtain their 
noisy values 
$\calR(x_i)$. 
For example, assume that the data collector obtains 
noisy values 
$\calR(x_i)$ of $n-1$ other users (except for the victim). 
In this case, the data collector can easily link the remaining 
noisy value 
to the victim. 
Thus, no privacy amplification is obtained in this case. 

More generally, 
assume that the data collector 
obtains shuffled values $\calM_S(D)$ and 
colludes with users $\{u_i | i \in \Omega \}$, where $\Omega \subsetneq [n]$. 
Their 
noisy values 
can be expressed as 
\colorB{a tuple $(\calR(x_i))_{i \in \Omega}$}. 
Then, the privacy against the above data collector can be quantified using an algorithm $\calM_S^*$ that, given a dataset $D$, outputs 
$\calM_S^*(D) = (\calM_S(D), \colorB{(\calR(x_i))_{i \in \Omega}})$:

\begin{proposition}
\label{prop:shuffle_collusion}
Let $\Omega \subsetneq [n]$. 
Let $D = (x_1, \cdots, x_n) \in [d]^n$. 
Let $\calR: [d] \rightarrow \calY$ be an obfuscation mechanism. 
Let $\calM_S: [d]^n \rightarrow \calY^n$ be a pure shuffle algorithm (see Theorem~\ref{thm:shuffle}). 
Let $\calM_S^*$ be an algorithm that, given a dataset $D$, outputs 
$\calM_S^*(D) = (\calM_S(D), \colorB{(\calR(x_i))_{i \in \Omega}})$. 
If $\calM_S$ provides $(\epsilon, \delta)$-DP with $\epsilon = g(n, \delta)$, then 
$\calM_S^*$ provides $(\epsilon^*, \delta)$-DP for $\Omega$-neighboring databases, 
where 
\begin{align}
\epsilon^* = g(n-|\Omega|, 
\delta).
\label{eq:shuffle_collude_epsilon}
\end{align}
\end{proposition}

The proof is given in Appendix~\ref{sub:prop_shuffle_collusion}. 
Proposition~\ref{prop:shuffle_collusion} states that 
the privacy budget $\epsilon$ is increased from (\ref{eq:shuffle_epsilon_f}) to (\ref{eq:shuffle_collude_epsilon}) after colluding with $|\Omega|$ users. 
For example, 
when $n=6\times 10^5$, $|\Omega|=6\times 10^4$, and $\delta=10^{-12}$, $\epsilon$ in the pure shuffle protocols can be increased from $1$ to $7.2$ (see Appendix~\ref{sec:additional_exp} for details).

\subsection{Utility and Communication Cost}
\label{sub:utility_communication}
\noindent{\textbf{Utility.}}~~Following~\cite{Kairouz_ICML16,Wang_USENIX17,Wang_PVLDB20}, we use the expected $l_2$ loss (i.e., squared error) as a utility metric in our theoretical analysis. 
Specifically, the expected $l_2$ loss of the estimate $\bmhf$ can be expressed as $\E\left[ \sum_{i=1}^d (\hf_i - f_i)^2 \right]$, where the expectation is over the randomness of the estimator. 
By the bias-variance decomposition~\cite{mlpp}, the expected $l_2$ loss is equal to the sum of the squared bias $(\E[\hf_i] - f_i)^2$ and the variance $\sum_{i=1}^d \V[\hf_i]$. 
If the estimate $\bmhf$ is unbiased, then the expected $l_2$ loss is equal to the variance. 

In our experiments, 
we use the MSE (Mean Squared Error) as a utility metric. The MSE is the sample mean of the squared error $\sum_{i=1}^d (\hf_i - f_i)^2$ over multiple realization of $\bmhf$. 

\smallskip{}
\noindent{\textbf{Communication Cost.}}~~For a communication cost, we evaluate an expected value of the total number $C_{tot}\in\nnreals$ of 
bits 
sent from one party 
to another. 
Specifically, let $C_{U-S}, 
C_{S-D} \in \nnreals$ be the expected number of 
bits 
sent from users to the shuffler and from the shuffler to the data collector, 
respectively. 
Then, $C_{tot}$ is written as 
\begin{align*}
C_{tot} = C_{U-S} + 
C_{S-D} \hspace{4mm} \text{(bits)}. 
\end{align*}
For example, 
if 
the 2048-bit RSA is used to encrypt 
each noisy 
value of size 2048 bits or less in the pure shuffle model, 
then 
$C_{U-S} = C_{S-D} = 2048n$ 
and $C_{tot} = 4096n$. 

\section{Proposed Protocols}
\label{sec:shuffle}

In this work, we propose local-noise-free protocols in the augmented shuffle model. 
We first explain the motivation and overview of our 
protocols 
in Section~\ref{sub:motivation}. 
Then, we introduce 
a generalized framework for our local-noise-free protocols and theoretically analyze our framework 
in Sections~\ref{sub:proposed_model} and \ref{sub:theoretical_framework}, respectively. 
We propose 
three protocols in our framework in Sections~\ref{sub:SBin}, \ref{sub:SAGeo}, and \ref{sub:SOGeo}. 
Finally, we compare our protocols with existing shuffle protocols in Section~\ref{sub:comparison}. 
The proofs of all statements in this section are given in Appendices~\ref{sub:lem_equivalence_add_noise} 
to \ref{sub:thm_SAGeo_squared_error}. 

\subsection{Motivation and Overview}
\label{sub:motivation}

\noindent{\textbf{Motivation.}}~~A major drawback of the existing shuffle protocols is that they are vulnerable to local data poisoning attacks in Section~\ref{sub:poisoning}. 
As explained in Section~\ref{sec:intro}, the existing protocols require users to add noise. 
Consequently, they are vulnerable to data poisoning, especially when $\epsilon$ is small, as fake users do not need to add noise.

Another major drawback of the existing shuffle protocols is that they are vulnerable to collusion attacks by the data collector and users. 
As shown in Proposition~\ref{prop:shuffle_collusion}, 
the data collector who obtains 
noisy values 
of $|\Omega|$ users can reduce the number of shuffled values from $n$ to $n - |\Omega|$, thereby increasing the privacy budget $\epsilon$. 
In practice, it is difficult to know an actual value of $\epsilon$, as it is difficult to know the number $|\Omega|$ of users colluding with the data collector. 

Both of the two issues explained above are inevitable when users add noise to their input data. 
Thus, one might think that these issues could be addressed by not allowing users to add noise. 
However, we note that users \textit{must} add noise to their input data in the pure shuffle model where the shuffler performs only random shuffling. 
For example, 
suppose that user $u_i$'s input value $x_i$ is an outlier (e.g., user $u_i$'s age is $x_i = 115$) and that all input values are randomly shuffled without adding noise. 
In this case, the data collector can easily link $x_i$ to user $u_i$ (i.e., no privacy amplification), as it is unique. 
This example shows that each user must add noise to their data when the shuffler performs only shuffling. 

Therefore, to address the above two issues, we focus on the augmented shuffle model and propose local-noise-free protocols in this model.

\smallskip{}
\noindent{\textbf{Overview.}}~~Below, we explain a high-level overview of our local-noise-free protocol, which is shown in Figure~\ref{fig:proposed_model}. 
Our protocol does not require each user $u_i$ to perturb her input value $x_i$. 
Since random shuffling alone cannot provide privacy amplification in this case, 
our protocol introduces two additional operations on the shuffler side: random sampling and dummy data addition. 

\begin{figure}[t]
  \centering
  \includegraphics[width=0.99\linewidth]{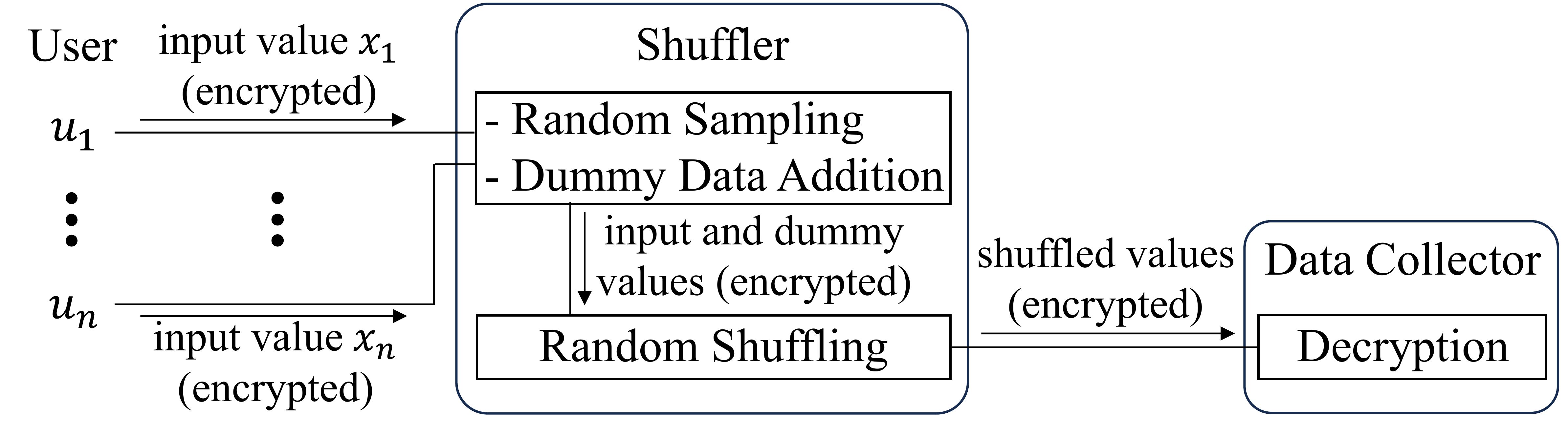}
  \vspace{-7mm}
  \caption{Overview of our local-noise-free protocol.}
  % \caption{Shuffle DP model without local noise.}
  \label{fig:proposed_model}
\end{figure}

Specifically, our model works as follows. 
First, the data collector has a secret key and publishes a public key. 
Each user $u_i$ encrypts her input value $x_i$ using the public key and sends it to the shuffler. 
Note that user $u_i$ does not add any noise to $x_i$. 
Then, the shuffler performs 
\colorB{the following 
three 
operations in this order}: 
\begin{itemize}
\item \textbf{Random Sampling.} For each encrypted 
value, the shuffler randomly selects it with some probability (and discards it with the remaining probability).
\item \textbf{Dummy Data Addition.} The shuffler generates dummy values and encrypts them using the public key. 
Then, the shuffler adds them to the encrypted input values. 
\item \colorB{\textbf{Random Shuffling.} The shuffler randomly shuffles the encrypted 
values (not discarded).} 
\end{itemize}
\colorB{After the shuffler sends the shuffled values to the data collector,} 
the data collector decrypts them using the secret key to obtain the shuffled values. \colorB{Note that random sampling is not strictly necessary to provide DP. 
We introduce random sampling for two reasons. 
First, random sampling amplifies privacy; i.e., it reduces $\epsilon$ and $\delta$. 
Second, it reduces the communication cost, as shown in our experiments.} 

\colorB{Our key insight is that random sampling and adding dummies, followed by shuffling, are equivalent to adding discrete noise to each bin of input data's histogram. 
Specifically, random sampling reduces the value of each bin in the histogram. 
Adding dummies increases the value of each bin. 
Random shuffling makes data sent to the data collector equivalent to their histogram~\cite{Balle_CRYPTO19}.
Thus, these three operations are equivalent to adding discrete noise to each bin of the histogram. 
We formalize this claim in Section~\ref{sub:theoretical_framework}.}

The shuffler can perform these 
three 
operations 
without decrypting input values sent by users. 
Since the shuffler does not have a secret key, she cannot decrypt the input values. 
Nonetheless, the shuffler can 
\colorB{generate dummy values and encrypt them}, 
as she has the public key.

As with the existing shuffle protocols, our protocols provide DP even if either the data collector or the shuffler is corrupted. 
Thus, our protocols offer a lower risk of data leakage than central DP. 

In the rest of Section~\ref{sec:shuffle}, we omit the encryption and decryption processes, as they are clear from the context. 

\setlength{\algomargin}{5mm}
\begin{algorithm}[t]
  \SetAlgoLined
  \KwData{Input values $(x_1, \cdots, x_n) \in [d]^n$, 
  \colorB{dummy-count} distribution $\calD$ over $\nnints$ (mean: $\mu$, variance: $\sigma^2$), sampling probability $\beta \in [0,1]$.}
  \KwResult{Shuffled values $(\tx_{\pi(1)}, \cdots, \tx_{\pi(\tn^*)})$.}
  % \KwResult{Estimate $\bmhf = (\hf_1, \cdots, \hf_d)$.}
  % \tcc{Send input values}
  % \ForEach{$i \in [n]$}{
  %   [$u_i$] Send $x_i$ to the shuffler\;
  % }
  \tcc{Random sampling 
  % Sample each input value 
  % with probability $\beta$ 
  ($\tx_i$: $i$-th sampled value, $\tn$: \#sampled values)}
  % [s] 
  $(\tx_1, \cdots, \tx_{\tn}) \leftarrow$\texttt{Sample}$((x_1, \cdots, x_n), \beta)$\;
  \tcc{For each item, sample $z_i$ from $\calD$ and add $z_i$ dummy values.}
  \ForEach{$i \in [d]$}{
    % [s] $z_i \leftarrow \calD$\;
    % [s] 
    $z_i \sim \calD$\;
    \ForEach{$j \in [z_i]$}{
      % [s] 
      $\tx_{\tn+(\sum_{k=1}^{i-1} z_k)+j} \leftarrow i$\;
      % [s] $x_{n+(\sum_{k=1}^{i-1} z_k)+j} \leftarrow i$\;
    }
  }
  % [s] 
  $\tn^* \leftarrow \tn + \sum_{i=1}^{d} z_i$\;
  % \tcc{Sample each value with probability $\beta$}
  % [s] $(\tx_1, \cdots, \tx_{\tn}) \leftarrow$\texttt{Sample}$((x_1, \cdots, x_{n+(\sum_{i=1}^{d} z_i)}), \beta)$\;
  \tcc{Shuffle input \& dummy values}
  % [s] 
  Sample a random permutation $\pi$ over $[\tn^*]$\;
  % [s] Send $(\tx_{\pi(1)}, \cdots, \tx_{\pi(\tn^*)})$ to the data collector\;
  % [s] 
  \KwRet{$(\tx_{\pi(1)}, \cdots, \tx_{\pi(\tn^*)})$}
  % \tcc{Compute an unbiased estimate}
  % [d] $(\tih_1, \cdots, \tih_d) \allowbreak ~~~ \leftarrow$\texttt{AbsoluteFrequency}$(\tx_{\pi(1)}, \cdots, \tx_{\pi(\tn^*)})$\;
  % \ForEach{$i \in [d]$}{
  %   [d] $\hf_i \leftarrow \frac{1}{n\beta}(\tih_i-\mu)$\;
  % }
  % [d] \KwRet{$\bmhf = (\hf_1, \cdots, \hf_d)$}
  \caption{\colorB{Augmented shuffler part $\calG_{\calD,\beta}$ of our generalized local-noise-free protocol $\calS_{\calD,\beta}$.} 
  % [$u_i$] and [s] represent that the process is run by user $u_i$ and the shuffler, respectively. 
  % Generalized local-noise-free protocol $\calS_{\calD,\beta}$. 
  % [$u_i$], [s], and [d] represent that the process is run by user $u_i$, the shuffler, and the data collector, respectively. 
  }\label{alg:framework}
\end{algorithm}

\begin{algorithm}[t]
  \SetAlgoLined
  \KwData{Shuffled values $(\tx_{\pi(1)}, \cdots, \tx_{\pi(\tn^*)})$, 
  \colorB{dummy-count} distribution $\calD$ over $\nnints$ (mean: $\mu$, variance: $\sigma^2$), sampling probability $\beta \in [0,1]$.}
  \KwResult{Estimate $\bmhf = (\hf_1, \cdots, \hf_d)$.}
  \tcc{Compute an unbiased estimate}
  % [d] 
  $\colorB{\bmth = (\tih_1, \cdots, \tih_d)} 
  % (\tih_1, \cdots, \tih_d) 
  \allowbreak ~~~ \leftarrow$\texttt{AbsoluteFrequency}$(\tx_{\pi(1)}, \cdots, \tx_{\pi(\tn^*)})$\;
  \ForEach{$i \in [d]$}{
    % [d] $\hf_i \leftarrow \frac{1}{n\beta}(\tih_i-\mu)$\;
    % [d] 
    $\hf_i \leftarrow \frac{1}{n\beta}(\tih_i-\mu)$\;
  }
  % [d] 
  \KwRet{$\bmhf = (\hf_1, \cdots, \hf_d)$}
  \caption{\colorB{Analyzer part $\calA_{\calD,\beta}$ of our generalized local-noise-free protocol $\calS_{\calD,\beta}$.} 
  % [d] represents that the process is run by the data collector. 
  }\label{alg:framework_analyzer}
\end{algorithm}

\subsection{Framework for Local-Noise-Free Protocols}
\label{sub:proposed_model}

\colorB{Algorithms~\ref{alg:framework} and \ref{alg:framework_analyzer} show 
an algorithmic description of our framework for local-noise-free protocols. 
Our generalized protocol, denoted by $\calS_{\calD,\beta}$, consists of the augmented shuffler part $\calG_{\calD,\beta}$ and the analyzer part $\calA_{\calD,\beta}$, i.e., $\calS_{\calD,\beta} = (\calG_{\calD,\beta}, \calA_{\calD,\beta})$.} 
$\calS_{\calD,\beta}$ 
has two parameters: 
a \colorB{dummy-count} distribution $\calD$ 
and a sampling probability $\beta \in [0,1]$. 
The \colorB{dummy-count} distribution $\calD$ is defined over $\nnints$ and has mean $\mu \in \nnreals$ and variance $\sigma^2 \in \nnreals$. 
In our framework, the shuffler randomly selects each input value with probability $\beta$. 
Then, for each item, the shuffler samples $z_i$ from $\calD$ and add $z_i$ dummy values ($i\in[d]$). 

Specifically, our framework works as follows. 
\colorB{After receiving each user's input value $x_i$ 
($i \in [n]$), 
the 
shuffler samples each input value with probability $\beta$ 
(line 1). 
For $i\in[d]$, let $z_i \in \nnints$ be a random variable representing the number of dummy values for item $i$. 
For each item, the shuffler generates the number $z_i$ 
of dummy values from $\calD$ ($z_i\sim\calD$) and adds $z_i$ dummy values 
(lines 2-7). 
The shuffler randomly shuffles all of the input and dummy values (i.e., $\tilde{n}+\sum_{i=1}^d z_i$ values in total, where $\tilde{n}$ represents the number of input values after sampling) and sends them to the data collector 
(lines 8-10). 
Then, 
the data collector calls the  function $\texttt{AbsoluteFrequency}$, which calculates the absolute frequency 
$\tih_i \in \nnints$ 
($i \in [d]$) of the $i$-th item from the shuffled values (Algorithm~\ref{alg:framework_analyzer}, line 1). 
Let $\bmth = (\tih_1, \cdots, \tih_d)$ be a histogram calculated by the data collector. 
The data collector calculates the estimate $\bmhf = (\hf_1, \cdots, \hf_d)$ from $\bmth$ 
as 
\begin{align}
\textstyle{\hf_i = \frac{1}{n\beta}(\tih_i - \mu)}
\label{eq:generalized_protocol_hf_i}
\end{align}
(lines 2-5).  
As shown later, 
$\bmhf$ is an unbiased estimate of $\bmf$.} 

\smallskip{}
\noindent{\colorB{\textbf{Toy Example.}}}~~\colorB{To facilitate understanding, we give a toy example of our generalized protocol. 
Assume that users' input values are $(x_1, x_2, x_3, x_4, x_5) = (1, 2, 1, 3, 2)$ ($n=5$, $d=3$) and that a binomial distribution $B(3, \frac{1}{2})$ with parameters $3$ and $\frac{1}{2}$ is used as a dummy-count distribution $\calD$. 
After receiving input values, 
the shuffler samples 
each input value with probability $\beta$. 
For example, assume that the sampled values are $(x_1, x_3, x_4) = (1, 1, 3)$. 
Then, the shuffler adds 
$z_1, z_2, z_3 \sim B(3, \frac{1}{2})$ dummies for items $1$, $2$, and $3$, respectively. For example, assume that $(z_1, z_2, z_3) = (2, 1, 1)$. 
Finally, the shuffler shuffles 
input and dummy values. 
In the above examples, the input values are $(1, 1, 3)$, and the dummies are $(1, 1, 2, 3)$. Thus, the shuffled values are, e.g., $(3, 1, 2, 1, 3, 1, 1)$. 
The shuffler sends them to the data collector. 
Finally, the data collector calculates a histogram as $\bmth = (\tih_1, \tih_2, \tih_3) = (4,1,2)$ and $\bmhf$ by (\ref{eq:generalized_protocol_hf_i}), where $\mu=\frac{3}{2}$.}

\smallskip{}
\noindent{\textbf{Remark.}}~~Our 
framework performs dummy data addition \textit{after} random sampling. 
We could consider a protocol that performs dummy data addition \textit{before} random sampling. 
However, the latter ``dummy-then-sampling'' protocol can be expressed using the former ``sampling-then-dummy'' protocol. 
Specifically, 
let $\calD^\beta$ be a \colorB{dummy-count} distribution that generates 
\colorB{the number of dummy values} by sampling $z_i\sim\calD$ dummy values and selecting each dummy value with probability $\beta$. 
Then, the ``dummy-then-sampling'' protocol 
with parameters $\calD$ and $\beta$ 
is equivalent to the ``sampling-then-dummy'' protocol 
with parameters $\calD^\beta$ and $\beta$ 
in that both of them output the same estimate $\bmhf$. 
The ``sampling-then-dummy'' protocol is more efficient because it does not need to perform random sampling for dummy values. 
Therefore, 
we adopt 
the ``sampling-then-dummy'' approach.

\subsection{Theoretical Properties of Our Framework}
\label{sub:theoretical_framework}
Below, we analyze the privacy, robustness, utility, and communication cost of our framework. 

\smallskip{}
\noindent{\textbf{Privacy and Robustness.}}~~In our framework, the \colorB{dummy-count} distribution $\calD$ plays a key role in providing DP. 
We reduce DP of our generalized local-noise-free protocol $\calS_{\calD,\beta}$ to that of a simpler mechanism with binary input, which we call a \textit{binary input mechanism}. 
Specifically, define a binary input mechanism $\calM_{\calD,\beta}$ with domain $\{0,1\}$ 
as follows: 
\begin{align*}
    \calM_{\calD,\beta}(x)=ax+z,~x\in\{0,1\}, 
\end{align*}
where $a\sim\mathrm{Ber}(\beta)$, $z\sim\calD$, and $\mathrm{Ber}(\beta)$ is the Bernoulli distribution with parameter $\beta$. 
\colorB{We prove that if the binary input mechanism $\calM_{\calD,\beta}$ provides DP, then $\calS_{\calD,\beta}$ also provides DP and is also robust to collusion with users.} 

\colorB{We first formally state our key insight that 
random sampling and adding dummies, followed by shuffling, 
are equivalent to adding discrete noise to 
the histogram:} 

\begin{lemma}\label{lem:equivalence_add_noise}
\colorB{For $i\in[d]$, let $\bmx_i\in\{0,1\}^d$ be a binary vector whose entries are $1$ only at position $x_i$. 
Let $a_1,\ldots,a_n \sim \mathrm{Ber}(\beta)$ be independent Bernoulli samples used for sampling $x_1,\ldots,x_n$, respectively; i.e., if $a_i = 1$, then $x_i$ is selected by random sampling; otherwise, $x_i$ is not selected.
Let $\bmz = (z_1,\ldots,z_d)$. 
Then, the histogram $\bmth$ (Algorithm~\ref{alg:framework_analyzer}, line 1) calculated by the data collector is: $\bmth = (\sum_{i=1}^n a_i \bmx_i) + \bmz$.}
\end{lemma}

\begin{lemma}\label{lem:equivalence_shuffle}
\colorB{Let $\calG_{\calD,\beta}'$ be an algorithm that, given a dataset $D$, outputs the histogram $\bmth$. 
Then, the privacy of $\calG_{\calD,\beta}$ (Algorithm~\ref{alg:framework}) is equivalent to that of $\calG_{\calD,\beta}'$; i.e., $\calG_{\calD,\beta}$ provides $(\epsilon,\delta)$-DP if and only if $\calG_{\calD,\beta}'$ provides $(\epsilon,\delta)$-DP.} 
\end{lemma}

\colorB{
Lemmas~\ref{lem:equivalence_add_noise} and \ref{lem:equivalence_shuffle} state that shuffled data sent to the data collector are equivalent to the histogram $\bmth = (\sum_{i=1}^n a_i \bmx_i) + \bmz$, where discrete noise is added by random sampling and dummy data addition. 
Based on these lemmas, 
we prove the privacy and robustness of 
$\calG_{\calD,\beta}$ (Algorithm~\ref{alg:framework}):} 

\begin{theorem}\label{thm:generalization}
    Let $\Omega \subsetneq [n]$. 
    % Let $\bmy^* = (y_1^*, \ldots, y_n^*) \in \{1, \ldots, d, \bot\}^n$ be an $n$-dim vector whose $i$-th element is $y_i^*=x_i$ if $i \in \Omega$ and $y_i^*=\bot$ otherwise. 
    Let 
    % $\calS_{\calD,\beta}^*$ 
    \colorB{$\calG_{\calD,\beta}^*$} 
    be an algorithm that, given a dataset $D$, outputs 
    % $\calS_{\calD,\beta}^*(D) = (\calS_{\calD,\beta}(D), \bmy^*)$. 
    % $\calS_{\calD,\beta}^*(D) = (\calS_{\calD,\beta}(D), \colorB{(x_i)_{i \in \Omega}})$. 
    \colorB{$\calG_{\calD,\beta}^*(D) = (\calG_{\calD,\beta}(D), (x_i)_{i \in \Omega})$}. 
    % Let a probability distribution $\calD$ over $\nnints$ and a sampling probability $\beta\in[0,1]$ be such that 
    If $\calM_{\calD,\beta}$ provides 
    % $(\epsilon,\delta)$-DP, 
    $(\frac{\epsilon}{2},\frac{\delta}{2})$-DP, 
    then 
    % $\calS_{\calD,\beta}$ 
    \colorB{$\calG_{\calD,\beta}$} 
    provides 
    % $(2\epsilon,2\delta)$-DP 
    $(\epsilon,\delta)$-DP 
    and 
    % against the data collector.
    % Furthermore, 
    % $\calS_{\calD,\beta}^*$ 
    \colorB{$\calG_{\calD,\beta}^*$} 
    provides 
    % $(2\epsilon,2\delta)$-DP 
    $(\epsilon,\delta)$-DP 
    for $\Omega$-neighboring databases. 
    % if the data collector colludes with a set $\Omega$ of users, $\calS_{\calD,\beta}$ provides $(2\epsilon,2\delta)$-DP for $\Omega$-neighboring databases.
\end{theorem}
\colorB{Since the analyzer part $\calA_{\calD,\beta}$ (Algorithm~\ref{alg:framework_analyzer}) is post-processing, our entire protocol $\calS_{\calD,\beta}$ also provides $(\epsilon,\delta)$-DP and is robust to collusion with users.} 
By Theorem~\ref{thm:generalization}, 
the privacy budget of $\calS_{\calD,\beta}$ is twice that of $\calM_{\calD,\beta}$. 
This is because $\calM_{\calD,\beta}$ deals with one-dimensional data. 
Specifically, neighboring data $x,x' \in \{0,1\}$ for $\calM_{\calD,\beta}$ differ by 1 in one dimension. 
In contrast, neighboring databases $D,D' \in [d]^n$ for $\calS_{\calD,\beta}$ differ by 1 in two dimensions. 
Thus, by group privacy, the privacy budget is doubled in $\calS_{\calD,\beta}$. 

Theorem~\ref{thm:generalization} shows that $\epsilon$ and $\delta$ are \textit{not} increased even when the data collector colludes with other users. 
This robustness property follows from the fact that $\calS_{\calD,\beta}$ obfuscates input data on the shuffler side rather than the user side. 

We also show that $\calS_{\calD,\beta}$ is robust to data poisoning: 
\begin{theorem}
\label{thm:generalization_poisoning}
Let $\lambda = \frac{n'}{n+n'}$ be the fraction of fake users, and $f_T = \sum_{i \in \calT} f_i$ be the sum of the frequencies over all target items. 
Then, $\calS_{\calD,\beta}$ provides the following robustness guarantee: 
\begin{align}
    \GMGA = \lambda (1 - f_T).
    \label{eq:GMGA_generalization}
\end{align}
\end{theorem}

Note that $\GMGA$ in (\ref{eq:GMGA_generalization}) does not depend on the privacy budget $\epsilon$. 
In other words, $\calS_{\calD,\beta}$ does not suffer from the fundamental security-privacy trade-off explained in Section~\ref{sec:intro}. 
This also comes from the fact that $\calS_{\calD,\beta}$ obfuscates input data on the shuffler side rather than the user side. 
Consequently, the same amount of noise is added to the messages of both genuine and fake users. 
Thus, $\GMGA$ in $\calS_{\calD,\beta}$ does not depend on $\epsilon$. 

In Section~\ref{sub:comparison}, we also show that $\GMGA$ in (\ref{eq:GMGA_generalization}) is \textit{always} smaller than that of the existing shuffle protocols.

\smallskip{}
\noindent{\textbf{Utility.}}~~Next, we analyze the utility 
of $\calS_{\calD,\beta}$: 

\begin{theorem}
\label{thm:generalization_utility}
$\calS_{\calD,\beta}$ outputs an unbiased estimate; i.e., for any $i\in[d]$, 
$\E[\hf_i] = f_i$. 
In addition, $\calS_{\calD,\beta}$ achieves the following expected $l_2$ loss: 
\begin{align}
\textstyle{\E\left[ \sum_{i=1}^d (\hf_i - f_i)^2 \right] = \frac{1-\beta}{\beta n}+\frac{\sigma^2d}{\beta^2 n^2}.}
\label{eq:l2_loss_generalization}
\end{align}
\end{theorem}
Theorem~\ref{thm:generalization_utility} states that 
the expected $l_2$ loss of $\calS_{\calD,\beta}$ can be calculated from the sampling probability $\beta$ and the variance $\sigma^2$ of the \colorB{dummy-count} distribution $\calD$. 
A larger $\beta$ and a smaller $\sigma^2$ give a smaller expected $l_2$ loss. 
\colorB{Theorem~\ref{thm:generalization_utility} also shows that the utility of our generalized protocol depends on the choice of $\calD$. 
In fact, our \SAGeo{} carefully chooses $\calD$ to provide higher utility than our \SBin{}.}

\smallskip{}
\noindent{\textbf{Communication Cost.}}~~Finally, we analyze the communication cost of $\calS_{\calD,\beta}$. 
Let $\alpha \in \nats$ be the number of bits required to encrypt each input value (e.g., $\alpha=2048$ when the $2048$-bit RSA is used). 
The amount of communication between users and the shuffler is $n$ ciphertexts, and that between the shuffler and the data collector is $\tn+\sum_{i=1}^d z_i$ ciphertexts, where $\tn$ is the number of input values after sampling. 
Since $\E[\tn+\sum_{i=1}^d z_i] = \beta n+\mu d$, 
we have 
$C_{U-S}=\alpha n$, $C_{S-D}=\alpha(\beta n+\mu d)$, and 
\begin{align}
C_{tot}=\alpha((1+\beta)n+\mu d).
\label{eq:C_tot_generalization}
\end{align}
Thus, $C_{tot}$ can be calculated from the sampling probability $\beta$ and the mean $\mu$ of the \colorB{dummy-count} distribution $\calD$. 
Smaller $\beta$ and $\mu$ give a smaller communication cost. 
In addition, by (\ref{eq:l2_loss_generalization}) and (\ref{eq:C_tot_generalization}), the trade-off between the $l_2$ loss and communication cost can be controlled by changing $\beta$. 

\subsection{\SBin{} (Sample, Binomial Dummies, and Shuffle)}
\label{sub:SBin}

\noindent{\textbf{Protocol.}}~~A natural way to add 
differentially private non-negative discrete noise is to use a binomial mechanism~\cite{Agarwal_NeurIPS18,Dwork_EUROCRYPT06}. 
Our first protocol \SBin{}, denoted by $\calS_{\Bin,\beta}$, is based on this. 
Specifically, \SBin{} instantiates a \colorB{dummy-count} distribution $\calD$ with a binomial distribution $B(M,\frac{1}{2})$ with the number $M \in \nats$ of trials and success probability $\frac{1}{2}$. 
That is, for $k\in\nnints$, the probability mass function at $z_i = k$ is given by 
\begin{align*}
\textstyle{\Pr(z_i = k) = \frac{\binom{M}{k}}{2^M}.}
\end{align*}

\smallskip{}
\noindent{\textbf{Privacy and Robustness.}}~~Let $\calM_{\Bin,\beta}$ be 
a binary input mechanism $\calM_{\calD,\beta}$ 
instantiated with the binomial distribution $B(M,\frac{1}{2})$. 
($\epsilon,\delta$)-DP of $\calM_{\Bin,\beta}$ is immediately derived from the existing analysis of the binomial mechanism~\cite{Agarwal_NeurIPS18,Dwork_EUROCRYPT06}. 
However, the theoretical results in~\cite{Agarwal_NeurIPS18,Dwork_EUROCRYPT06} do not provide tight bounds on the number $M$ of trials to provide ($\epsilon,\delta$)-DP. 
Therefore, we prove a tighter bound 
for $\calM_{\Bin,\beta}$. 

Specifically, we prove the following result: 

\begin{theorem}
\label{thm:SBin_DP}
Let $\epsilon_0 \in [\log(\frac{2}{M}+1), \infty)$ and $\eta = \frac{e^{\epsilon_0} - 1}{e^{\epsilon_0} + 1} - \frac{2}{M(e^{\epsilon_0} + 1)}$. 
$\calM_{\Bin,\beta}$ provides $(\frac{\epsilon}{2},\frac{\delta}{2})$-DP, where
\begin{align}
\epsilon &= 2 \log (1 + \beta (e^{\epsilon_0} - 1)) \label{eq:SBin_epsilon}\\
\delta &= 4 \beta e^{-\frac{\eta^2 M}{2}}. 
\label{eq:SBin_delta}
\end{align}
\end{theorem}
It follows from Theorems~\ref{thm:generalization}, \ref{thm:generalization_poisoning}, and \ref{thm:SBin_DP} that \SBin{} $\calS_{\Bin,\beta}$ provides 
$(\epsilon,\delta)$-DP 
and is robust to data poisoning and collusion with users. 
Given the sampling probability $\beta$ and the required values of $\epsilon$ and $\delta$, the number $M$ of trials is uniquely determined by (\ref{eq:SBin_epsilon}) and (\ref{eq:SBin_delta}). 

In the proof of Theorem~\ref{thm:SBin_DP}, we use 
the multiplicative Chernoff bound for 
the binomial random variable 
with success probability $\frac{1}{2}$~\cite{probability_computing}, which is tighter than the 
one for 
a general binomial random variable. 
This brings us a tighter bound than the bounds derived from~\cite{Agarwal_NeurIPS18,Dwork_EUROCRYPT06}.

Specifically, in Appendix~\ref{sec:comparison_bound}, 
we show $(\epsilon,\delta)$-DP of $\calS_{\Bin,1}$ derived from~\cite{Dwork_EUROCRYPT06} and~\cite{Agarwal_NeurIPS18} in Theorems~\ref{thm:SBin_DP_Dwork} and \ref{thm:SBin_DP_Agarwal}, respectively. 
Figure~\ref{fig:bound_comparison} shows  
the smallest number $M$ of trials in $\calS_{\Bin,1}$ required to provide $(\epsilon,\delta)$-DP in each of Theorems~\ref{thm:SBin_DP}, \ref{thm:SBin_DP_Dwork}, and \ref{thm:SBin_DP_Agarwal} (denoted by \textsf{Our bound}, \textsf{DKMMN08}, and \textsf{ASYKM18}, respectively). 
We observe that our bound is tighter than the existing bounds in~\cite{Dwork_EUROCRYPT06,Agarwal_NeurIPS18} in all cases. 

\begin{figure}[t]
  \centering
  \includegraphics[width=0.99\linewidth]{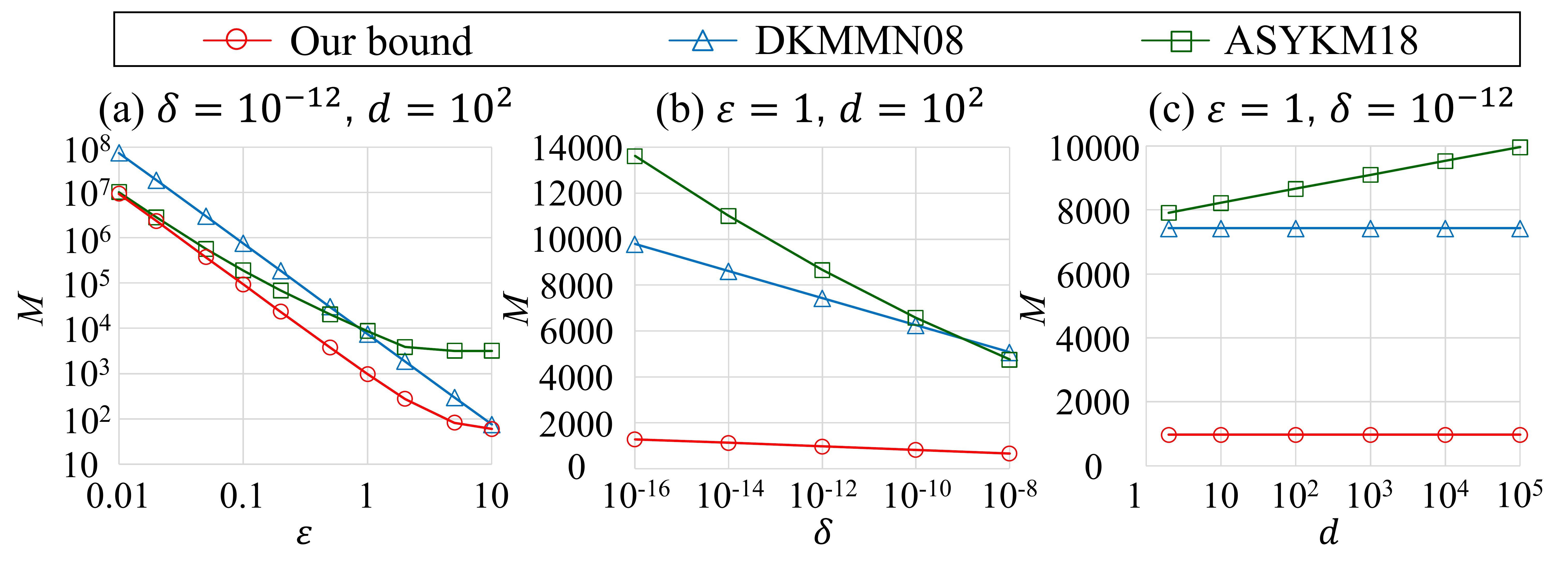}
  \vspace{-6mm}
  \caption{Three bounds on the number $M$ of trials in \SBin{} ($\beta=1$). \textsf{DKMMN08} and \textsf{ASYKM18} are the bounds in~\cite{Dwork_EUROCRYPT06} and \cite{Agarwal_NeurIPS18}, respectively. We set $\epsilon=1$, $\delta=10^{-12}$, and $d=10^2$ as default values.}
  \label{fig:bound_comparison}
\end{figure}

\smallskip{}
\noindent{\textbf{Utility.}}~~The binomial distribution $B(M,\frac{1}{2})$ has variance $\sigma^2 = \frac{M}{4}$. 
Thus, by Theorem~\ref{thm:generalization_utility}, \SBin{} achieves the following expected $l_2$ loss: 
\begin{align*}
\textstyle{\E\left[ \sum_{i=1}^d (\hf_i - f_i)^2 \right] 
= \frac{1-\beta}{\beta n} + \frac{Md}{4 \beta^2 n^2}}.
\end{align*}
When $\epsilon$ is small 
(i.e., $e^\epsilon \approx \epsilon + 1$), 
the expected $l_2$ loss can approximated, using $\epsilon$ and $\delta$, as $\frac{1-\beta}{\beta n} + \frac{8d \log(4\beta/\delta)}{\epsilon^2 n^2}$\conference{.}\arxiv{ (see Appendix~\ref{sub:loss_SBin} for more details).} 
In Sections~\ref{sub:comparison} and \ref{sec:exp}, we show that \SBin{} provides higher utility than the existing shuffle protocols.

\smallskip{}
\noindent{\textbf{Communication Cost.}}~~$B(M,\frac{1}{2})$ has mean $\frac{M}{2}$. 
Thus, 
the total communication cost $C_{tot}$ of $\SBin{}$ is given by (\ref{eq:C_tot_generalization}), where $\mu = \frac{M}{2}$. 
When $\epsilon$ is small, 
$C_{tot}$ 
can be approximated as $C_{tot} \approx \alpha ((1+\beta)n + \frac{16\beta^2 d \log(4\beta/\delta)}{\epsilon^2})$\conference{.}\arxiv{ (see Appendix~~\ref{sub:communication_SBin} for more details).}

\subsection{\SAGeo{} (Sample, Asymmetric Two-Sided Geometric Dummies, and Shuffle)}
\label{sub:SAGeo}

\noindent{\textbf{Protocol.}}~~\colorB{Our second protocol \SAGeo{}, denoted by $\calS_{\AGeo,\beta}$, significantly improves the accuracy of \SBin{}} 
by introducing a novel \colorB{dummy-count} distribution that has not been studied in the DP literature: \textit{asymmetric two-sided geometric distribution}. 

\colorB{The intuition behind this distribution is as follows. 
The geometric and binomial distributions are discrete versions of the Laplace and Gaussian distributions, respectively. 
It is well known that the Laplace mechanism provides higher utility than the Gaussian mechanism for a single statistic, such as a single frequency distribution~\cite{desfontainesblog20201115}. 
Thus, the geometric distribution would provide higher utility than the binomial distribution in our task. 
Moreover, random sampling has the effect of anonymization and can reduce the scale of the left-side of a distribution. 
This is the reason that \SAGeo{} adopts the asymmetric geometric distribution.}

Formally, 
the asymmetric two-sided geometric distribution has parameters $\nu \in \nnints$, $q_l \in (0,1)$, and $q_r \in (0,1)$. 
We denote this distribution by $\AGeo(\nu,q_l,q_r)$. 
$\AGeo(\nu,q_l,q_r)$ has support in 
$\nnints$, 
and the probability mass function at $z_i=k$ is given by 
\begin{align*}
\Pr(z_i = k) = 
\begin{cases}
\frac{1}{\kappa} q_l^{\nu - k}   &   \text{(if $k = 0, 1, \ldots, \nu-1$)}\\
\frac{1}{\kappa} q_r^{k - \nu}   &   \text{(if $k = \nu, \nu+1, \ldots $)},
\end{cases}
\end{align*}
where $\kappa$ is a normalizing constant given by 
\begin{align*}
\textstyle{\kappa = \frac{q_l(1-q_l^\nu)}{1-q_l} + \frac{1}{1 - q_r}.}
\end{align*}
$\AGeo(\nu,q_l,q_r)$ is truncated at $k = 0$ and has a mode at $k = \nu$. 
$\AGeo(\nu,q_l,q_r)$ has mean $\mu$, where 
\begin{align}
\textstyle{\mu = \frac{1}{\kappa} \left( \sum_{k=0}^{\nu-1} k q_l^{\nu-k} + \frac{q_r + (1-q_r)\nu}{(1 - q_r)^2}\right)}\conference{.}
\label{eq:SAGeo_mu_ast}
\end{align}
\arxiv{(see Appendix~\ref{sub:mu_SAGeo} for details).} 
The non-truncated asymmetric two-sided geometric distribution with mode zero is 
studied for data compression 
in~\cite{Ding_EUSIPCO12}. 
$\AGeo(\nu,q_l,q_r)$ can be regarded as a \textit{shifted} and \textit{zero-truncated} version of the distribution in~\cite{Ding_EUSIPCO12}. 

\SAGeo{} $\calS_{\AGeo,\beta}$ 
samples each input value with probability $\beta \in [1 - e^{-\epsilon/2}, 1]$ and 
instantiates a \colorB{dummy-count} distribution $\calD$ with $\AGeo(\nu,q_l,q_r)$ with 
\begin{align}
q_l = \textstyle{\frac{e^{-{\epsilon/2}} - 1 + \beta}{\beta}}, ~~
q_r = \textstyle{\frac{\beta}{e^{\epsilon/2} - 1 + \beta}}. 
\label{eq:SAGeo_q_r}
\end{align}
The parameter $\nu$ is determined so that $\delta$ in ($\epsilon,\delta$)-DP is smaller than or equal to a required value, as explained later. 

Note that we always have $q_l \leq q_r$, where the equality holds if and only if $\beta=1$ or $\epsilon=0$. 
When $\beta = 1$, 
we have $q_l = q_r = \frac{1}{e^{\epsilon/2}}$. In this case, $\AGeo(\nu,q_l,q_r)$ is equivalent to a symmetric two-sided geometric distribution~\cite{Chan_SODA19,Qiu_PVLDB20,Zhou_EUROCRYPT23}. 
As $\beta$ is reduced from $1$ to $1-e^{-\epsilon}$, the parameters $q_l$ and $q_r$ are reduced from $\frac{1}{e^{\epsilon/2}}$ to $\frac{e^{-{\epsilon/2}} - 1 + \beta}{\beta}$ and $\frac{\beta}{e^{\epsilon/2} - 1 + \beta}$, respectively. 
In particular, the left-hand curve of $\AGeo(\nu,q_l,q_r)$ becomes steeper, as shown in the left panel of Figure~\ref{fig:AGeo}. 
In other words, random sampling has the effect of anonymization and, therefore, reduces the variance of the \colorB{dummy-count} distribution required to provide ($\epsilon,\delta$)-DP, as shown in the right panel of Figure~\ref{fig:AGeo}.

\begin{figure}[t]
  \centering
  \includegraphics[width=0.95\linewidth]{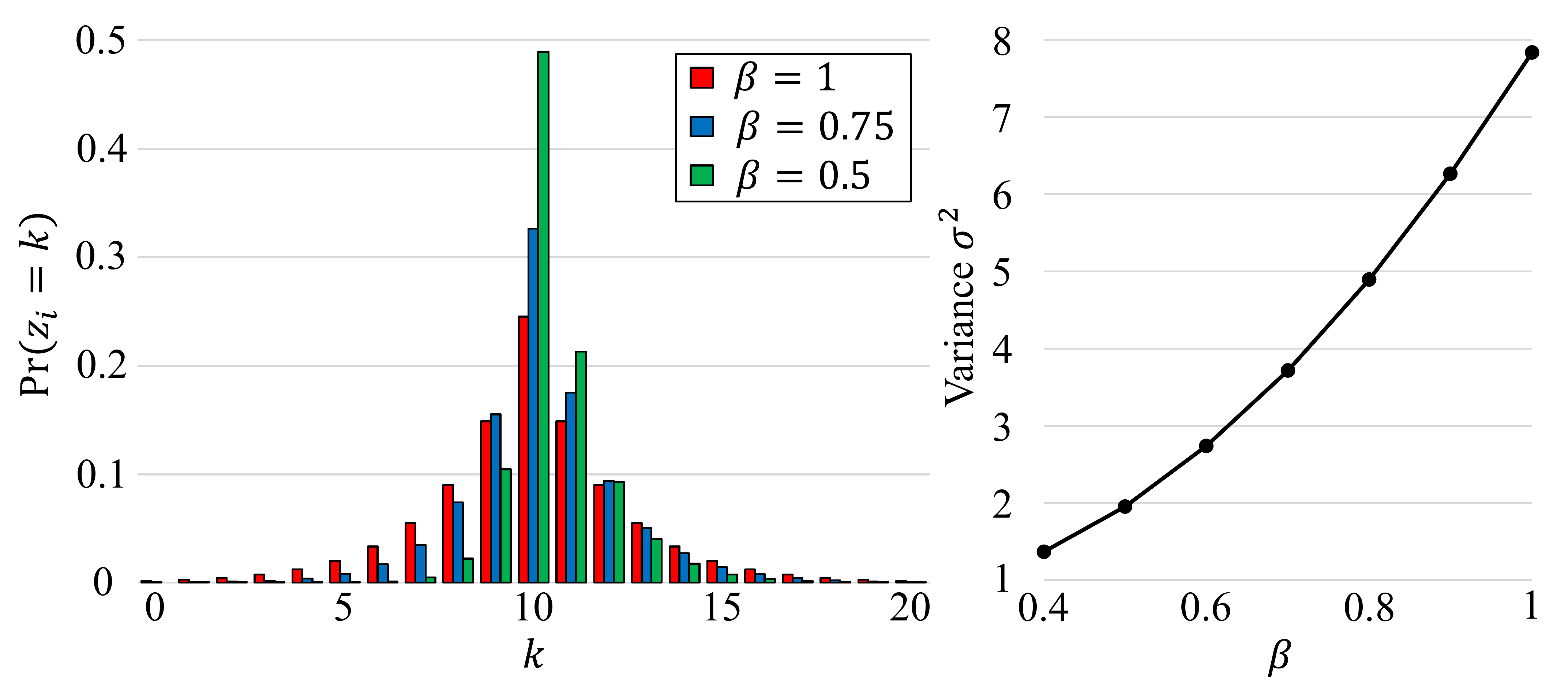}
  \vspace{-4mm}
  \caption{The asymmetric two-sided geometric distribution $\AGeo(\nu,q_l,q_r)$ and its variance $\sigma^2$ (upper bound in (\ref{eq:SAGeo_MSE})) when $\epsilon=1$ and $\nu=10$.}
  \label{fig:AGeo}
\end{figure}

\colorB{Random sampling also has the effect of reducing $\delta$ in DP at the same value of $\nu$, as shown later. 
When $\beta$ reaches $1-e^{-\epsilon/2}$, $q_l$ becomes $0$. 
In this case, $\AGeo(\nu,q_l,q_r)$ becomes a \textit{one-sided} geometric distribution, and $\delta$ becomes $0$.}
We explain this special case in Section~\ref{sub:SOGeo} in detail.

\smallskip{}
\noindent{\textbf{Privacy and Robustness.}}~~Let $\calM_{\AGeo,\beta}$ be a binary input mechanism instantiated with 
$\AGeo(\nu,q_l,q_r)$.
$\calM_{\AGeo,\beta}$ 
provides $(\epsilon,\delta)$-DP: 

\begin{theorem}
\label{thm:SAGeo_DP}
Let $\epsilon \in \nnreals$, 
$q_l = \frac{e^{-{\epsilon/2}} - 1 + \beta}{\beta}$, and $q_r = \frac{\beta}{e^{\epsilon/2} - 1 + \beta}$. 
Then, 
$\calM_{\AGeo,\beta}$ provides $(\frac{\epsilon}{2},\frac{\delta}{2})$-DP, 
where
\begin{align}
\delta = 
\begin{cases}
0   &   \text{(if $\beta = 1 - e^{-\epsilon/2}$)}\\
\frac{2}{\kappa} q_l^\nu (1 - e^{\epsilon/2} + \beta e^{\epsilon/2})    &   \text{(if $\beta > 1 - e^{-\epsilon/2}$)}.
\label{eq:SAGeo_delta}
\end{cases}
\end{align}
\end{theorem}
By Theorems~\ref{thm:generalization}, \ref{thm:generalization_poisoning}, and \ref{thm:SAGeo_DP}, \SAGeo{} $\calS_{\AGeo,\beta}$ provides 
$(\epsilon,\delta)$-DP 
and is robust to data poisoning and collusion. 
Given $\beta$, $\epsilon$, and a required value of $\delta$, the parameter $\nu$ is uniquely determined as a minimum value such that $\delta$ in (\ref{eq:SAGeo_delta}) is smaller than or equal to the required value.

\smallskip{}
\noindent{\textbf{Utility.}}~~Next, we show the 
variance of $\AGeo(\nu,q_l,q_r)$:

\begin{proposition}
\label{prop:SAGeo_squared_error}
Let $\kappa^* = \frac{q_l}{1-q_l} + \frac{1}{1 - q_r}$. The variance $\sigma^2$ of $\AGeo(\nu,q_l,q_r)$ is upper bounded as follows:
\begin{align}
\textstyle{\sigma^2 
\leq 
\frac{1}{\kappa^*}\left(\frac{q_l(1+q_l)}{(1-q_l)^3} + \frac{q_r(1+q_r)}{(1-q_r)^3}\right),}
\label{eq:SAGeo_MSE}
\end{align}
\end{proposition}
By Theorem~\ref{thm:generalization_utility} and Proposition~\ref{prop:SAGeo_squared_error}, \SAGeo{} achieves the following expected $l_2$ loss: 
\begin{align*}
\textstyle{\E\left[ \sum_{i=1}^d (\hf_i - f_i)^2 \right] 
\hspace{-0.5mm} \leq \hspace{-0.5mm}
\frac{1-\beta}{\beta n} \hspace{-0.5mm}+\hspace{-0.5mm} \frac{d}{\kappa^* \beta^2 n^2} \hspace{-0.5mm} \left(\frac{q_l(1+q_l)}{(1-q_l)^3} + \frac{q_r(1+q_r)}{(1-q_r)^3}\right).}
\end{align*}
When $\epsilon$ is close to $0$ (i.e., $e^\epsilon \approx \epsilon + 1$), the right-hand side of (\ref{eq:SAGeo_MSE}) can be approximated as $\frac{1-\beta}{\beta n} + \frac{8d}{\epsilon^2 n^2}$\conference{.}\arxiv{ (see Appendix~\ref{sub:loss_SAGeo} for details).} 

\smallskip{}
\noindent{\textbf{Communication Cost.}}~~The total communication cost $C_{tot}$ of $\SAGeo{}$ is given by (\ref{eq:C_tot_generalization}), where $\mu$ is given by (\ref{eq:SAGeo_mu_ast}). 
When $\epsilon$ is close to $0$ and $\beta > 1 - e^{-\epsilon/2}$, 
$C_{tot}$ can be upper bounded as follows: $C_{tot} \leq \alpha ((1 + \beta)n + \frac{2\beta d (\log(\epsilon / 2\delta) + 1)}{\epsilon})$\conference{.}\arxiv{ (see Appendix~\ref{sub:communication_SAGeo} for details).}

\begin{table*}[t]
\caption{Performance guarantees of various shuffle protocols 
($\lambda$ ($= \frac{n'}{n+n'}$): fraction of fake users, $f_T$ ($= \sum_{i \in \calT} f_i$): frequencies over target items, $|\calT|$: \#target items, \colorB{$\alpha$: \#bits required to encrypt each input value}). 
\BC{} and \LWY{} 
assume that $n \geq \frac{400\log(4/\delta)}{\epsilon^2}$ and $n \geq \frac{32d\log(2/\delta)}{\epsilon^2}$, respectively, as described in \cite{Balcer_ITC20,Luo_CCS22}. 
The expected $l_2$ loss, the communication cost, and the overall gain are approximate values when $\epsilon$ is close to $0$.}
\vspace{-2.8mm}
\centering
\hbox to\hsize{\hfil
\begin{tabular}{c||c|c|c|c|c}
\hline
&   Privacy &   Expected $l_2$ loss    &  
\begin{tabular}{c}
\hspace{-2mm}Communication cost\hspace{-2mm} \\ \hspace{-2mm}$C_{tot}$ ($\times \alpha$)\hspace{-2mm}
\end{tabular}
&   Overall gain $\GMGA$ &   
\begin{tabular}{c}
\hspace{-2mm}Robustness to  \hspace{-2mm} \\ \hspace{-2mm}collusion with users\hspace{-2mm}
\end{tabular}\\
\hline
\begin{tabular}{c}
\hspace{-4mm}\SBin{}\hspace{-4mm} \\ \hspace{-4mm}($\beta \in [0,1]$)\hspace{-4mm}
\end{tabular}
&   $(\epsilon,\delta)$-DP &   $\frac{1-\beta}{\beta n} + \frac{8d \log\frac{4\beta}{\delta}}{\epsilon^2 n^2}$    &   
$(1+\beta)n + \frac{16\beta^2 d\log\frac{4\beta}{\delta}}{\epsilon^2}$  &   $\lambda (1 - f_T)$ &   $\checkmark$ (Theorem~\ref{thm:generalization})\\
\hline
\begin{tabular}{c}
\hspace{-4mm}\SAGeo{}\hspace{-4mm} \\ \hspace{-3mm}($\beta \in (1 - e^{-\frac{\epsilon}{2}}, 1]$)\hspace{-3mm}
\end{tabular}
&   $(\epsilon,\delta)$-DP &   $\leq \frac{1-\beta}{\beta n} + \frac{8d}{\epsilon^2 n^2}$   &   $\leq (1 + \beta)n + \frac{2\beta d (\log\frac{\epsilon}{2\delta} + 1)}{\epsilon}$ &   $\lambda (1 - f_T)$ &   $\checkmark$ (Theorem~\ref{thm:generalization})\\
\hline
\SOGeo{}
&   $\epsilon$-DP &   $\frac{2}{\epsilon n} + \frac{8d}{\epsilon^2 n^2}$   &   $n + d$ &   $\lambda (1 - f_T)$ &   $\checkmark$ (Theorem~\ref{thm:generalization})\\
\hline
\begin{tabular}{c}
\hspace{-4mm}\GRRS{}\hspace{-4mm} \\ \hspace{-4mm}(\hspace{-0.1mm}\cite{Kairouz_ICML16,Wang_USENIX17} +~\cite{Feldman_FOCS21})\hspace{-4mm}
\end{tabular}
&   $(\epsilon,\delta)$-DP &   
$\frac{32d\log\frac{4}{\delta}}{\epsilon^2 n^2}$ 
&   $2n$ &   
\begin{tabular}{c}
$\lambda (1 - f_T)$ \\ \hspace{-4mm}$+ \frac{4 \lambda(d - |\calT|)\sqrt{2(d+1)\log\frac{4}{\delta}}}{\epsilon  d\sqrt{n}}$ \hspace{-4mm}
\end{tabular}
&   $\times$ (Proposition~\ref{prop:shuffle_collusion})\\
\hline
\begin{tabular}{c}
\hspace{-4mm}\OUES{}\hspace{-4mm} \\ \hspace{-4mm}(\hspace{-0.1mm}\cite{Wang_USENIX17} +~\cite{Feldman_FOCS21})\hspace{-4mm}
\end{tabular}
&   $(\epsilon,\delta)$-DP &   $\frac{64d\log\frac{4}{\delta}}{\epsilon^2 n^2}$ &   $2n$ &   
\begin{tabular}{c}
$\lambda (2|\calT| - f_T)$ \\ $+ \frac{8 \lambda |\calT|\sqrt{\log\frac{4}{\delta}}}{\epsilon \sqrt{n}}$
\end{tabular}
&   $\times$ (Proposition~\ref{prop:shuffle_collusion})\\
\hline
\begin{tabular}{c}
\hspace{-4mm}\OLHS{}\hspace{-4mm} \\ \hspace{-4mm}(\hspace{-0.1mm}\cite{Wang_USENIX17} +~\cite{Feldman_FOCS21})\hspace{-4mm}
\end{tabular}
&   $(\epsilon,\delta)$-DP &   $\frac{64d\log\frac{4}{\delta}}{\epsilon^2 n^2}$ &   $2n$ &   
\begin{tabular}{c}
$\lambda (2|\calT| - f_T)$ \\ $+ \frac{8 \lambda |\calT|\sqrt{\log\frac{4}{\delta}}}{\epsilon \sqrt{n}}$
\end{tabular}
&   $\times$ (Proposition~\ref{prop:shuffle_collusion})\\
\hline
\begin{tabular}{c}
\hspace{-4mm}\RAPS{}\hspace{-4mm} \\ \hspace{-4mm}(\hspace{-0.1mm}\cite{Erlingsson_CCS14} +~\cite{Feldman_FOCS21})\hspace{-4mm}
\end{tabular}
&   $(\epsilon,\delta)$-DP &   $\frac{64d\log\frac{4}{\delta}}{\epsilon^2 n^2}$ &   $2n$ &   
\begin{tabular}{c}
$\lambda (|\calT| - f_T)$ \\ $+ \frac{8 \lambda |\calT|\sqrt{\log\frac{4}{\delta}}}{\epsilon \sqrt{n}}$
\end{tabular}
&   $\times$ (Proposition~\ref{prop:shuffle_collusion})\\
\hline
\BC{}~\cite{Balcer_ITC20}   &   $(\epsilon,\delta)$-DP &   $\geq \frac{100d\log\frac{4}{\delta}}{\epsilon^2 n^2}$   &   $\geq 2n(1+\frac{d}{2})$    &   
$\lambda (1 - f_T + \frac{200|\calT|}{\epsilon^2 n}\log\frac{4}{\delta})$
&   $\times$ (Proposition~\ref{prop:shuffle_collusion})\\
\hline
\CM{}~\cite{Cheu_SP22}
&   $(\epsilon,\delta)$-DP &   
$\geq \frac{132d\log\frac{4}{\delta}}{5\epsilon^2 n^2}$ 
&   
$\geq 2\left(n+\frac{528 \log\frac{4}{\delta}}{5\epsilon^2}\right)$ 
&   
$\geq \lambda (|\calT| - f_T)$ 
&   $\times$ (Proposition~\ref{prop:shuffle_collusion})\\
\hline
\LWY{}~\cite{Luo_CCS22}   &   $(\epsilon,\delta)$-DP &   $\frac{32d\log\frac{2}{\delta}}{\epsilon^2 n^2}$   &   $2n \left(1+\frac{32d \log\frac{2}{\delta}}{\epsilon^2 n} \right)$    &   
\begin{tabular}{c}
$\geq \lambda (1 - f_T)$ \\ \hspace{-4mm}$+ \frac{32\lambda (d-|\calT|)\log\frac{2}{\delta}}{\epsilon^2 n}$ \hspace{-4mm}
\end{tabular}   
&   $\times$ (Proposition~\ref{prop:shuffle_collusion})\\
\hline
\end{tabular}
\hfil}
\label{tab:performance}
\end{table*}

\subsection{\SOGeo{} (Sample, One-Sided Geometric Dummies, and Shuffle)}
\label{sub:SOGeo}

\noindent{\textbf{Protocol.}}~~Our third protocol \SOGeo{} 
is a special case of \SAGeo{} where $\beta = 1 - e^{-\epsilon/2}$. 
By (\ref{eq:SAGeo_mu_ast}), 
(\ref{eq:SAGeo_q_r}), 
and (\ref{eq:SAGeo_delta}), when we set $\beta = 1 - e^{-\epsilon/2}$ in \SAGeo{}, we have $q_l = 0$, 
$q_r = \frac{1}{1+e^{\epsilon/2}}$, 
$\delta = 0$, $\nu = 0$, and $\mu = \frac{q_r}{1-q_r}$. 
In this case, the number $z_i$ of dummy values 
follows
the one-sided geometric distribution with parameter $q_r$, denoted by $\OGeo(q_r)$. 
In $\OGeo(q_r)$, the probability mass function at $z_i = k$ is given by 
\begin{align*}
\Pr(z_i = k) = 
(1 - q_r) q_r^k ~~~~ \text{($k = 0, 1, \ldots$)}.
\end{align*}

\smallskip{}
\noindent{\textbf{Theoretical Properties.}}~~Since \SOGeo{} is a special case of \SAGeo{}, it inherits the theoretical properties of \SAGeo{}. 
Notably, \SOGeo{} provides pure $\epsilon$-DP ($\delta=0$), as shown in Theorem~\ref{thm:SAGeo_DP}. 
Note that $\OGeo(q_r)$ alone cannot provide pure $\epsilon$-DP, as it cannot reduce the absolute frequency for each item. 
Thanks to random sampling, \SOGeo{} reduces the absolute frequency by $1$ with probability $e^{-\epsilon/2}(1-q_r)$, which brings us pure $\epsilon$-DP.

For the expected $l_2$ loss, $\OGeo(q_r)$ has variance $\sigma^2 = \frac{q_r}{(1-q_r)^2}$. 
Thus, by Theorem~\ref{thm:generalization_utility}, \SOGeo{} achieves the following expected $l_2$ loss: 
\begin{align*}
\textstyle{\E\left[ \sum_{i=1}^d (\hf_i - f_i)^2 \right] 
= \frac{1-\beta}{\beta n} + \frac{q_r d}{(1-q_r)^2 \beta^2 n^2}}.
\end{align*}
When $\epsilon$ is close to $0$ (i.e., $e^\epsilon \approx \epsilon + 1$), 
it 
can be approximated as $\frac{2}{\epsilon n} + \frac{8d}{\epsilon^2 n^2}$\conference{.}\arxiv{ (see Appendix~\ref{sub:loss_SOGeo} for details).} 

For the communication cost, \SOGeo{} achieves a very small value of $C_{tot}$, as $\beta$ is small. 
Specifically, since $\OGeo(q_r)$ has mean $\frac{q_r}{1-q_r}$, the total communication cost $C_{tot}$ of $\SOGeo{}$ is given by (\ref{eq:C_tot_generalization}), where $\mu = \frac{q_r}{1-q_r}$. 
When $\epsilon$ is close to $0$, it can be approximated as 
$C_{tot} \approx \alpha (n + d)$\conference{.}\arxiv{ (see Appendix~\ref{sub:communication_SOGeo} for details).}

\subsection{Comparison with Existing Protocols}
\label{sub:comparison}

Finally, we compare our protocols with 
seven existing shuffle protocols (four single-message protocols and three multi-message protocols). 
For single-message protocols, we consider the single-message protocols based on the GRR~\cite{Kairouz_ICML16,Wang_USENIX17}, OUE~\cite{Wang_USENIX17},  OLH~\cite{Wang_USENIX17}, and RAPPOR~\cite{Erlingsson_CCS14} (denoted by \GRRS{}, \OUES{}, \OLHS{}, and \RAPS{}, respectively). 
For multi-message protocols, we consider the protocols in~\cite{Balcer_ITC20,Cheu_SP22,Luo_CCS22} (denoted by \BC{}, \CM{}, and \LWY{}, respectively). 
For \LWY{}, we use a protocol for a small domain ($d < \tilde{O}(n)$), as this paper deals with such domain size. 
We did not evaluate a single-message protocol based on the Hadamard response~\cite{Acharya_AISTATS19} and the multi-message protocol in~\cite{Ghazi_EUROCRYPT21}, because they are less accurate than \OLHS{} and \LWY{}, respectively, as shown in~\cite{Wang_PVLDB20,Luo_CCS22}. 

For the existing single-message protocols, 
the expected $l_2$ loss and the overall gain $\GMGA$ of the LDP mechanisms 
are shown in~\cite{Wang_USENIX17} and~\cite{Cao_USENIX21}, respectively. 
Based on these results, we calculate the expected $l_2$ loss and $\GMGA$ of 
the shuffle versions 
by using Theorem~\ref{thm:shuffle} for the OUE/OLH/RAPPOR and a tighter bound in~\cite{Feldman_FOCS21} (Corollary IV.2 in~\cite{Feldman_FOCS21}) for the GRR. 
\arxiv{See Appendix~\ref{sec:existing_performance} for details.}

For the existing multi-message protocols, each fake user can send much more messages than genuine users. 
However, such attacks may be easily detected, as genuine users send much fewer messages on average. 
Thus, we consider an attack that maximizes the overall gain while keeping the expected number of messages unchanged to avoid detection. 
See Appendix~\ref{sec:existing_mechanisms} for details. 

Table~\ref{tab:performance} shows 
the performance guarantees. 
Below, we highlight the key findings: 
\begin{itemize}
    \item \SBin{} ($\beta=1$) achieves the expected $l_2$ loss at least $4$ times smaller than the existing protocols. 
    % than \GRRS{} and $8$ times smaller than \OUES{} and \OLHS{}. 
    \item \SAGeo{} ($\beta=1$) achieves the expected $l_2$ loss at least $\log \frac{4}{\delta}$ times smaller than \SBin{} ($\beta=1$). 
    % \item In \SBin{} and \SAGeo{}, the trade-off between the expected $l_2$ loss and the communication cost $C_{tot}$ can be controlled by changing $\beta$. 
    \item \SOGeo{} achieves pure $\epsilon$-DP ($\delta=0$). It also achieves the communication cost $C_{tot}$ smaller than the existing protocols when $d<n$. 
    % a much smaller expected $l_2$ loss than \SBin{}, as the former does not have a factor of $\log \frac{4}{\delta}$.  
    % \item \SBin{} suffers from a high communication cost when $\epsilon$ is small. \UDS{} addresses this issue. 
    \item In most existing protocols, the overall gain $\GMGA$ goes to $\infty$ as $\epsilon$ approaches $0$. 
    In contrast, $\GMGA$ of our protocols does not depend on $\epsilon$ and is always smaller than that of all the existing protocols. 
    \item Our protocols are robust to collusion attacks by the data collector and users (Theorem~\ref{thm:generalization}), whereas the existing protocols are not (Proposition~\ref{prop:shuffle_collusion}).  
\end{itemize}
In summary, \SAGeo{} is accurate, \SOGeo{} achieves pure $\epsilon$-DP and is communication-efficient, 
and all of our protocols are robust to data poisoning and collusion with users. 
We also show that our experimental results in Section~\ref{sec:exp} are consistent with 
Table~\ref{tab:performance}.

\section{Experimental Evaluation}
\label{sec:exp}

\subsection{Experimental Set-up}
\label{sub:set-up}

\noindent{\textbf{Datasets.}}~~We conducted experiments using 
the following four datasets with a variety of $n$ (\#users) and $d$ (\#items): 
\begin{itemize}
\item \textbf{Census dataset}:
The 1940 US Census data~\cite{ruggles2023ipums}. Following~\cite{Wang_PVLDB20}, we sampled $1\%$ of the users and used urban attributes. 
There were  
$n=602156$ users 
and $d=915$ different types of attribute values.

\item \textbf{Foursquare dataset}: The Foursquare dataset (global-scale check-in dataset) in~\cite{yang2016participatory}. 
We extracted $n=359054$ check-ins in Manhattan, assuming that each check-in is made by a different user. 
We used 
$d=407$ categories for check-in locations, 
such as Gym, Coffee Shop, Church, and Hotel. 

\item \textbf{Localization dataset}:
The dataset for recognizing person activities collected through wearable sensors~\cite{kaluvza2010agent}. It 
includes 
$n=164860$ records, with each record indicating a specific activity, including walking, falling, and sitting on the ground, amounting to $d=11$ different activity types in total.

\item \textbf{RFID dataset}:
The dataset for recognizing activities of elderly individuals based on RFID technology~\cite{torres2013sensor}. 
It comprises $n=75128$ 
records, 
each representing a specific activity, such as sitting or lying on a bed and moving around. 
It categorizes activities into $d=4$ distinct types.
\end{itemize}

\smallskip{}
\noindent{\textbf{Protocols.}}~~Using the four datasets, we compared 
our three protocols with the seven state-of-the-art shuffle protocols shown in Table~\ref{tab:performance}. 
Note that \BC{}~\cite{Balcer_ITC20} (resp.~\LWY{}~\cite{Luo_CCS22}) assumes that $n \geq \frac{400\log(4/\delta)}{\epsilon^2}$ and $\epsilon \in (0,2]$ (resp.~$n \geq \frac{32d\log(2/\delta)}{\epsilon^2}$ and $\epsilon \in (0,3]$). 
Thus, we evaluated their performance only in this case. 

It is shown in \cite{Cheu_SP22} 
that 
\CM{} provides high accuracy when the number $\xi\in\nats$ of dummy values per user is $\xi=10$. 
However, when $\epsilon$ is small (e.g., $\epsilon \leq 0.1)$, the minimum value $\xi_{min}$ of $\xi$ required 
in the privacy analysis in~\cite{Cheu_SP22} 
is larger than $10$. 
Thus, we set $\xi = \max\{10,\xi_{min}\}$ in \CM{}. 
For \BC{}, \CM{}, and \LWY{}, we used the privacy analysis results in~~\cite{Balcer_ITC20,Cheu_SP22,Luo_CCS22} and the unbiased estimators in~~\cite{Balcer_ITC20,Cheu_SP22,Luo_CCS22}. 
For the other existing protocols, 
we used the privacy amplification bound in~\cite{Feldman_FOCS21} (as described in Section~\ref{sub:existing_shuffle}) and the unbiased estimator in~\cite{Wang_USENIX17}. 

Note that the estimates can be negative values. 
Thus, following~\cite{Erlingsson_CCS14,Wang_USENIX17}, we 
kept only estimates above a significance threshold determined by the Bonferroni correction. 
Then, we uniformly assigned the remaining probabilities to each estimate below the threshold. 

\smallskip{}
\noindent{\textbf{Existing Defenses.}}~~We also evaluated three existing defenses against data poisoning and collusion with users. 
For existing defenses against data poisoning, we evaluated the normalization technique in~\cite{Cao_USENIX21} and LDPRecover~\cite{Sun_ICDE24}, 
as described in Section~\ref{sec:related}. 

For an existing defense against collusion with users, we evaluated the defense in~\cite{Wang_PVLDB20}. 
Note that the number of dummy values in their defense affects the utility. 
Specifically, if we add $an$ ($a \in \nnreals$) dummy values in the defense in~\cite{Wang_PVLDB20}, the expected $l_2$ loss is increased by $(1+a)^2$ times. 
To avoid a significant increase in the loss, we set $a$ to $0.5$. 
In this case, the MSE increases by $2.25$ times.

\smallskip{}
\noindent{\textbf{Performance Metrics.}}~~We ran each protocol $100$ times and 
evaluated the MSE as the sample mean of the squared error over the $100$ runs.  
For the communication cost, we 
evaluated $C_{tot}$. 
For data poisoning, we evaluated the average of $\GMGA$ over the $100$ runs. 
For collusion, we evaluated an actual value of $\epsilon$ after the data collector colludes with users. 
\colorB{We calculated the actual $\epsilon$ in the existing protocols, the defense in~\cite{Wang_PVLDB20}, and our protocols based on Proposition~\ref{prop:shuffle_collusion}, Corollary 10 in~\cite{Wang_PVLDB20}, and Theorem~\ref{thm:generalization}, respectively.} 

Note that although Table~\ref{tab:performance} introduces an approximation that holds when $\epsilon$ is close to 0, we 
evaluated the exact values of the MSE, $C_{tot}$, and $G_{MGA}$ in our experiments.

\begin{figure}[t]
\centering
\includegraphics[keepaspectratio=true,width=0.99\linewidth]{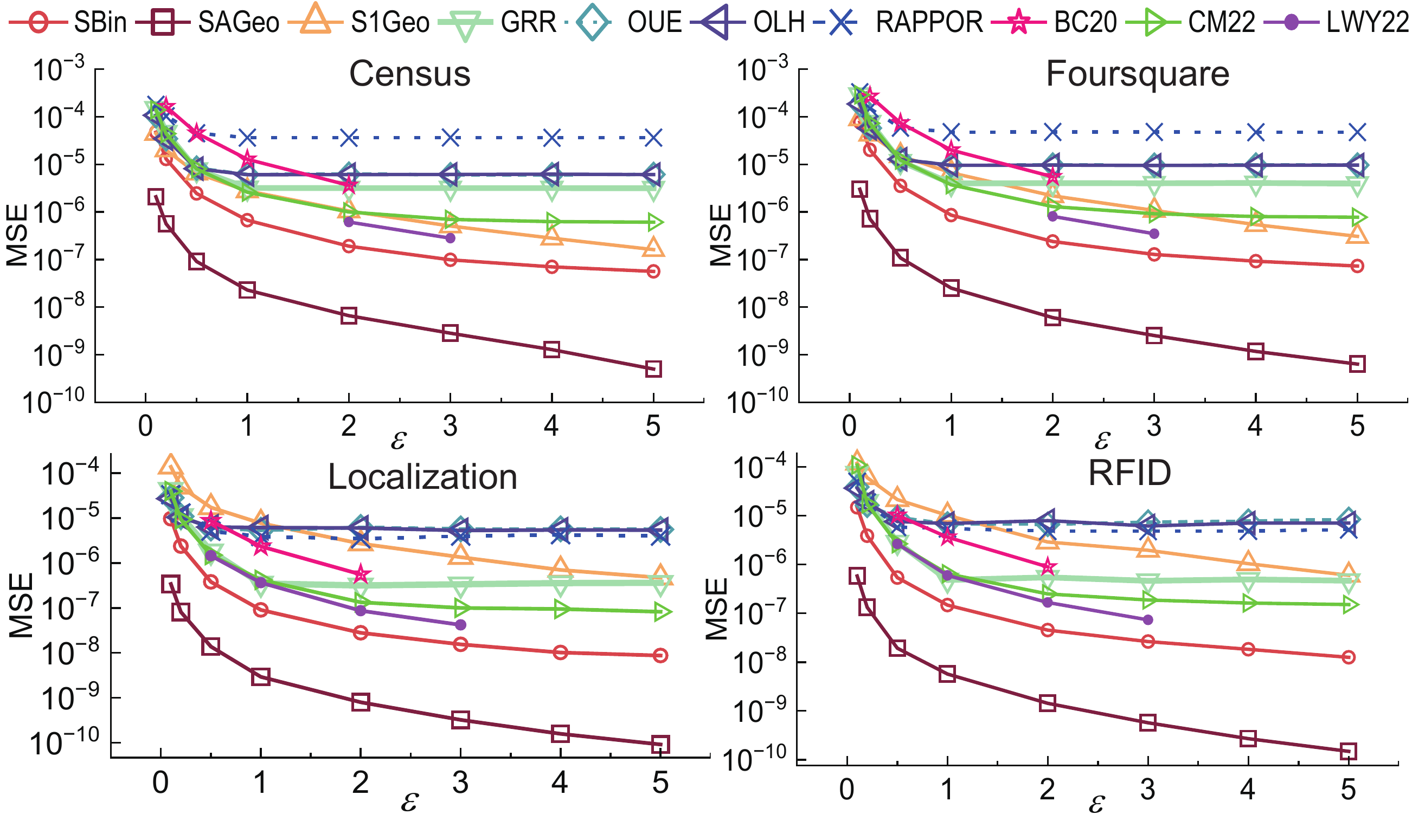}%M1
\vspace{-4mm}
\caption{MSE vs. $\epsilon$ ($\delta = 10^{-12}$, $\beta = 1$).}
\label{fig:epsilon_loss_l2}
\vspace{1mm}
\centering
\includegraphics[keepaspectratio=true,width=0.99\linewidth]{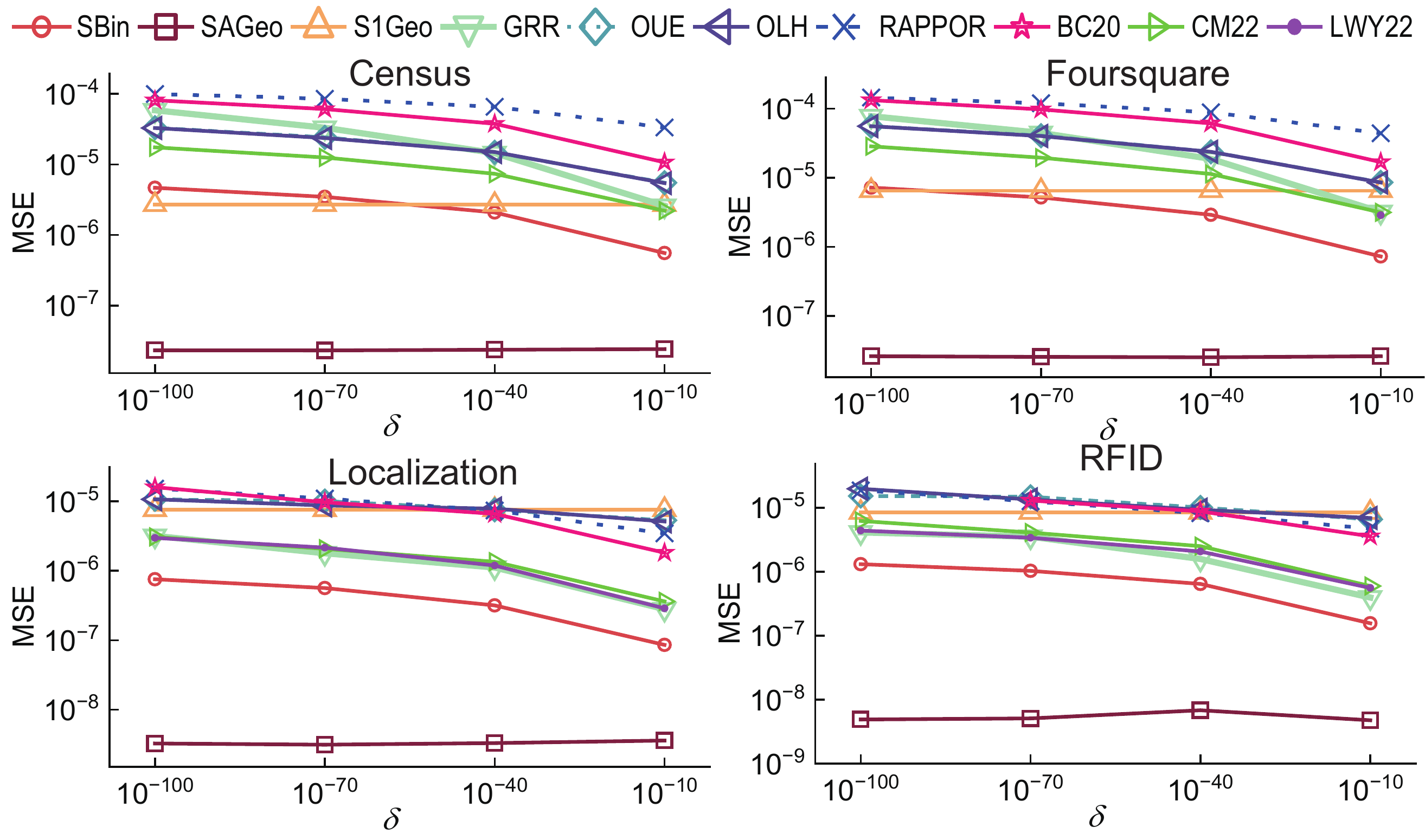}%M1
\vspace{-3mm}
\caption{MSE vs. $\delta$ ($\epsilon = 1$, $\beta = 1$).}
\label{fig:delta_loss_l2}
\vspace{1mm}
  \centering
  \includegraphics[width=0.99\linewidth]{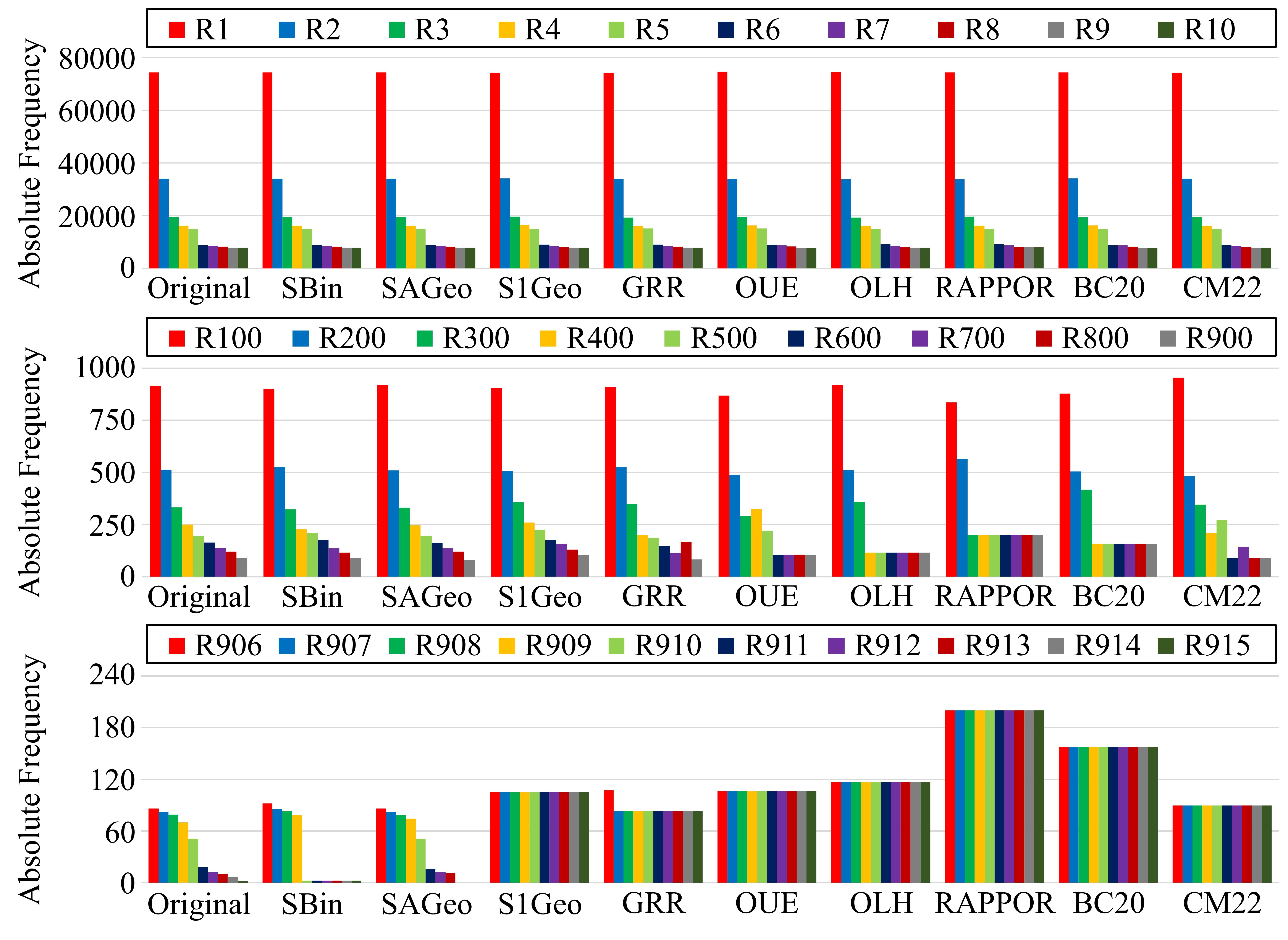}
  \vspace{-7mm}
  \caption{\colorB{Estimated histograms (Census, $\epsilon=1$, $\delta=10^{-12}$, $\beta=1$).}}
  \label{fig:histogram}
\end{figure}

\subsection{Experimental Results}
\label{sub:results}
\noindent{\textbf{Utility.}}~~We first evaluated the MSE for each protocol while changing the values of $\epsilon$ or $\delta$. 
Figure~\ref{fig:epsilon_loss_l2} and~\ref{fig:delta_loss_l2} show 
the results (we omit ``\Shuffle{}'' from the protocol names). 
We observe that \SBin{} ($\beta = 1$) outperforms all of the existing protocols in all cases. 
In addition, \SAGeo{} ($\beta = 1$) significantly outperforms \SBin{} in all cases, demonstrating 
high accuracy of \SAGeo{}. 

Figure~\ref{fig:epsilon_loss_l2} and~\ref{fig:delta_loss_l2} also show that \SOGeo{} is comparable to the existing protocols and is effective especially when $\delta$ is small. 
This is because \SOGeo{} provides pure $\epsilon$-DP. 
In other words, the MSE of \SOGeo{} is independent of the required value of $\delta$. 
Note that the MSE of \SAGeo{} is also independent of $\delta$, as it can decrease the value of $\delta$ by increasing the mode $\nu$ of $\AGeo(\nu,q_l,q_r)$. 
However, the communication cost $C_{tot}$ becomes larger in this case, as shown 
later. 

\colorB{We then examined the original histogram and each protocol's estimate in the Census dataset. 
Figure~\ref{fig:histogram} shows the results ($\epsilon=1$, $\delta=10^{-12}$). 
Here, we sorted $915$ items in descending order of the original frequency and named them as R1, ..., R915. 
We report R1-10 (popular items), R100, 200, ..., 900, and R906-915 (unpopular items).}

\colorB{Figure~\ref{fig:histogram} shows that 
all the existing protocols fail to estimate frequencies for unpopular items. 
In contrast, our \SAGeo{} successfully estimates frequencies until R913. 
\SAGeo{} fails to estimate frequencies for R914 and R915; the estimated frequencies are $0$ for these items. 
This is because the original absolute frequencies are too small 
(only $6$ and $2$ in R914 and R915, respectively). 
In other words, they are \textit{outliers}. 
We believe that this is a fundamental limitation of DP; it is impossible to provide high privacy and utility for outliers at the same time. 
We also emphasize that \SAGeo{} provides higher utility than the existing protocols, especially for less popular items.}

\smallskip{}
\noindent{\textbf{Communication Cost.}}~~We also evaluated the communication cost $C_{tot}$ of each protocol by changing $d$, $\epsilon$, $n$, and $\delta$ to various values. 
Here we set $d=100$, $\epsilon=1$, $n=10^4$, and $\delta = 10^{-12}$ as default values. 

\begin{figure}[t]
\centering
\includegraphics[keepaspectratio=true,width=0.99\linewidth]{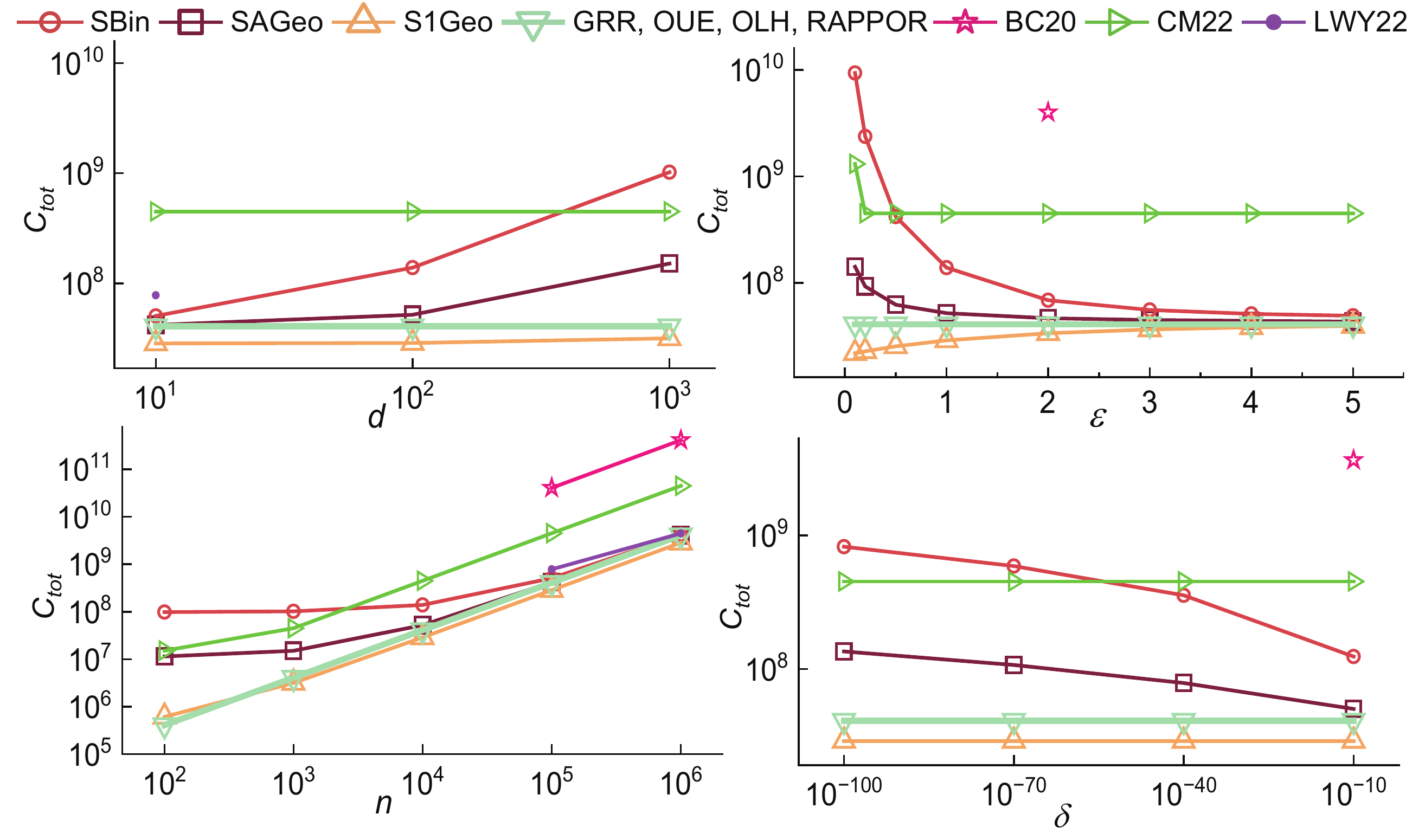}%M3 (epsilon=0.1)
\vspace{-4mm}
\caption{Communication cost $C_{tot}$ (bits). We set $d=100$, $\epsilon=1$, $n=10^4$, and $\delta = 10^{-12}$ as default values ($\beta=1$, \colorB{2048-bit RSA}). 
\colorB{The GRR, OUE, OLH, and RAPPOR have the same $C_{tot}$ because the size of their obfuscated data is $\leq d$ (resp.~$2048$) bits before (resp.~after) encryption.}
}
\label{fig:communication_cost}
\vspace{2mm}
\centering
\includegraphics[keepaspectratio=true,width=0.99\linewidth]{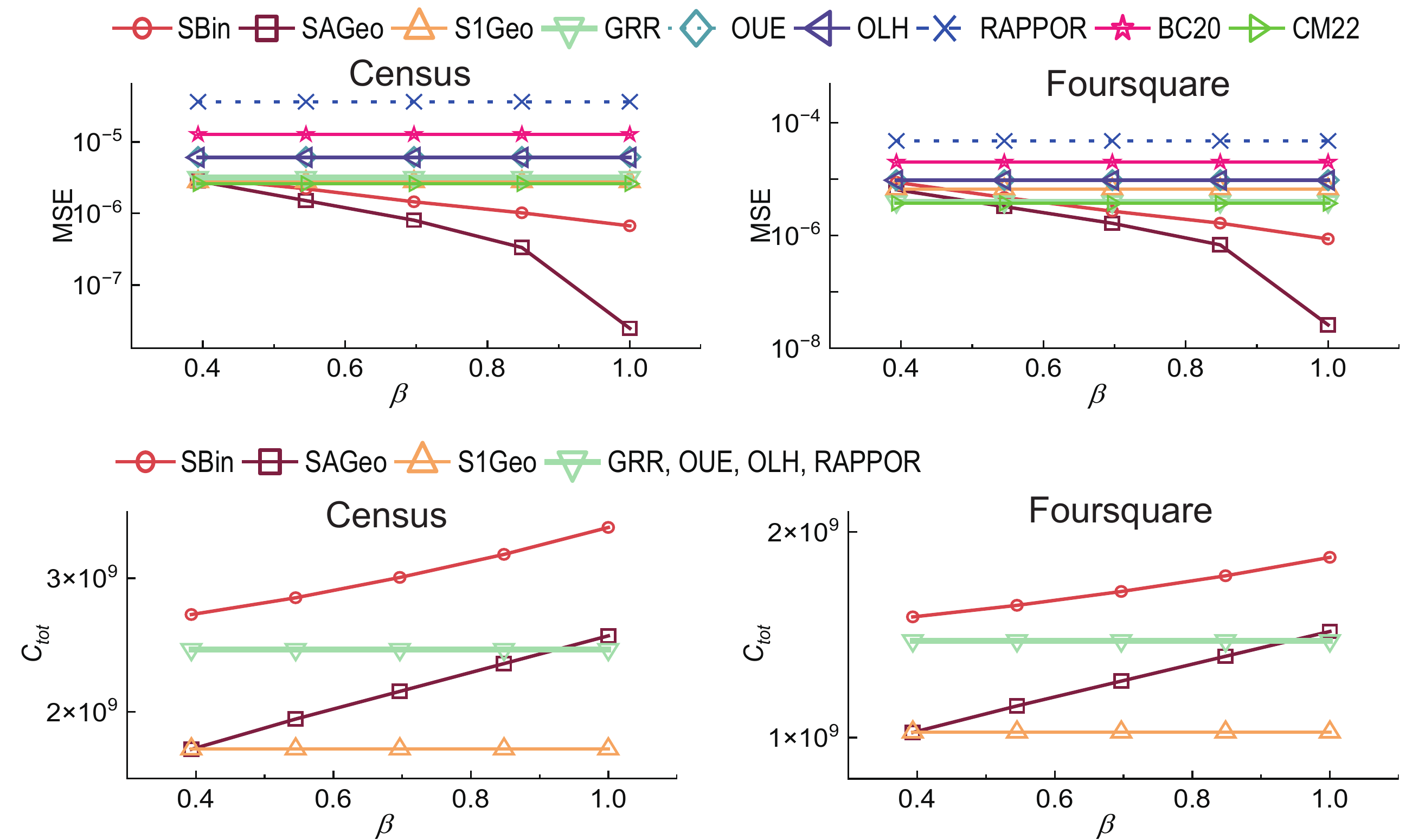}
\vspace{-7mm}
\caption{MSE/$C_{tot}$ vs. $\beta$  ($\epsilon = 1$, $\delta = 10^{-12}$, \colorB{2048-bit RSA}).}
\label{fig:beta}
\end{figure}

Figure~\ref{fig:communication_cost} shows the results. 
\colorB{Here, we used the 2048-bit RSA for encryption. 
If we use ECIES with 256-bit security in the Bouncy Castle library~\cite{bouncy}, the size of encrypted data is $712$ bits. 
Thus, we can reduce $C_{tot}$ by $0.35$ ($=712/2048)$ in this case.} 
Figure~\ref{fig:communication_cost} shows that 
when 
$\epsilon$ is small, $C_{tot}$ of \SBin{} is large. 
$C_{tot}$ of \BC{} and \CM{} is also large because they add a lot of dummy values. 
\SAGeo{} is much more efficient than these protocols and achieves $C_{tot}$ close to the existing single-message protocols. 
Figure~\ref{fig:communication_cost} also shows that \SOGeo{} provides the smallest $C_{tot}$. 

\smallskip{}
\noindent{\textbf{Changing $\beta$.}}~~As described in Section~\ref{sub:theoretical_framework}, the trade-off between the $l_2$ loss and $C_{tot}$ in our protocols can be controlled by changing $\beta$. 
Therefore, we evaluated the MSE and $C_{tot}$ in \SBin{} and \SAGeo{} while changing the value of $\beta$ from $1 - e^{-\frac{\epsilon}{2}}$ to $1$. 
Here, we used two large-scale datasets, i.e., the Census and Foursquare datasets. 

Figure~\ref{fig:beta} shows the results. 
Here, we omit the communication costs of \BC{} and \CM{}, as they are much larger than that of \SBin{}. 
We observe that when $\beta \leq 0.8$, \SAGeo{} provides a smaller $C_{tot}$ than the existing protocols. 
In addition, when $\beta = 0.8$, \SAGeo{} provides a much smaller MSE than the existing protocols. 
This means that \SAGeo{} can outperform the existing protocols in terms of both the MSE and $C_{tot}$. 

\smallskip{}
\noindent{\textbf{Robustness against Data Poisoning.}}~~Next, we evaluated the robustness of each protocol against local data poisoning attacks. 
Specifically, we set $\delta = 10^{-12}$ and the fraction $\lambda$ of fake users to $0.1$. 
Then, we evaluated 
the overall gain $G_{MGA}$ 
while changing $\epsilon$. 
For target items $\calT$, we set $|\calT| = 10$ (resp.~$2$) in the Census and Foursquare (resp.~Localization and RFID) datasets and randomly selected $|\calT|$ target items from $d$ items. 
For the existing protocols, we considered four cases: 
(i) neither the significance threshold~\cite{Erlingsson_CCS14} nor defense is used; (ii) only the significance threshold is used; (iii) the significance threshold and normalization technique~\cite{Cao_USENIX21} are used; (iv) the significance threshold and LDPRecover~\cite{Sun_ICDE24} are used. 
\colorB{For lack of space, we report the results in the Census dataset for (iii) and (iv).} 
For the existing multi-message protocols, we used the attack algorithms in Appendix~\ref{sec:existing_mechanisms} (we evaluated lower bounds on $G_{MGA}$ for \CM{} and \LWY{}). 

\begin{figure}[t]
\centering
\includegraphics[keepaspectratio=true,width=0.99\linewidth]{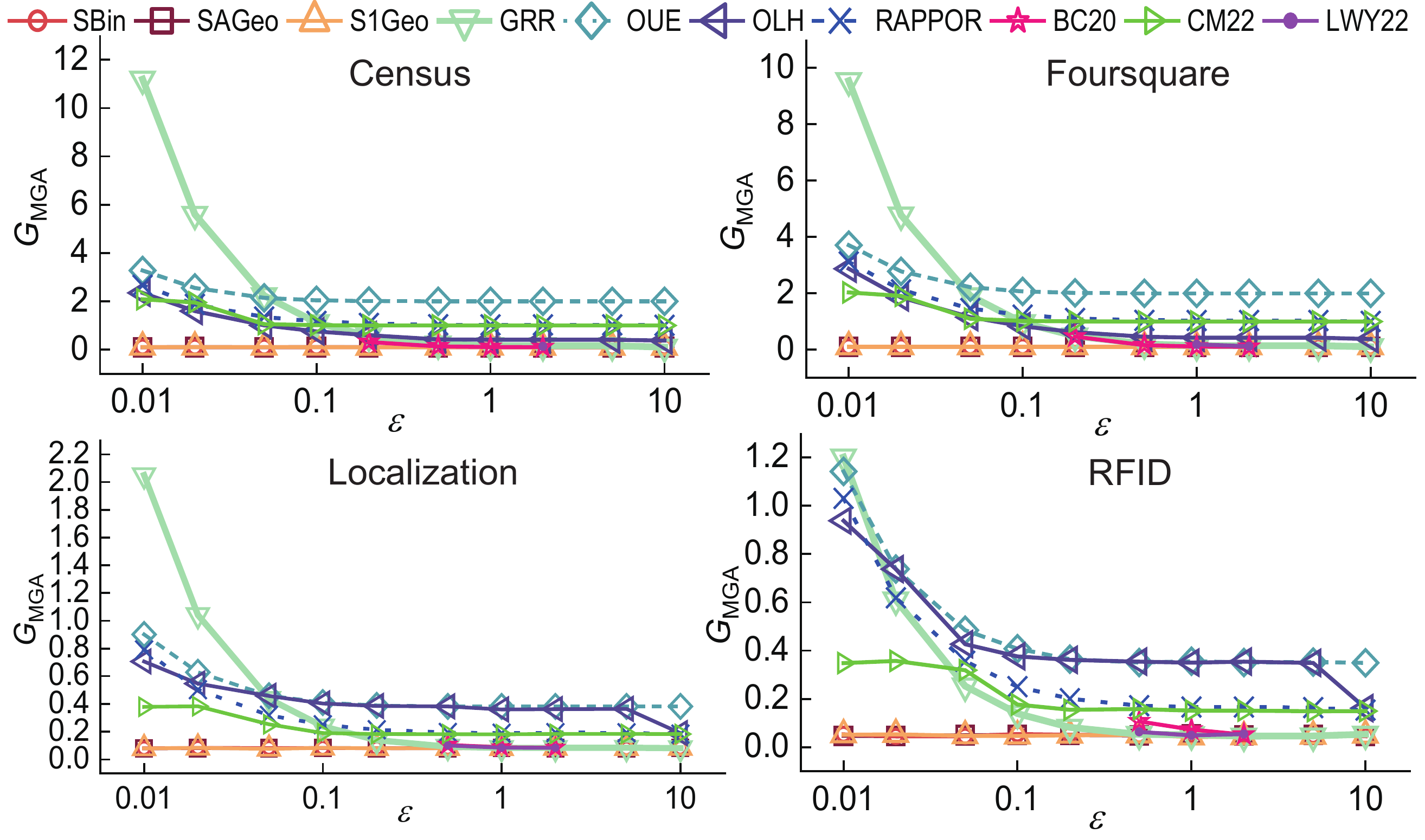}
\vspace{-7mm}
\caption{$G_{MGA}$ vs. $\epsilon$ when neither the significance threshold~\cite{Erlingsson_CCS14} nor defense is used in the existing protocols 
($\delta=10^{-12}$, $\beta=1$, $\lambda=0.1$).}
\label{fig:mga_epsilon_nothresholding}
\vspace{2mm}
\centering
\includegraphics[keepaspectratio=true,width=0.99\linewidth]{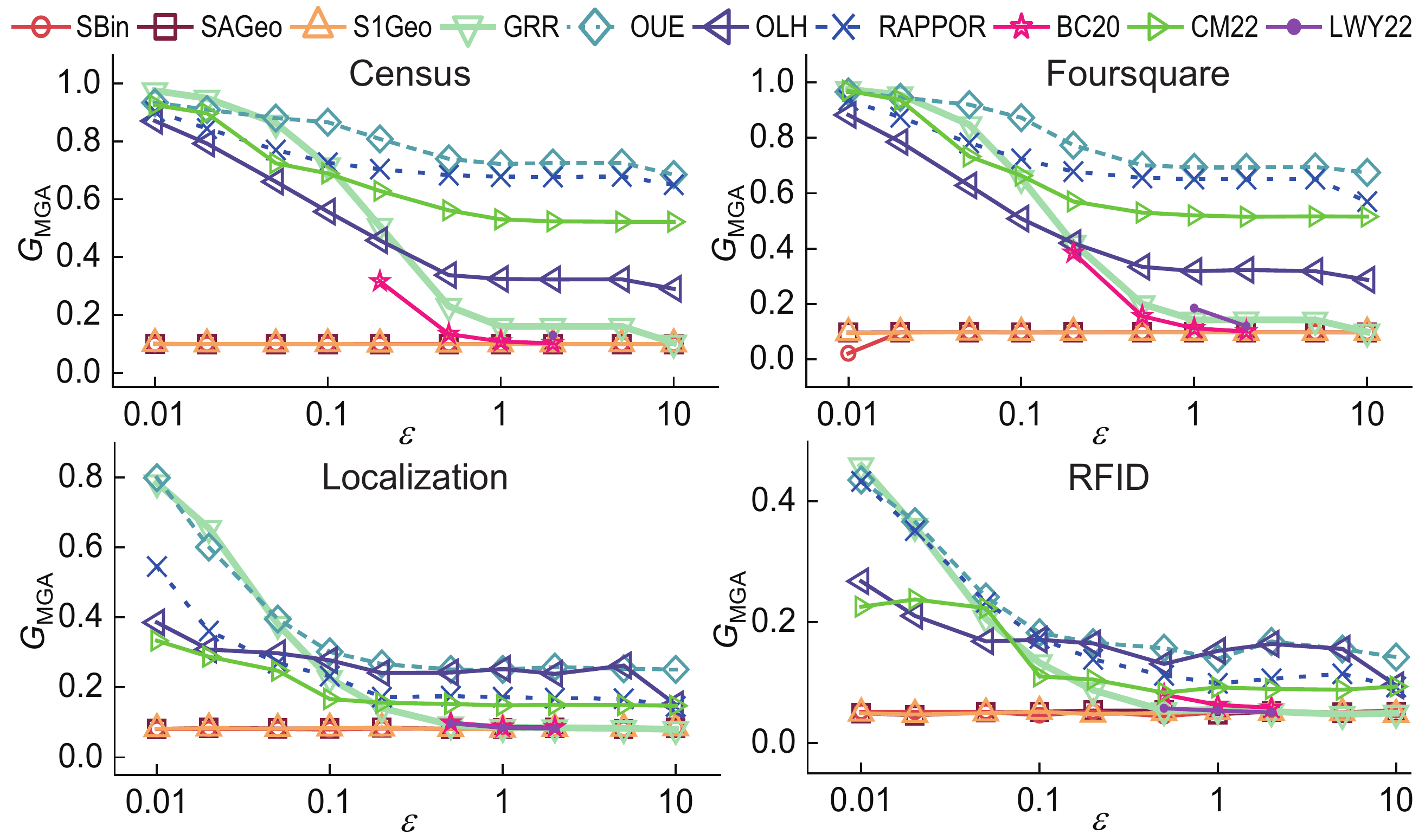}%M7
\vspace{-3mm}
\caption{$G_{MGA}$ vs. $\epsilon$ when only the significance threshold~\cite{Erlingsson_CCS14} is used in the existing protocols ($\delta=10^{-12}$, $\beta=1$, $\lambda=0.1$).}
\label{fig:mga_epsilon}
\vspace{2mm}
\centering
\includegraphics[keepaspectratio=true,width=0.99\linewidth]{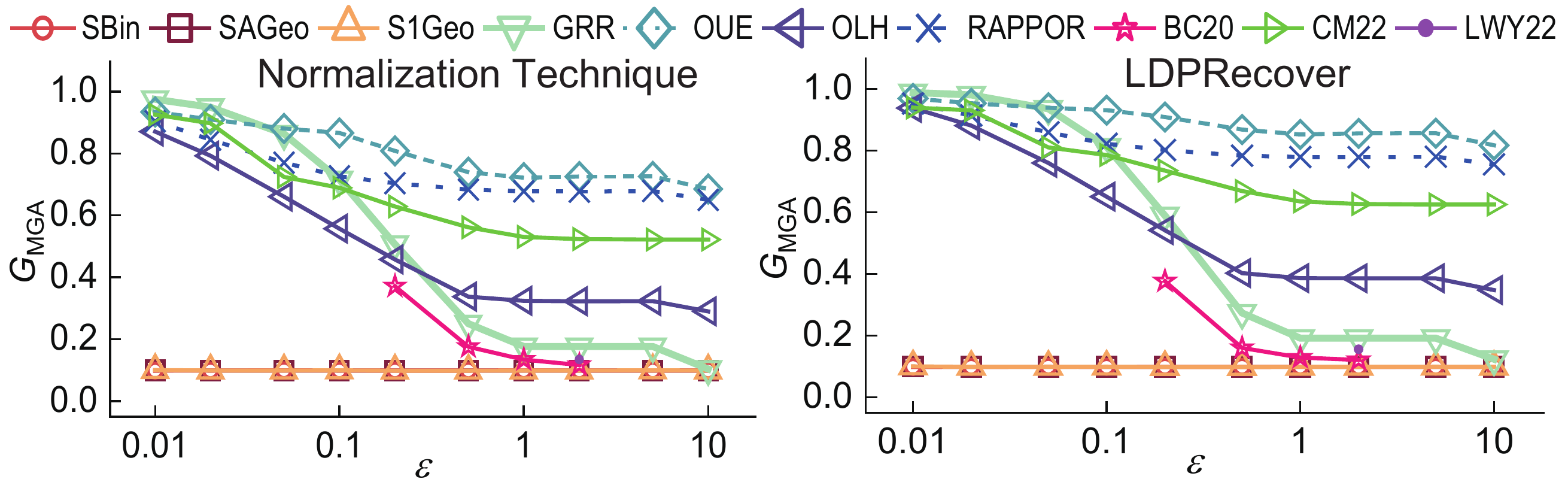}%M9
\vspace{-3mm}
\caption{$G_{MGA}$ vs. $\epsilon$ when 
\colorB{the normalization technique~\cite{Cao_USENIX21} or LDPRecover~\cite{Sun_ICDE24} is used, along with the significance threshold,} 
in the existing protocols 
(\colorB{Census,} $\delta=10^{-12}$, $\beta=1$, $\lambda=0.1$).}
\label{fig:mga_counter_epsilon}
\end{figure}

Figures~\ref{fig:mga_epsilon_nothresholding} to~\ref{fig:mga_counter_epsilon}
show the results. 
Figure~\ref{fig:mga_epsilon} shows that the significance threshold greatly reduces the overall gain $G_{MGA}$. 
\colorB{However, by comparing Figures~\ref{fig:mga_epsilon} and~\ref{fig:mga_counter_epsilon}, we can see that the normalization technique and LDPRecover do not reduce $G_{MGA}$; we also confirmed that they are ineffective in the other datasets.} 
There are two reasons for this. 
First, both the significance threshold and the normalization technique normalize the estimates. 
In other words, they are similar techniques. 
Second, LDPRecover is designed to recover more accurate estimates from heavily poisoned estimates that are not normalized. 
Thus, its effectiveness is limited when the estimates are normalized. 
We also confirmed that when the significance threshold is not used, both the significance threshold and LDPRecover greatly reduce $G_{MGA}$. 
However, they still suffer from large values of $G_{MGA}$. 

Figures~\ref{fig:mga_epsilon_nothresholding} to~\ref{fig:mga_counter_epsilon}
also show that $G_{MGA}$ of the existing protocols increases with decrease in $\epsilon$. 
In contrast, all of our protocols do not suffer from the increase of $G_{MGA}$, which is consistent with our theoretical results in Table~\ref{tab:performance}.

\smallskip{}
\noindent{\textbf{Robustness against Collusion with Users.}}~~We also 
evaluated the robustness against collusion %attacks by the data collector and 
with $|\Omega|$ users. 
Specifically, 
we refer to the privacy budget $\epsilon$ when no collusion occurs as a \textit{target $\epsilon$}. 
We set the target $\epsilon$ to $0.1$ (we also varied the target $\epsilon$ in Appendix~\ref{sec:additional_exp}) and $\delta = 10^{-12}$. 
Then, we varied the ratio $|\Omega|/n$ of colluding users from $0$ to $1-\frac{1}{n}$ and evaluated an actual $\epsilon$ after the collusion. 

\begin{figure}[t]
\centering
\includegraphics[keepaspectratio=true,width=0.99\linewidth]{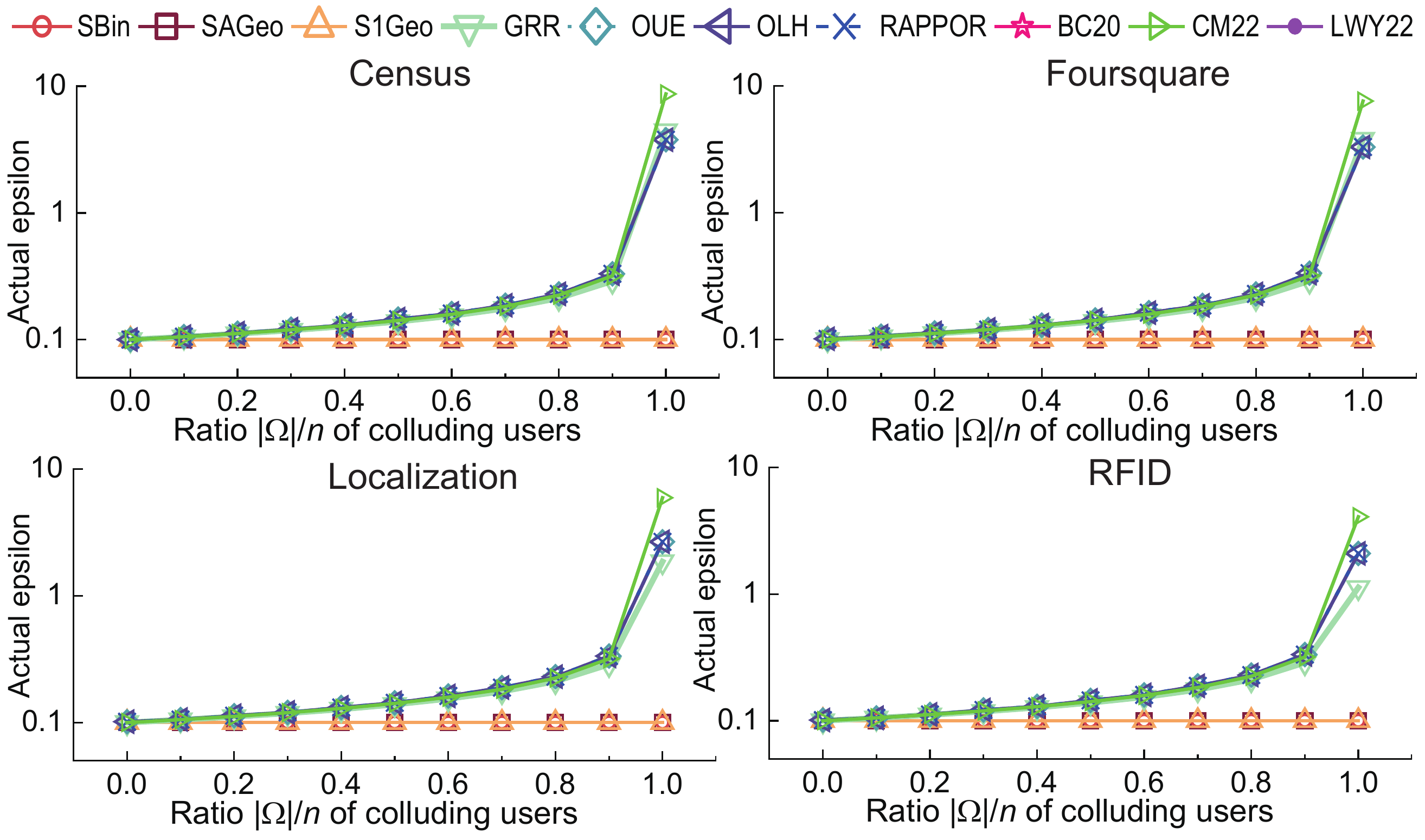}%M3 (epsilon=0.1)
\vspace{-3mm}
\caption{Actual $\epsilon$ vs. $|\Omega|/n$ 
(target $\epsilon=0.1$, $\delta = 10^{-12}$, $\beta=1$).}
\label{fig:collusion_epsilon01}
\vspace{2mm}
\centering
\includegraphics[keepaspectratio=true,width=0.99\linewidth]{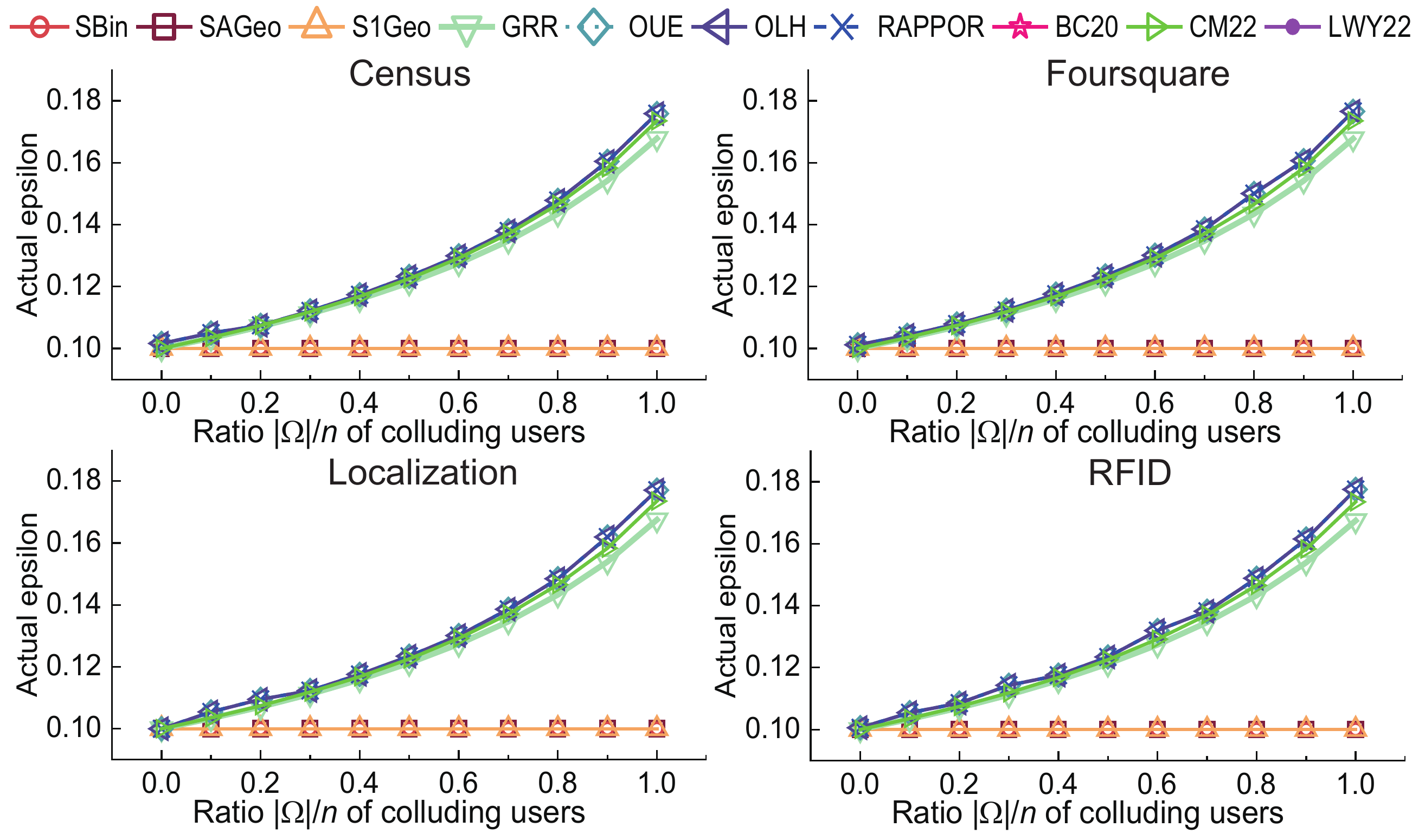}%M5 (epsilon=0.1)
\vspace{-3mm}
\caption{Actual $\epsilon$ vs. $|\Omega|/n$ when the defense in~\cite{Wang_PVLDB20} is used in the existing protocols (target $\epsilon=0.1$, $\delta = 10^{-12}$, $\beta=1$).}
\label{fig:collusion_counter_epsilon01}
\end{figure}

Figures~\ref{fig:collusion_epsilon01} and \ref{fig:collusion_counter_epsilon01} show the results when no defense and the defense in~\cite{Wang_PVLDB20} are used in the existing protocols, respectively.
The actual $\epsilon$ of the existing protocols increases with increase in 
$|\Omega|/n$. 
Although the defense in~\cite{Wang_PVLDB20} greatly reduces the actual $\epsilon$, 
the actual $\epsilon$ is still increased by the collusion attacks 
\colorB{(the defense in~\cite{Wang_PVLDB20} also increases the MSE by $2.25$ times in our experiments).} 
In contrast, all of our protocols 
always keep the actual $\epsilon$ at the target $\epsilon$ ($=0.1$), demonstrating the robustness of our protocols. 

\smallskip{}
\noindent{\textbf{\colorB{Run Time.}}}~~\colorB{Finally, we measured the run time of the shuffler and the data collector in the Census dataset using a workstation with Intel Xeon W-2295 (3.00 GHz, 18Core) and 128 GB main memory. 
For encryption schemes, we used the 2048-bit RSA or ECIES with 256-bit security in the Bouncy Castle library~\cite{bouncy}. 
We also measured the run time when no encryption scheme was used.}

\begin{figure}[t]
\centering
\includegraphics[keepaspectratio=true,width=0.99\linewidth]{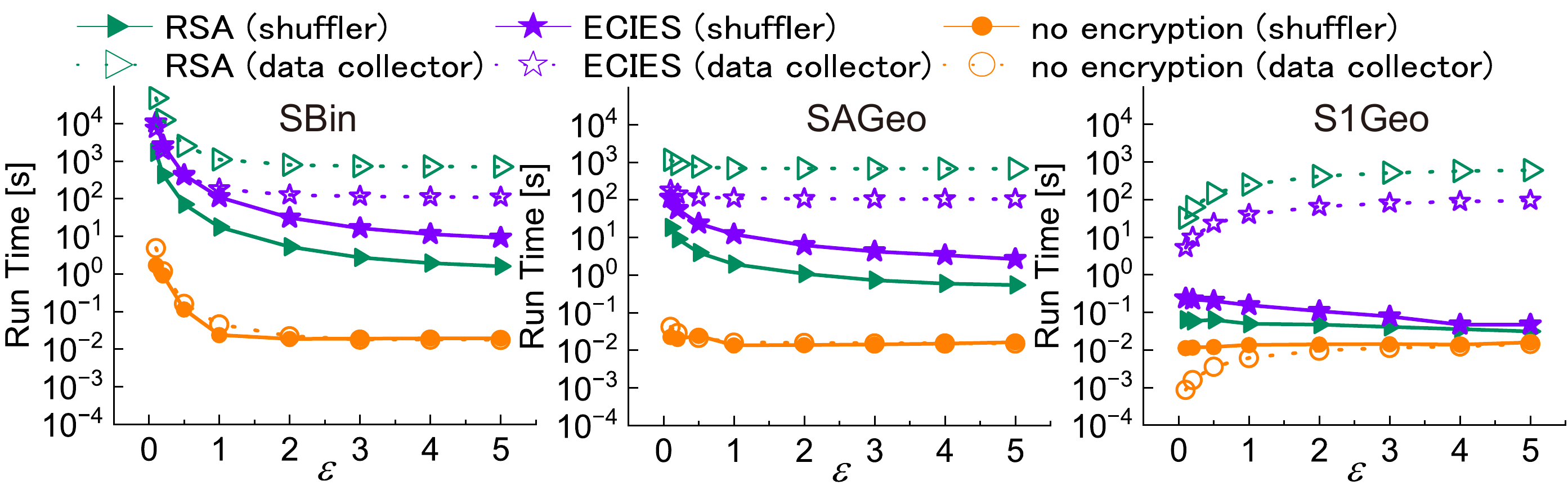}%M9
\vspace{-3mm}
\caption{\colorB{Run time of the shuffler and the data collector (Census, $\delta=10^{-12}$, $\beta=1$).}}
\label{fig:run_time}
\end{figure}

\colorB{Figure~\ref{fig:run_time} shows the results. 
We observe that the processing time for encryption and decryption is dominant. 
Since our protocols encrypt or decrypt dummy values for each item, 
the run time is linear in $d$. 
Thus, when $d$ is very large, the run time can be large\footnote{\colorB{Moreover, the shuffler needs to wait for responses from all users before shuffling in real systems. This applies to all existing shuffle protocols.}}. 
The same applies to the communication cost. 
Improving the efficiency of our protocols is left for future work.}

\section{Conclusions}
\label{sec:conclusions}
We proposed a generalized framework for local-noise-free protocols in the augmented shuffle model and rigorously analyzed its privacy, robustness, utility, and communication cost. 
Then, we proposed three concrete protocols in our framework. 
\colorB{Through theoretical analysis and comprehensive experiments, we showed that none of the existing protocols are sufficiently robust to data poisoning and collusion with users. 
Then, we showed that our protocols are robust to them. 
We also showed that \SAGeo{} achieves the highest accuracy and that \SOGeo{} provides pure $\epsilon$-DP and 
a low communication cost.} 

\ifCLASSOPTIONcompsoc
  \section*{\colorB{Acknowledgments}}
\else
  \section*{Acknowledgment}
\fi

\colorB{This study was supported in part by JSPS KAKENHI 22H00521, 24H00714, 24K20775, JST AIP Acceleration Research JPMJCR22U5, and JST CREST 
JPMJCR22M1.}

\bibliographystyle{IEEEtran}
\bibliography{main}

\appendices

\section{Results of Additional Experiments}
\label{sec:additional_exp}

Figure~\ref{fig:mga_lambda} shows 
the relationship between $G_{MGA}$ and the fraction $\lambda$ of fake users. 
The existing protocols suffers from large $G_{MGA}$ when $\lambda$ is large. 
Figure~\ref{fig:collusion_epsilon} shows the relationship between the actual $\epsilon$ and the target $\epsilon$ ($|\Omega| / n = 0.1$, $\delta = 10^{-12}$). 
The existing protocols rapidly increase the actual $\epsilon$ with increase in $|\Omega| / n$. 
The existing pure shuffle protocols using LDP mechanisms 
converge to the local privacy budget $e_L$. 
For example, when the target $\epsilon$ is $1$ in 
the Census dataset ($n=602156$), 
the actual $\epsilon$ in these existing protocols is as high as $7.2$. 
In contrast, our protocols 
always keep the actual $\epsilon$ at the same value of the target $\epsilon$. 

\begin{figure}[t]
\centering
\includegraphics[keepaspectratio=true,width=0.99\linewidth]{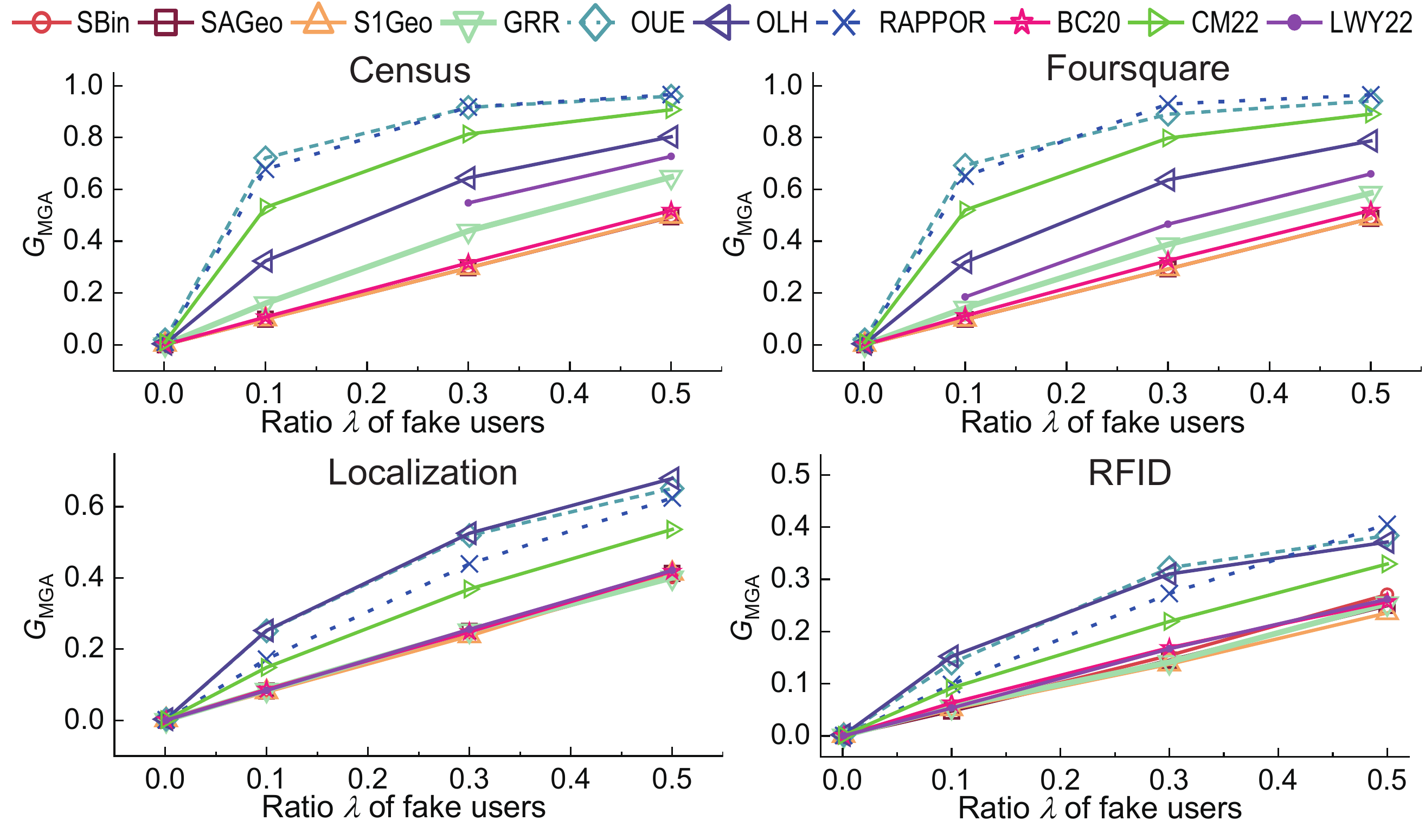}%M6
\vspace{-3mm}
\caption{$G_{MGA}$ vs. $\lambda$ when only the significance threshold~\cite{Erlingsson_CCS14} is used in the existing protocols ($\epsilon=1$, $\delta=10^{-12}$, $\beta=1$).}
\label{fig:mga_lambda}
\vspace{4mm}
\centering
\includegraphics[keepaspectratio=true,width=0.99\linewidth]{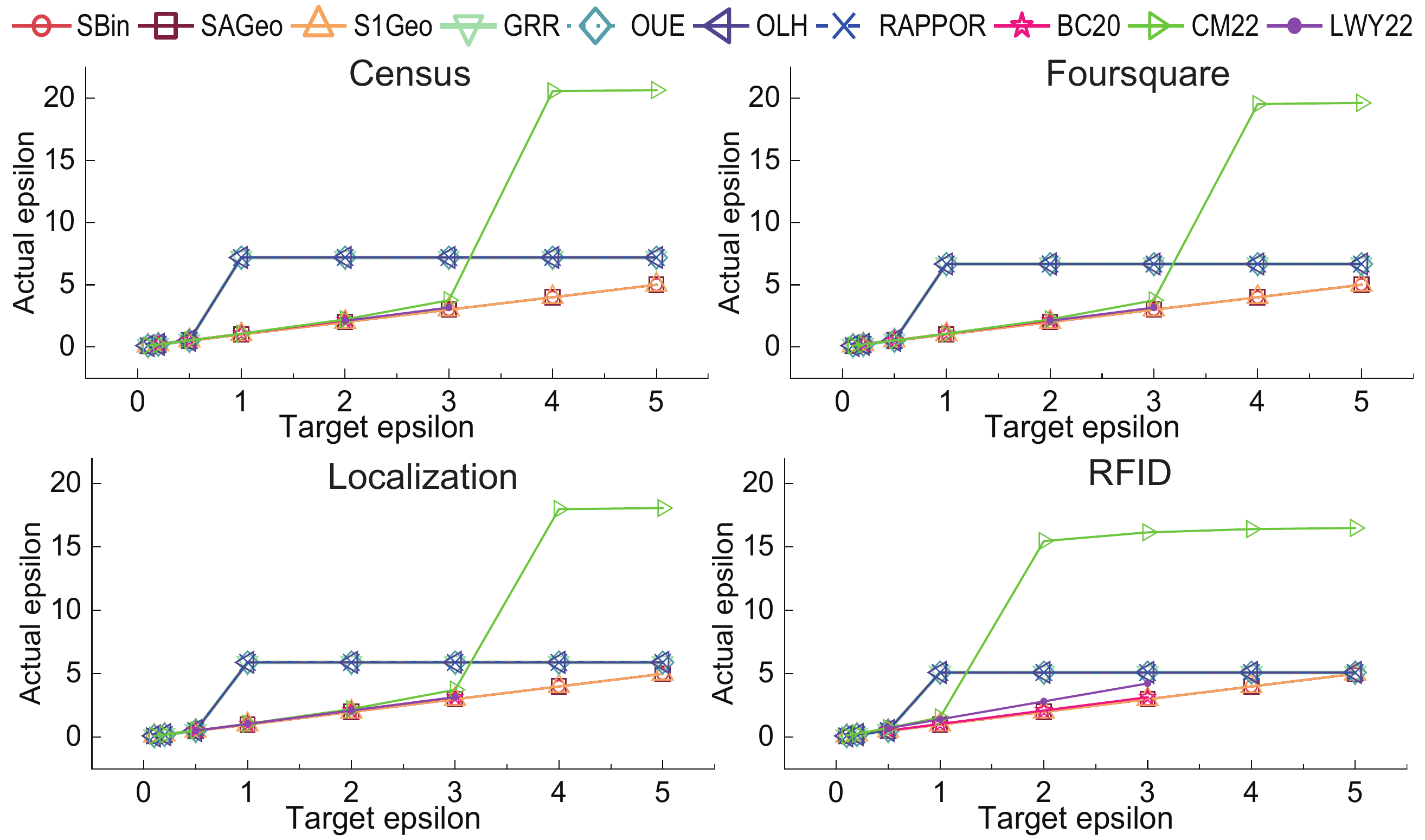}%M4
\vspace{-3mm}
\caption{Actual $\epsilon$ vs. target $\epsilon$ 
($|\Omega| / n = 0.1$, $\delta = 10^{-12}$, $\beta=1$).}
\label{fig:collusion_epsilon}
\end{figure}

\section{Proofs of Statements in Sections~\ref{sec:preliminaries} and \ref{sec:shuffle}}
\label{sec:proofs}

\subsection{Proof of Proposition~\ref{prop:shuffle_collusion}}
\label{sub:prop_shuffle_collusion}
Assume that the data collector obtains 
$\calM_S^*(D) = (\calM_S(D), \colorB{(\calR(x_i))_{i \in \Omega}})$. 
The tuple of shuffled values $\calM_S(D) = 
(\calR(x_{\pi(1)}), \ldots, \calR(x_{\pi(n)}))$ is 
equivalent to the histogram of noisy values $\{\calR(x_i) | i \in [n] \}$~\cite{Balle_CRYPTO19}. 
After obtaining 
$\colorB{(\calR(x_i))_{i \in \Omega}}$, 
the data collector can subtract $\{\calR(x_i) | i \in \Omega \}$ from the histogram of $\{\calR(x_i) | i \in [n] \}$ to obtain the histogram of $\{\calR(x_i) | i \notin \Omega \}$. 
The data collector cannot obtain more detailed information about $\{\calR(x_i) | i \notin \Omega \}$, as 
$\calR(x_{\pi(1)}), \ldots, \calR(x_{\pi(n)})$ are randomly shuffled. 

Thus, the information about $\{\calR(x_i) | i \notin \Omega \}$ available to the data collector is equivalent to the histogram of $\{\calR(x_i) | i \notin \Omega \}$, which is obtained by applying 
$\calM_S$ to $n - |\Omega|$ input values $\{x_i | i \notin \Omega \}$. 
Thus, $\calM_S^*$ provides $(\epsilon^*, \delta)$-DP for $\Omega$-neighboring databases, where $\epsilon^* = g(n - |\Omega|, \delta)$. 
\qed

\subsection{\colorB{Proof of Lemma~\ref{lem:equivalence_add_noise}}}
\label{sub:lem_equivalence_add_noise}
\colorB{The input data's histogram  
is expressed as $\sum_{i=1}^n \bmx_i$. 
The random sampling operation changes the histogram from $\sum_{i=1}^n \bmx_i$ to $\sum_{i=1}^n a_i \bmx_i$. 
The dummy data addition operation adds non-negative discrete noise $\bmz$ to the histogram. 
The shuffling does not change the histogram. 
Thus, the histogram calculated by the data collector is: $\bmth = (\sum_{i=1}^n a_i \bmx_i) + \bmz$. \qed}

\subsection{\colorB{Proof of Lemma~\ref{lem:equivalence_shuffle}}}
\label{sub:lem_equivalence_shuffle}
\colorB{Since the output data $(\tx_{\pi(1)}, \cdots, \tx_{\pi(\tn^*)})$ of $\calG_{\calD,\beta}$ (Algorithm~\ref{alg:framework}) are randomly shuffled, they are equivalent to their histogram $\bmth$. This observation is pointed out in~\cite{Balle_CRYPTO19} and can be explained as follows. 
The histogram $\bmth$ can be obtained from any possible permutation of input and dummy values; e.g.,  $\bmth = (1,1,2)$ can be obtained from $(1,2,3,3)$, $(1,3,2,3)$, $(1,3,3,2)$, $\ldots$, or $(3,3,2,1)$. 
Since the permutation is kept secret, the shuffled values do not leak any information other than $\bmth$. 
Thus, the shuffled input and dummy values are equivalent to $\bmth$, meaning that the privacy of $\calG_{\calD,\beta}$ is equivalent to that of $\calG'_{\calD,\beta}$. \qed}

\subsection{Proof of Theorem~\ref{thm:generalization}}
\label{sub:thm_generalization}
Let $\Omega=[n]\setminus\{1\}=\{2,3,\ldots,n\}$.
It is sufficient to show that 
$\calS_{\calD,\beta}^*$ provides 
$(\epsilon,\delta)$-DP 
for $\Omega$-neighboring databases, \colorB{as it considers the worst-case scenario where the data collector colludes with $n-1$ other users}.  
Let $D=(x_1,\ldots,x_n)$ and $D'=(x_1',\ldots,x_n')$ be $\Omega$-neighboring databases.
That is, $x_k=x_k'$ for 
any $k \in \{2,3,\ldots,n\}$ 
and $x_1=j\neq x_1'=j'$.
Let $\bmx_i\in\{0,1\}^d$ (resp.\ $\bmx_i'\in\{0,1\}^d$) be 
a 
binary vector whose entries are $1$ only at position $x_i$ (resp.\ $x_i'$).

\colorB{By Lemma~\ref{lem:equivalence_add_noise}, the histogram calculated by the data collector is $\bmth=\sum_{i=1}^n a_i\bmx_i+\bmz$ (resp.~$\bmth'=\sum_{i=1}^n a_i\bmx_i'+\bmz$) if $D$ (resp.~$D'$) is input. 
By Lemma~\ref{lem:equivalence_shuffle}, the output of the augmented shuffler part $\calG_{\calD,\beta}$ is equivalent to the histogram $\bmth$ (or $\bmth'$).
Thus, it is sufficient to show that} 
\begin{align*}
    % \Pr[(\bmH,(x_i)_{i\in\Omega})\in S]\leq e^{\epsilon} \Pr[(\bmH',(x_i')_{i\in\Omega})\in S] + \delta
    \colorB{\Pr[(\bmth,(x_i)_{i\in\Omega})\in S]\leq e^{\epsilon} \Pr[(\bmth',(x_i')_{i\in\Omega})\in S] + \delta}
\end{align*}
for any set $S$ of outcomes.
It is important to 
note that $D$ and $D'$ are fixed. 
In other words, the randomness comes from only $a_1, \ldots, a_n$ and $\bmz$. 
Since $x_i=x_i'$ for all $i\in\Omega$, 
the above is equivalent to showing that
\begin{align*}
    % \Pr[\bmH\in S']\leq e^{\epsilon} \Pr[\bmH'\in S'] + \delta
    \colorB{\Pr[\bmth \in S']\leq e^{\epsilon} \Pr[\bmth' \in S'] + \delta}
\end{align*}
for any set 
$S' \subseteq \nnints^d$. 
Since the distributions of $\sum_{i\in\Omega} a_i\bmx_i$ and $\sum_{i\in\Omega} a_i\bmx_i'$ are identical, this 
is equivalent to showing 
\begin{align}\label{eq:general-1}
    % \Pr[a_1\bmx_1+\bmz\in S']\leq e^{2\epsilon} \Pr[a_1\bmx_1'+\bmz\in S'] + 2\delta
    \Pr[a_1\bmx_1+\bmz\in S']\leq e^{\epsilon} \Pr[a_1\bmx_1'+\bmz\in S'] + \delta
\end{align}
for any set 
$S' \subseteq \nnints^d$. 
If $k \in [d]$ is not $j$ or $j'$, the $k$-th entries of both $a_1\bmx_1+\bmz$ and $a_1\bmx_1'+\bmz$ are $z_k$ and hence their distributions are identical.
Note that the $j$-th entry of $a_1\bmx_1+\bmz$ is $a_1+z_j$ and hence follows the same distribution as $\calM_{\calD,\beta}(1)$.
On the other hand, the $j$-th entry of $a_1\bmx_1'+\bmz$ is $z_j$ and hence follows the same distribution as $\calM_{\calD,\beta}(0)$.
Similarly, the $j'$-th entry of $a_1\bmx_1'+\bmz$ follows the same distribution as $\calM_{\calD,\beta}(1)$ and the $j'$-th entry of $a_1\bmx_1+\bmz$ follows the same distribution as $\calM_{\calD,\beta}(0)$.
Therefore, 
by group privacy~\cite{DP}, the inequality~(\ref{eq:general-1}) holds. 
\qed

\subsection{Proof of Theorem~\ref{thm:generalization_poisoning}}
\label{sub:thm_generalization_poisoning}
Since $\calS_{\calD,\beta}$ does not add any local noise, a message from each fake user takes a value from $1$ to $d$; i.e., $\bmy' = (y'_1, \cdots, y_{n'}) \in [d]^{n'}$. 
Let $f'_i \in [0,1]$ be the relative frequency after data poisoning; i.e., $f'_i = \frac{1}{n+n'} (\sum_{j=1}^n \mathone_{x_j = i} + \sum_{j=1}^{n'} \mathone_{y'_j = i})$. 
$\calS_{\calD,\beta}$ outputs an unbiased estimate; i.e., for any $i \in [d]$, 
we have
$\E[\hf_i] = f_i$ and $\E[\hf'_i] = f'_i$. 
Thus, the overall gain $G(\bmy')$ can be written as follows: 
\begin{align}
    \textstyle{G(\bmy')
    % = \sum_{i \in \calT} \E[\Delta \hf_i] 
    = \sum_{i \in \calT} (\E[\hf'_i] - \E[\hf_i]) 
    = \sum_{i \in \calT} (f'_i - f_i).}
    \label{eq:Gbmy}
\end{align}
This is maximized when all the fake users send an item in $\calT$. 
In this case, we have 
\begin{align}
\textstyle{(n+n') \left(\sum_{i \in \calT} f'_i \right) = n \left(\sum_{i \in \calT} f_i \right) + n'.}
\label{eq:nnpfip}
\end{align}
By (\ref{eq:Gbmy}) and (\ref{eq:nnpfip}), 
$\calS_{\calD,\beta}$ provides the following guarantee: 
\begin{align*}
    \GMGA 
    % = \max_{\bmy'} G(\bmy') 
    = \textstyle{\frac{n}{n+n'} \left(\sum_{i \in \calT} f_i \right) + \frac{n'}{n+n'} - \sum_{i \in \calT} f_i} 
    = \lambda (1 - f_T).
\end{align*}
where $\lambda = \frac{n'}{n+n'}$ and $f_T = \sum_{i \in \calT} f_i$. 
\qed

\subsection{Proof of Theorem~\ref{thm:generalization_utility}}
\label{sub:thm_generalization_utility}
\colorB{Since $\hf_i = \frac{1}{n\beta}(\tih_i - \mu)$, we have 
$\E[\hf_i] 
= \frac{1}{n\beta}(\E[\tih_i] - \mu) 
= \frac{1}{n\beta}(\beta n f_i+\mu - \mu) 
= f_i$, 
which means that $\hf_i$ is an unbiased estimate of $f_i$.}
Next, 
we analyze the expected $l_2$ loss. 
By the law of total variance, it holds that
\begin{align*}
    \V[\tih_i] 
    &= \E[\V[\tih_i | z_i]] + \V[\E[\tih_i | z_i]] \\
    &= \E[nf_i \beta(1-\beta)] + \V[nf_i \beta + z_i] \\
    &= nf_i \beta(1-\beta) + \V[z_i] 
    = nf_i \beta(1-\beta)+\sigma^2.
\end{align*}
Since $\hf_i$ is unbiased and $\V[\hf_i] = \frac{\V[\tih_i]}{\beta^2 n^2}$, 
we have 
\begin{align*}
\textstyle{\E\left[ \sum_{i=1}^d (\hf_i - f_i)^2 \right] 
= \sum_{i=1}^d \V[\hf_i] 
= \frac{1-\beta}{\beta n}+\frac{\sigma^2d}{\beta^2 n^2}.}
\end{align*}
\qed

\subsection{Proof of Theorem~\ref{thm:SBin_DP}}
\label{sub:thm_SBin_DP}
Let $\delta_0 = 2 e^{-\frac{\eta^2 M}{2}}$. 
We first prove that $\calM_{\Bin,1}$ 
($\beta=1$) 
provides 
$(\epsilon_0,\delta_0)$-DP. 
Then, we prove 
Theorem~\ref{thm:SBin_DP}.

\smallskip{}
\noindent{\textbf{$(\epsilon_0,\delta_0)$-DP of $\calM_{\Bin,1}$.}}~~Let $x,x' \in \{0,1\}$. 
Let $o \in \mathrm{Range}(\calM_{\Bin,1})$. 
We show that the following inequalities hold with probability at least $1-\delta_0$:
\begin{align}
&e^{-\epsilon_0} \Pr[\calM_{\Bin,1}(x') = o] \nonumber\\
&\leq \Pr[\calM_{\Bin,1}(x) = o] 
\leq e^{\epsilon_0} \Pr[\calM_{\Bin,1}(x') = o],
\label{eq:double_DP_inequality_2}
\end{align}
which implies that $\calM_{\Bin,1}$ provides $(\epsilon_0,\delta_0)$-DP~\cite{Kasivisiwanathan_JPC14}.

First, assume that 
$x=1$ and $x'=0$. 
Let 
$z \in [0, M]$ be a value sampled from $B(M,\frac{1}{2})$ when running $\calM_{\Bin,1}(x)$. 
Let $\bar{z} \in [-\frac{M}{2}, \frac{M}{2}]$ be a value such that $z = \frac{M}{2} + \bar{z}$. 
Let $b(z;M,\frac{1}{2})$ be the probability of $z$ successes in $B(M,\frac{1}{2})$, i.e., 
$b(z;M,\frac{1}{2}) = b(\frac{M}{2}+\bar{z};M,\frac{1}{2}) = \binom{M}{\frac{M}{2}+\bar{z}} (\frac{1}{2})^M$. 
When $\bar{z} \leq \frac{M}{2}-1$, we have
\begin{align}
\textstyle{\frac{\Pr[\calM_{\Bin,1}(x) = o]}{\Pr[\calM_{\Bin,1}(x') = o]}} 
= \textstyle{\frac{b(\frac{M}{2}+\bar{z};M,\frac{1}{2})}{b(\frac{M}{2}+\bar{z}+1;M,\frac{1}{2})}} 
= \textstyle{\frac{\frac{M}{2}+\bar{z}+1}{\frac{M}{2}-\bar{z}}.}
\label{eq:calM_2_LR}
\end{align}
Note that 
$\frac{\frac{M}{2}+\bar{z}+1}{\frac{M}{2}-\bar{z}} \leq e^{\epsilon_0} \iff \bar{z} \leq \frac{M\eta}{2}$, 
where $\eta = \frac{e^{\epsilon_0} - 1}{e^{\epsilon_0} + 1} - \frac{2}{M(e^{\epsilon_0} + 1)}$. 
Thus, the second inequality in (\ref{eq:double_DP_inequality_2}) does not hold if and only if $\bar{z} > \frac{M \eta}{2}$. 
In addition, 
since $\epsilon_0 \geq \log(\frac{2}{M}+1)$, we have $\eta = \frac{M(e^{\epsilon_0} - 1) - 2}{M(e^{\epsilon_0} + 1)} \geq 0$. 
Thus, 
the probability that $\bar{z}$ is larger than $\frac{M \eta}{2}$ can be upper bounded as follows: 
\begin{align}
    &\textstyle{\Pr\left(\bar{z} > \frac{M \eta}{2}\right) = \Pr \left(z > \frac{M}{2}(1 + \eta)\right) \leq e^{-\frac{\eta^2 M}{2}}} \nonumber\\
    % \label{eq:Chernoff_1}\\
    & \text{(by the multiplicative Chernoff bound for $\textstyle{B(M,\frac{1}{2})}$~\cite{probability_computing})}. \nonumber
\end{align}
Similarly, we have 
$\frac{\frac{M}{2}+\bar{z}+1}{\frac{M}{2}-\bar{z}} \geq e^{-\epsilon_0} \iff \bar{z} \geq - \frac{M\eta}{2} - 1$. 
This means that the first inequality in (\ref{eq:double_DP_inequality_2}) does not hold if and only if $\bar{z} < - \frac{M \eta}{2} - 1$. 
The probability that $\bar{z}$ is smaller than $- \frac{M \eta}{2} - 1$ can be upper bounded as follows: 
$\Pr(\bar{z} < - \frac{M \eta}{2} - 1 ) < \Pr(\bar{z} < - \frac{M \eta}{2} ) = \Pr(\bar{z} > \frac{M \eta}{2} ) \text{ (as $B(M,\frac{1}{2})$ is symmetric)} = \Pr (z > \frac{M}{2}(1 + \eta)) \leq e^{-\frac{\eta^2 M}{2}}$.
Thus, 
both the inequalities in (\ref{eq:double_DP_inequality_2}) hold with probability at least $1 - \delta_0$, where 
$\delta_0 = 2 e^{-\frac{\eta^2 M}{2}}$. 

Next, assume that 
$x=0$ and $x'=1$. 
When $\bar{z} \geq - \frac{M}{2}+1$, % we have
\begin{align}
\textstyle{\frac{\Pr[\calM_{\Bin,1}(x) = o]}{\Pr[\calM_{\Bin,1}(x') = o]} 
= \frac{b(\frac{M}{2}+\bar{z};M,\frac{1}{2})}{b(\frac{M}{2}+\bar{z}-1;M,\frac{1}{2})}} 
&= \textstyle{\frac{\frac{M}{2}-\bar{z}+1}{\frac{M}{2}+\bar{z}}.}
\label{eq:calM_2_LR_2}
\end{align}
(\ref{eq:calM_2_LR_2}) is equivalent to 
(\ref{eq:calM_2_LR}) when the sign of $\bar{z}$ is reversed. 
Since $B(M,\frac{1}{2})$ is symmetric, 
(\ref{eq:double_DP_inequality_2}) holds 
with probability at least $1 - \delta_0$. 
Thus, $\calM_{\Bin,1}$ provides $(\epsilon_0,\delta_0)$-DP. 

\smallskip{}
\noindent{\textbf{$(\frac{\epsilon}{2},\frac{\delta}{2})$-DP of $\calM_{\Bin,\beta}$.}}~~We use the following lemma: 
\begin{lemma}
\label{lem:amp_smpl}
Let $\epsilon_1 \in \nnreals$, $\delta_1 \in [0,1]$, and $\beta_1, \beta_2 \in (0,1]$. 
If $\calM_{\Bin,\beta_1}$ provides $(\epsilon_1,\delta_1)$-DP, then 
$\calM_{\Bin,\beta_2}$ provides $(\epsilon_2,\delta_2)$-DP, where 
$\epsilon_2 = \log(1 + \frac{\beta_2}{\beta_1}(e^{\epsilon_1} - 1))$ and 
$\delta_2 = \frac{\beta_2}{\beta_1}\delta_1$. 
\end{lemma}
This is an amplification lemma (Theorem 1 in~\cite{Li_AsiaCCS12}) applied to the algorithm $\calM_{\Bin,\beta}$. 
By Lemma~\ref{lem:amp_smpl} with $\epsilon_1 = \epsilon_0$, $\delta_1 = \delta_0$, $\beta_1 = 1$, and $\beta_2 = \beta$, 
$\calM_{\Bin,\beta}$ provides $(\frac{\epsilon}{2},\frac{\delta}{2})$-DP, where $\epsilon = 2 \log (1 + \beta (e^{\epsilon_0} - 1))$ and $\delta = 2 \beta \delta_0$.

\qed

\subsection{Proof of Theorem~\ref{thm:SAGeo_DP}}
\label{sub:thm_SAGeo_DP}
Let $\epsilon_0 = \frac{\epsilon}{2}$ 
and $\delta_0 = \frac{\delta}{2}$. 
Let $o \in \mathrm{Range}(\calM_{\AGeo,\beta})$. 
Assume that $x=0$ and $x'=1$. 
Then, we have
\begin{align*}
&\Pr[\calM_{\AGeo,\beta}(x) = o] = 
\begin{cases}
\frac{1}{\kappa} q_l^{\nu - o}   &   \text{(if $o = 0, 1, \ldots, \nu-1$)}\\
\frac{1}{\kappa} q_r^{o - \nu}   &   \text{(if $o = \nu, \nu+1, \ldots $)}
\end{cases}\\
&\Pr[\calM_{\AGeo,\beta}(x') = o] \\
&= 
\begin{cases}
    (1-\beta)\frac{1}{\kappa} q_l^{\nu} & \text{($o=0$)}\\
    (1-\beta)\frac{1}{\kappa} q_l^{\nu-o} + \beta \frac{1}{\kappa} q_l^{\nu-o+1} & \text{($o=1, 2, \ldots, \nu$)}\\
    (1-\beta)\frac{1}{\kappa} q_r^{o-\nu} + \beta \frac{1}{\kappa} q_r^{o-\nu-1} & \text{($o=\nu,\nu+1,\ldots$)}.
    % (1-\beta)\kappa q^{\nu} & \text{($x=0$)}\\
    % (1-\beta)\kappa q^{|x-\nu|} + \beta \kappa q^{|(x-1)-\nu|} & \text{($x=1, \ldots, 2\nu$)}\\
    % \beta \kappa q^{\nu} & \text{($x=2\nu+1$)}.
\end{cases}
\end{align*}
Since $q_l = \frac{e^{-{\epsilon_0}} - 1 + \beta}{\beta}$ and $q_r = \frac{\beta}{e^{\epsilon_0} - 1 + \beta}$, we have 
\begin{align*}
\textstyle{\frac{\Pr[\calM_{\AGeo,\beta}(x) = o]}{\Pr[\calM_{\AGeo,\beta}(x') = o]}} 
= 
\begin{cases}
    1 - \beta & \text{($o=0$)}\\
    e^{-\epsilon_0} & \text{($o=1, \ldots, \nu$)}\\
    e^{\epsilon_0} & \text{($o=\nu+1, \nu+2, \ldots$)}.
    % 1 - \beta & \text{($x=0$)}\\
    % e^{-\epsilon_0} & \text{($x=1, \ldots, \nu$)}\\
    % e^{\epsilon_0} & \text{($x=\nu+1, \nu+2, \ldots$)}.
\end{cases}
\end{align*}
If $\beta = 1 - e^{-\epsilon_0}$, then $1 - \beta = e^{-\epsilon_0}$. 
Thus, 
$\frac{\Pr[\calM_{\AGeo,\beta}(x) = o]}{\Pr[\calM_{\AGeo,\beta}(x') = o]}$ is always $e^{\epsilon_0}$ or $e^{-\epsilon_0}$ (note that this also holds when $x=1$ and $x'=0$). 
This means that $\calM_{\AGeo,\beta}$ provides $\epsilon_0$-DP ($\delta_0 = 0$). 
If $\beta > 1 - e^{-\epsilon_0}$, 
we have the following equations for $o=0$: 
\begin{align*}
&\Pr[\calM_{\AGeo,\beta}(x) = 0] - e^{\epsilon_0} \Pr[\calM_{\AGeo,\beta}(x') = 0] \\
&= \textstyle{\frac{1}{\kappa} q_l^{\nu} - e^{\epsilon_0} (1-\beta)\frac{1}{\kappa} q_l^{\nu}} 
= \textstyle{\frac{1}{\kappa} q_l^{\nu} (1 - e^{\epsilon_0} + \beta e^{\epsilon_0})} 
= \delta_0
\end{align*}
Thus, we have 
$\Pr[\calM_{\AGeo,\beta}(x') = 0] \leq \Pr[\calM_{\AGeo,\beta}(x) = 0]$ and 
$\Pr[\calM_{\AGeo,\beta}(x) = 0] = e^{\epsilon_0} \Pr[\calM_{\AGeo,\beta}(x') = 0] + \delta_0$. 
Similarly, when $x=0$ and $x'=1$, we have 
$\Pr[\calM_{\AGeo,\beta}(x) = 0] \leq \Pr[\calM_{\AGeo,\beta}(x') = 0]$ and 
$\Pr[\calM_{\AGeo,\beta}(x') = 0] = e^{\epsilon_0} \Pr[\calM_{\AGeo,\beta}(x) = 0] + \delta_0$. 
Since $\frac{\Pr[\calM_{\AGeo,\beta}(x) = o]}{\Pr[\calM_{\AGeo,\beta}(x') = o]}$ is always $e^{\epsilon_0}$ or $e^{-\epsilon_0}$ for $o \ne 0$, 
$\calM_{\AGeo,\beta}$ provides $(\epsilon_0,\delta_0)$-DP. 
\qed

\subsection{Proof of Proposition~\ref{prop:SAGeo_squared_error}}
\label{sub:thm_SAGeo_squared_error}

Let $X_\nonageo$ be a random variable with a \textit{non-truncated} asymmetric geometric distribution that has a mode at $0$, i.e., 
\begin{align*}
\Pr(X_\nonageo = k) = 
\begin{cases}
\frac{1}{\kappa^*} q_l^{-k}   &   \text{(if $k = 0, -1, \ldots$)}\\
\frac{1}{\kappa^*} q_r^k   &   \text{(if $k = 1, 2, \ldots$)},
\end{cases}
\end{align*}
where 
$\kappa^* = \frac{q_l}{1-q_l} + \frac{1}{1 - q_r}$. 
Note that the variance of $X_\nonageo$ is reduced by truncating $X_\nonageo$ to $[-\nu,\infty)$ (as the truncation always moves $X_\nonageo$ toward the mean). 
Thus, 
the variance $\sigma^2$ of $\AGeo(\nu,q_l,q_r)$ 
is upper bounded as 
$\sigma^2 
\leq \V[X_\nonageo]$. 
In addition, $\V[X_\nonageo]$ can be upper bounded as 
$\V[X_\nonageo] \leq \E[X_\nonageo^2] = \frac{1}{(1-q_l)\kappa^*} \sum_{k=1}^\infty (1-q_l)q_l^k k^2 + + \frac{1}{(1-q_r)\kappa^*} \sum_{k=1}^\infty (1-q_r)q_r^k k^2$. 
Notice that $\sum_{k=1}^\infty (1-q)q^{k} k^2$ is equal to the second moment of the one-sided geometric distribution. 
Let $X_\onegeo$ be a random variable with the one-sided geometric distribution with parameter $q \in (0,1)$. 
Then, the second moment $\E[X_\onegeo^2]$ can be written as 
$\E[X_\onegeo^2] 
= \V[X_\onegeo^2] + \E[X_\onegeo]^2 
= \frac{q}{(1-q)^2} + \frac{q^2}{(1-q)^2}
= \frac{q(1+q)}{(1-q)^2}$. 
Thus, 
we have 
$\sigma^2 \leq \V[X_\nonageo] \leq \frac{1}{\kappa^*} \left( \frac{q_l(1+q_l)}{(1-q_l)^3} + \frac{q_r(1+q_r)}{(1-q_r)^3} \right)$.
\qed

\section{$(\epsilon,\delta)$-DP of \SBin{} from \cite{Dwork_EUROCRYPT06,Agarwal_NeurIPS18}}
\label{sec:comparison_bound}

\begin{theorem}[$(\epsilon,\delta)$-DP of $\calS_{\Bin,1}$ derived from~\cite{Dwork_EUROCRYPT06}]
\label{thm:SBin_DP_Dwork}
$\calS_{\Bin,1}$ provides $(\epsilon,\delta)$-DP, where 
$\delta = 4 e^{-\frac{\epsilon^2 M}{256}}$. 
\end{theorem}
\begin{proof}
It is proved in~\cite{Dwork_EUROCRYPT06} that 
adding noise generated from $B(M,\frac{1}{2})$ provides $(\epsilon_0,\delta_0)$-unbounded DP, where $\delta_0 = 2 e^{-\frac{\epsilon_0^2 M}{64}}$. 
By group privacy~\cite{DP}, $(\epsilon_0,\delta_0)$-unbounded DP implies $(2\epsilon_0,2\delta_0)$-bounded DP, which shows Theorem~\ref{thm:SBin_DP_Dwork}.  
\end{proof}

\begin{theorem}[$(\epsilon,\delta)$-DP of $\calS_{\Bin,1}$ derived from~\cite{Agarwal_NeurIPS18}]
\label{thm:SBin_DP_Agarwal}
When $M \geq 4 \max\{23 \log(10d/\delta), 4\}$, $\calS_{\Bin,1}$ provides $(\epsilon,\delta)$-DP, where 
$\epsilon = \frac{c_1(\delta)}{\sqrt{M}} + \frac{c_2(\delta)}{M}$, 
$c_1(\delta) = 4\sqrt{2\log \frac{1.25}{\delta}}$, and 
$c_2(\delta) = \frac{8}{1 - \delta/10}\left( \frac{7\sqrt{2}}{4} \sqrt{\log\frac{10}{\delta}} + \frac{1}{3} \right) + \frac{16}{3}\left( \log\frac{1.25}{\delta} + \log\frac{20d}{\delta} \log\frac{10}{\delta} \right)$. 
\end{theorem}
\begin{proof}
Immediately derived from Theorem 1 of~\cite{Agarwal_NeurIPS18} with parameters $p=\frac{1}{2}$ and $s=1$ and sensitivity bounds $\Delta_1 = \Delta_2 = \Delta_\infty = 2$ (as we consider bounded DP).
\end{proof}

\arxiv{
\section{Performance Guarantees of Our Protocols}
\label{sec:loss_communication_proposed}
\subsection{Expected $l_2$ Loss of \SBin{}}
\label{sub:loss_SBin}

By (\ref{eq:SBin_delta}), when $\calS_{\Bin,\beta}$ provides $(\epsilon,\delta)$-DP, $M$ can be expressed 
as $M = \frac{2\log(4\beta / \delta)}{\eta^2}$. 
Thus, the expected $l_2$ loss of \SBin{} is 
$\frac{1-\beta}{\beta n} + \frac{d \log(4\beta/\delta)}{2 \eta^2 \beta^2 n^2}$. 

When $\epsilon$ is close to $0$ (i.e., $e^\epsilon \approx \epsilon + 1$), we have 
$\eta \approx \frac{\epsilon_0}{\epsilon_0 + 2} - \frac{2}{M(\epsilon_0 + 2)} \approx \frac{\epsilon_0}{2}$. 
By (\ref{eq:SBin_epsilon}), 
we also have 
$\frac{\epsilon}{2} \approx \beta \epsilon_0$. 
Thus, $\eta$ can be approximated as $\eta \approx \frac{\epsilon}{4\beta}$, and the expected $l_2$ loss of \SBin{} is approximated as 
$\frac{1-\beta}{\beta n} + \frac{8d \log(4\beta/\delta)}{\epsilon^2 n^2}$. 

\subsection{Communication Cost of \SBin{}}
\label{sub:communication_SBin}
Since $M = \frac{2\log(4\beta / \delta)}{\eta^2}$ and $\eta \approx \frac{\epsilon}{4\beta}$ (see Appendix~\ref{sub:loss_SBin}), $C_{tot}$ of \SBin{} can be approximated as: $C_{tot}=\alpha((1+\beta)n+\frac{Md}{2}) 
\approx \alpha ((1+\beta)n + \frac{16\beta^2 d \log(4\beta/\delta)}{\epsilon^2})$. 

\subsection{$\mu$ in \SAGeo{}}
\label{sub:mu_SAGeo}
Let $X_\ageo$ be a random variable with the asymmetric geometric distribution. 
We show that $\E[X_\ageo] = \mu$, where $\mu$ is given by (\ref{eq:SAGeo_mu_ast}). 
$\E[X_\ageo]$ can be written as follows: 
\begin{align}
\textstyle{\E[X_\ageo] = \frac{1}{\kappa} \left(\sum_{k=0}^{\nu-1} k q_l^{\nu - k} + \sum_{k=\nu}^\infty k q_r^{k - \nu} \right).}
\label{eq:E_X_ageo}
\end{align}
Let $X_\onegeo$ be a random variable with the one-sided geometric distribution with parameter $q \in (0,1)$, i.e.,  
\begin{align*}
\Pr(X_\onegeo = k) = (1-q) q^k ~~~~ \text{($k \in \nnints$)}.
\end{align*}
Since the expectation of $X_\onegeo$ is $\E[X_\onegeo] = \sum_{k=0}^\infty k(1-q)q^k = \frac{q}{1-q}$, the following equality holds for any $q \in (0,1)$: 
\begin{align}
\textstyle{\sum_{k=0}^\infty k q^k = \frac{q}{(1-q)^2}.}
\label{eq:k_q_k_q_1_q_2}
\end{align}
In addition, $\sum_{k=0}^\infty q^k = \frac{1}{1-q}$ for any $q \in (0,1)$. 
Thus, 
\begin{align}
\textstyle{\sum_{k=\nu}^\infty k q_r^{k - \nu}} 
&= \textstyle{\sum_{k=0}^\infty (k + \nu) q_r^k} \nonumber\\
&= \textstyle{\sum_{k=0}^\infty k q_r^k + \sum_{k=0}^\infty \nu q_r^k} \nonumber\\
&= \textstyle{\frac{q_r}{(1-q_r)^2} + \frac{\nu}{1-q_r}} \nonumber\\
&= \textstyle{\frac{q_r + (1-q_r)\nu}{(1 - q_r)^2}.}
\label{eq:sum_k_qr_kmu}
\end{align}
By (\ref{eq:E_X_ageo}) and (\ref{eq:sum_k_qr_kmu}), we have 
\begin{align}
\textstyle{\E[X_\ageo] = \frac{1}{\kappa} \left( \sum_{k=0}^{\nu-1} x q_l^{\nu-k} + \frac{q_r + (1-q_r)\nu}{(1 - q_r)^2}\right) = \mu.}
\label{eq:E_X_ageo_mu_ast}
\end{align}

\subsection{Expected $l_2$ Loss of \SAGeo{}}
\label{sub:loss_SAGeo}
When $\epsilon$ is close to $0$ (i.e., $e^\epsilon \approx \epsilon + 1$), the parameters $q_l$ and $q_r$ 
can be approximated as follows: 
$q_l = \frac{e^{-\epsilon/2} - 1 + \beta}{\beta} \approx \frac{2\beta - \epsilon}{2\beta}$, $q_r = \frac{\beta}{e^{\epsilon/2} - 1 + \beta} \approx \frac{2\beta}{\epsilon + 2\beta}$. 
Thus, we have 
\begin{align}
\kappa^* 
&= \textstyle{\frac{q_l}{1-q_l} + \frac{1}{1 - q_r} \approx \frac{2\beta - \epsilon}{\epsilon} + \frac{\epsilon + 2\beta}{\epsilon} = \frac{4\beta}{\epsilon}}
\label{eq:SAGeo_kappa_approx} \\
\textstyle{\frac{q_l(1+q_l)}{(1-q_l)^3}} 
&\approx \textstyle{\frac{\left( \frac{2\beta - \epsilon}{2\beta} \right)\left( \frac{4\beta - \epsilon}{2\beta} \right)}{\left( \frac{\epsilon}{2\beta} \right)^3} = \frac{2\beta(2\beta-\epsilon)(4\beta-\epsilon)}{\epsilon^3} \approx \frac{16\beta^3}{\epsilon^3}} \label{eq:SAGeo_q_l_1_q_l_approx} \\
\textstyle{\frac{q_r(1+q_r)}{(1-q_r)^3}} 
&\approx \textstyle{\frac{\left( \frac{2\beta}{\epsilon + 2\beta} \right)\left( \frac{\epsilon + 4\beta}{\epsilon + 2\beta} \right)}{\left( \frac{\epsilon}{\epsilon + 2\beta} \right)^3} = \frac{(\epsilon+2\beta)2\beta(\epsilon + 4\beta)}{\epsilon^3} \approx \frac{16\beta^3}{\epsilon^3}}. \label{eq:SAGeo_q_r_1_q_r_approx}
\end{align}
By (\ref{eq:SAGeo_kappa_approx}), (\ref{eq:SAGeo_q_l_1_q_l_approx}), and (\ref{eq:SAGeo_q_r_1_q_r_approx}), the expected $l_2$ loss can be approximated as follows: 
\begin{align*}
\textstyle{\E\left[ \sum_{i=1}^d (\hf_i - f_i)^2 \right]} 
&\leq \textstyle{\frac{1-\beta}{\beta n} \hspace{-0.5mm}+\hspace{-0.5mm} \frac{d}{\kappa^* \beta^2 n^2} \hspace{-0.5mm} \left( \frac{q_l(1+q_l)}{(1-q_l)^3} + \frac{q_r(1+q_r)}{(1-q_r)^3} \right)} \\
&\approx \textstyle{\frac{1-\beta}{\beta n} + \frac{\epsilon}{4\beta} \cdot \frac{d}{\beta^2 n^2} \cdot \frac{32\beta^3}{\epsilon^3}} \\
&\approx \textstyle{\frac{1-\beta}{\beta n} + \frac{8d}{\epsilon^2 n^2}}.
\end{align*}

\subsection{Communication Cost of \SAGeo{}}
\label{sub:communication_SAGeo}

By (\ref{eq:sum_k_qr_kmu}), $\mu$ in (\ref{eq:SAGeo_mu_ast}) can be upper bounded as follows:
\begin{align}
\mu 
&= \textstyle{\frac{1}{\kappa} \left( \sum_{k=0}^{\nu-1} x q_l^{\nu-k} + 
\frac{q_r + (1-q_r)\nu}{(1 - q_r)^2} \right)} \nonumber\\
&= \textstyle{\frac{1}{\kappa} \left( \sum_{k=0}^{\nu-1} x q_l^{\nu-k} + 
\sum_{k=0}^\infty k q_r^k + \sum_{k=0}^\infty \nu q_r^k \right)} \nonumber\\
&\leq \textstyle{\frac{1}{\kappa} \left( \nu \sum_{k=0}^{\nu-1} q_l^{\nu-k} + \nu \sum_{k=\nu}^\infty q_r^{k-\nu} + 
\sum_{k=0}^\infty k q_r^k \right)} \nonumber\\
&= \textstyle{\nu + \frac{1}{\kappa} \cdot \frac{q_r}{(1-q_r)^2} \hspace{5mm} \text{(by (\ref{eq:k_q_k_q_1_q_2}))}} \nonumber\\
&\leq \textstyle{\nu + \frac{q_r}{1-q_r}} \hspace{5mm} \text{(as $\kappa \geq 1/(1-q_r)$)}.
\label{eq:mu_ast_upper}
\end{align}
If $\beta > 1 - e^{-\epsilon/2}$ ($\iff q_l > 0$), then we have
\begin{align}
&\delta = \textstyle{\frac{2}{\kappa} q_l^\nu (1 - e^{\epsilon/2} + \beta e^{\epsilon/2})} \nonumber\\
\iff& \textstyle{\nu = \frac{\log \kappa \delta - \log 2(1 - e^{\epsilon/2} + \beta e^{\epsilon/2})}{\log q_l}.}
\label{eq:mu_q_l_positive}
\end{align}
Note that $q_l = \frac{e^{-\epsilon/2} - 1 + \beta}{\beta}$, $q_r = \frac{\beta}{e^{\epsilon/2} - 1 + \beta}$, and $\kappa = \frac{q_l(1-q_l^\nu)}{1-q_l} + \frac{1}{1 - q_r} \leq \frac{q_l}{1-q_l} + \frac{1}{1 - q_r}$. 
Thus, when $\epsilon$ is close to $0$ (i.e., $e^\epsilon \approx \epsilon + 1$), $\nu$ in (\ref{eq:mu_q_l_positive}) can be upper bounded as 
\begin{align}
\nu 
&\leq \textstyle{\frac{\log (\frac{q_l}{1-q_l} + \frac{1}{1 - q_r}) \delta - \log 2(1 - e^{\epsilon/2} + \beta e^{\epsilon/2})}{\log q_l}} \nonumber\\
&= \textstyle{\frac{\log (\frac{e^{-\epsilon/2} - 1 + \beta}{1-e^{-\epsilon/2}} + \frac{e^{\epsilon/2} - 1 + \beta}{e^{\epsilon/2} - 1}) \delta - \log 2(1 - e^{\epsilon/2} + \beta e^{\epsilon/2})}{\log (\frac{e^{-\epsilon/2} - 1 + \beta}{\beta})}} \nonumber\\
&\approx \textstyle{\frac{\log (\frac{\beta -\epsilon/2}{\epsilon/2} + \frac{\epsilon/2 + \beta}{\epsilon/2}) \delta - \log 2(-\epsilon/2 + \beta(\epsilon/2 + 1))}{\log (1 - \frac{\epsilon}{2\beta})}} \nonumber\\
&\approx \textstyle{\frac{\log \frac{4\beta\delta}{\epsilon} - \log 2\beta}{-\frac{\epsilon}{2\beta}} 
= \frac{2\beta \log \frac{\epsilon}{2\delta}}{\epsilon}.}
\label{eq:mu_q_l_positive_upper}
\end{align}
In addition, we have 
\begin{align}
\textstyle{\frac{q_r}{1-q_r} 
\approx \frac{\beta - \epsilon/2}{\epsilon/2}
\approx \frac{2\beta}{\epsilon}.}
\label{eq:q_r_1_q_r}
\end{align}
By (\ref{eq:mu_ast_upper}), (\ref{eq:mu_q_l_positive_upper}), and (\ref{eq:q_r_1_q_r}), we have 
\begin{align*}
\textstyle{\mu \leq \frac{2\beta (\log \frac{\epsilon}{2\delta} + 1)}{\epsilon}.}
\end{align*}
By (\ref{eq:C_tot_generalization}), $C_{tot}$ can be upper bounded as follows: 
$C_{tot} \leq \alpha ((1 + \beta)n + \frac{2\beta d (\log(\epsilon / 2\delta) + 1)}{\epsilon})$.

\subsection{Expected $l_2$ Loss of \SOGeo{}}
\label{sub:loss_SOGeo}
Let $\beta = 1 - e^{-\epsilon/2}$ and $q_r = \frac{1}{1+e^{\epsilon/2}}$. 
When $\epsilon$ is close to $0$ (i.e., $e^\epsilon \approx \epsilon + 1$), $\beta$ and $q_r$ can be approximated as 
\begin{align}
\beta \approx \textstyle{\frac{\epsilon}{2}}, ~
q_r \approx \textstyle{\frac{1}{\epsilon/2 + 2} \approx \frac{1}{2}.} \label{eq:SOGeo_qr_approx}
\end{align}
Thus, the expected $l_2$ loss can be approximated as follows: 
\begin{align*}
\textstyle{\E\left[ \sum_{i=1}^d (\hf_i - f_i)^2 \right]} 
&\leq \textstyle{\frac{1-\beta}{\beta n} + \frac{q_r d}{(1-q_r)^2 \beta^2 n^2}} \\
&\approx \textstyle{\frac{1-\epsilon/2}{\epsilon n/2} + \frac{d/2}{\epsilon^2 n^2 / 16}} \\
&\approx \textstyle{\frac{2}{\epsilon n} + \frac{8d}{\epsilon^2 n^2}.}
\end{align*}

\subsection{Communication Cost of \SOGeo{}}
\label{sub:communication_SOGeo}
By 
(\ref{eq:SOGeo_qr_approx}), when $\epsilon$ is close to $0$, 
we have 
\begin{align*}
\textstyle{C_{tot} = \alpha(n\beta + \frac{q_r d}{1-q_r}) \approx \frac{\epsilon n}{2} + d.}
\end{align*}
}

\arxiv{
\section{Performance Guarantees of the Existing Single-Message Shuffle Protocols}
\label{sec:existing_performance}

\begin{table*}[t]
\caption{Performance guarantees of the existing LDP mechanisms 
($\lambda$ ($= \frac{n'}{n+n'}$): fraction of fake users, $f_T$ ($= \sum_{i \in \calT} f_i$): frequencies over target items, $|\calT|$: \#target items).} 
\centering
\hbox to\hsize{\hfil
\begin{tabular}{c||c|c|c|c}
\hline
&   GRR &   OUE &   OLH &   RAPPOR\\
\hline
Expected $l_2$ loss 
&   $\frac{d(e^{\epsilon_L}+d-2)}{n(e^{\epsilon_L}-1)^2}$ & $\frac{4de^{\epsilon_L}}{n(e^{\epsilon_L}-1)^2}$ &   $\frac{4de^{\epsilon_L}}{n(e^{\epsilon_L}-1)^2}$ & $\frac{d e^{\epsilon_L / 2}}{n(e^{\epsilon_L / 2}-1)^2}$\\
\hline
Overall gain 
$\GMGA$    &   $\lambda (1 - f_T) + \frac{\lambda(d - |\calT|)}{e^{\epsilon_L}-1}$ &   $\lambda (2|\calT| - f_T) + \frac{2 \lambda |\calT|}{e^{\epsilon_L}-1}$ &   $\lambda (2|\calT| - f_T) + \frac{2 \lambda |\calT|}{e^{\epsilon_L}-1}$ & $\lambda (|\calT| - f_T) + \frac{\lambda |\calT|}{e^{\epsilon_L / 2}-1}$\\
\hline
\end{tabular}
\hfil}
\label{tab:performance_LDP}
\end{table*}

The expected $l_2$ loss of the GRR, OUE, OLH, and RAPPOR is shown in~\cite{Wang_USENIX17}. 
In addition, the overall gain $\GMGA$ of the GRR, OUE, and OLH is shown in~\cite{Cao_USENIX21}. 
Both~\cite{Wang_USENIX17} and~\cite{Cao_USENIX21} assume that the empirical estimator in~\cite{Wang_USENIX17} is used to calculate an unbiased estimate of $\bmf$. 
Table~\ref{tab:performance_LDP} summarizes 
the performance guarantees of these LDP mechanisms.\footnote{Note that~\cite{Wang_USENIX17} considers the absolute frequency, whereas we consider the relative frequency. 
We can obtain $\V[\hf_i]$ by dividing the values ($=\V[n \hf_i]/n$) in Table 1 of~\cite{Wang_USENIX17} by $n$. 
Then we can calculate the expected $l_2$ loss 
as: $\E[ \sum_{i=1}^d (\hf_i - f_i)^2 ] = \sum_{i=1}^d \V[\hf_i]$.} 
Since the overall gain $\GMGA$ of the RAPPOR is not shown in~\cite{Cao_USENIX21}, we calculate it 
in Proposition~\ref{prop:GMGA_RAPPOR}. 

Note that the values in Table~\ref{tab:performance_LDP} depend on the local privacy budget $\epsilon_L$, as the GRR, OUE, OLH, and RAPPOR are LDP mechanisms. 
Therefore, we first revisit the privacy amplification results in the existing shuffle model~\cite{Feldman_FOCS21} 
and express 
the local privacy budget $\epsilon_L$ 
using $\epsilon$, $\delta$, and $n$. 
Then, we calculate the expected $l_2$ loss and the overall gain $\GMGA$ of the existing shuffle protocols using $\epsilon$, $\delta$, and $n$.

\smallskip{}
\noindent{\textbf{$\epsilon_L$ in the General Case.}}~~By Theorem~\ref{thm:shuffle}, the existing shuffle protocols provide $(\epsilon, \delta)$-DP, where
\begin{align*}
\textstyle{\epsilon = \log \left( 1 + \frac{e^{\epsilon_L}-1}{e^{\epsilon_L}+1} \left( \frac{8\sqrt{e^{\epsilon_L} \log(4/\delta)}}{\sqrt{n}} + \frac{8 e^{\epsilon_L}}{n} \right) \right).}
\end{align*}
When $\epsilon_L$ is close to $0$, $e^{\epsilon_L}$ and $\epsilon$ can be approximated as $e^{\epsilon_L} \approx \epsilon_L + 1$ and $e^\epsilon \approx \epsilon + 1$, respectively. 
Then, for large $n$, we have
\begin{align*}
\epsilon 
&\approx \textstyle{\frac{\epsilon_L}{\epsilon_L+2} 
\left(\frac{8\sqrt{(\epsilon_L+1) \log(4/\delta)}}{\sqrt{n}} + \frac{8 (\epsilon_L+1)}{n} \right)} \\
&\approx \textstyle{\frac{\epsilon_L}{2} 
\left(\frac{8\sqrt{\log(4/\delta)}}{\sqrt{n}} + \frac{8}{n} \right)} 
\approx \textstyle{\frac{4\epsilon_L\sqrt{\log(4/\delta)}}{\sqrt{n}}.} 
\end{align*}
Thus, $\epsilon_L$ can be expressed as follows:
\begin{align}
\epsilon_L \approx \textstyle{\frac{\epsilon \sqrt{n}}{4 \sqrt{\log(4/\delta)}}.}
\label{eq:epsilon_L_general}
\end{align}

\smallskip{}
\noindent{\textbf{$\epsilon_L$ in \GRRS{}.}}~~Feldman \textit{et al.}~\cite{Feldman_FOCS21} show a tighter bound than Theorem~\ref{thm:shuffle} for the GRR (Corollary IV.2 in~\cite{Feldman_FOCS21}). 
Specifically, they show that a shuffle protocol based on the GRR provides $(\epsilon, \delta)$-DP, where
\begin{align*}
\textstyle{\epsilon = \log \left( 1 + (e^{\epsilon_L} - 1) \left(\frac{4\sqrt{2(d+1)\log(4/\delta)}}{\sqrt{(e^{\epsilon_L}+d-1)dn}} + \frac{4(d+1)}{dn} \right) \right).}
\end{align*}
When $\epsilon_L$ is close to $0$, we have
\begin{align*}
\epsilon 
&\approx \textstyle{\epsilon_L \left(\frac{4\sqrt{2(d+1)\log(4/\delta)}}{\sqrt{(\epsilon_L+d)dn}} + \frac{4(d+1)}{dn} \right)} \\
&\approx \textstyle{\frac{4\epsilon_L\sqrt{2(d+1)\log(4/\delta)}}{d\sqrt{n}}.}
\end{align*}
Thus, for the GRR, $\epsilon_L$ can be expressed as follows:
\begin{align}
\epsilon_L \approx \textstyle{\frac{\epsilon d\sqrt{n}}{4\sqrt{2(d+1)\log(4/\delta)}}.}
\label{eq:epsilon_L_GRR}
\end{align}

\smallskip{}
\noindent{\textbf{Expected $l_2$ Loss of \GRRS{}.}}~~By (\ref{eq:epsilon_L_GRR}) and Table~\ref{tab:performance_LDP}, when $\epsilon_L$ is close to $0$ (i.e., $e^{\epsilon_L} \approx \epsilon_L + 1$), the expected $l_2$ loss of \GRRS{} can be approximated as follows: 
\begin{align*}
\textstyle{\frac{d(e^{\epsilon_L}+d-2)}{n(e^{\epsilon_L}-1)^2}} 
\approx \textstyle{\frac{d(d - 1)}{\epsilon_L^2 n}} 
&= \textstyle{\frac{32(d+1)(d-1)\log(4/\delta)}{d \epsilon^2 n^2}} \\
&\approx \textstyle{\frac{32d\log(4/\delta)}{\epsilon^2 n^2} \hspace{5mm} \text{(for large $d$)}.}
\end{align*}

\smallskip{}
\noindent{\textbf{Expected $l_2$ Loss of \OUES{} and \OLHS{}.}}~~By Table~\ref{tab:performance_LDP}, \OUES{} and \OLHS{} achieve the same expected $l_2$ loss. 
By (\ref{eq:epsilon_L_general}), when $\epsilon_L$ is close to $0$, 
the expected $l_2$ loss is 
\begin{align*}
\textstyle{\frac{4de^{\epsilon_L}}{n(e^{\epsilon_L}-1)^2}}
\approx \textstyle{\frac{4d}{\epsilon_L^2 n}} 
\approx \textstyle{\frac{64d\log(4/\delta)}{\epsilon^2 n^2}.}
\end{align*}

\smallskip{}
\noindent{\textbf{Expected $l_2$ Loss of \RAPS{}.}}~~By (\ref{eq:epsilon_L_general}) and Table~\ref{tab:performance_LDP}, when $\epsilon_L$ is close to $0$, the expected $l_2$ loss of \RAPS{} is 
\begin{align*}
\textstyle{\frac{de^{\epsilon_L/2}}{n(e^{\epsilon_L/2}-1)^2}}
\approx \textstyle{\frac{d(\epsilon_L/2 +1)}{(\epsilon_L/2)^2 n}} 
\approx \textstyle{\frac{4d}{\epsilon_L^2 n}} 
\approx \textstyle{\frac{64d\log(4/\delta)}{\epsilon^2 n^2}.}
\end{align*}

\smallskip{}
\noindent{\textbf{$\GMGA$ of \GRRS{}.}}~~By (\ref{eq:epsilon_L_GRR}) and Table~\ref{tab:performance_LDP}, when $\epsilon_L$ is close to $0$ (i.e., $e^{\epsilon_L} \approx \epsilon_L + 1$), $\GMGA$ of \GRRS{} is 
\begin{align*}
\GMGA 
&= \textstyle{\lambda (1 - f_T) + \frac{\lambda(d - |\calT|)}{e^{\epsilon_L}-1}} \\
&\approx \textstyle{\lambda (1 - f_T) + \frac{4 \lambda(d - |\calT|)\sqrt{2(d+1)\log(4/\delta)}}{\epsilon  d\sqrt{n}}.}
\end{align*}

\smallskip{}
\noindent{\textbf{$\GMGA$ of \OUES{} and \OLHS{}.}}~~By Table~\ref{tab:performance_LDP}, \OUES{} and \OLHS{} achieve the same $\GMGA$. 
By (\ref{eq:epsilon_L_general}), when $\epsilon_L$ is close to $0$, 
$\GMGA$ is 
\begin{align*}
\GMGA 
&= \textstyle{\lambda (2|\calT| - f_T) + \frac{2 \lambda |\calT|}{e^{\epsilon_L}-1}} \\
&\approx \textstyle{\lambda (2|\calT| - f_T)+ \frac{8 \lambda |\calT|\sqrt{\log(4/\delta)}}{\epsilon \sqrt{n}}.}
\end{align*}

\smallskip{}
\noindent{\textbf{$\GMGA$ of \RAPS{}.}}~~We first show the overall gain $\GMGA$ of the RAPPOR in the local model. 
For $i\in[n]$, let $y_i = (y_i[1], \ldots, y_i[d]) \in \{0,1\}^d$ be a message sent from the $i$-th user. 
In the RAPPOR, the data collector calculates an unbiased estimate $\bmhf = (\hf_1, \cdots, \hf_d)$ of $\bmf$ 
as 
\begin{align}
\textstyle{\forall j \in [d]: \hf_j = \frac{(\sum_{i=1}^n y_i[j]) - nq}{n(1-2q)},}
\label{eq:RAPPOR_estimate}
\end{align}
where $q = \frac{1}{e^{\epsilon_L / 2} + 1}$ and $\epsilon_L$ is the privacy budget in the local model (see~\cite{Wang_USENIX17} for the proof that (\ref{eq:RAPPOR_estimate}) is unbiased). 
Then, $\GMGA$ can be written as follows:

\begin{proposition}
\label{prop:GMGA_RAPPOR}
Let $\lambda = \frac{n'}{n+n'}$ and $f_T = \sum_{i \in \calT} f_i$. 
$\GMGA$ of the RAPPOR with the privacy budget $\epsilon_L$ is 
\begin{align*}
\textstyle{\GMGA = \lambda (|\calT| - f_T) + \frac{\lambda |\calT|}{e^{\epsilon_L / 2}-1}.}
\end{align*}
\end{proposition}
\begin{proof}
For $i\in[n']$, let $y'_i \in \{0,1\}^d$ be a message sent from the $i$-th fake user and $\bmy' = (y'_1, \cdots, y'_{n'})$. 
By (\ref{eq:RAPPOR_estimate}), the estimate $\bmhf' = (\hf'_1, \cdots, \hf'_d)$ after data poisoning can be written as follows:
\begin{align}
\textstyle{\forall j \in [d]: \hf'_j = \frac{(\sum_{i=1}^{n+n'} y'_i[j]) - (n+n')q}{(n+n')(1-2q)}.}
\label{eq:RAPPOR_estimate_poisoning}
\end{align}
By (\ref{eq:RAPPOR_estimate}) and (\ref{eq:RAPPOR_estimate_poisoning}), we have
\begin{align*}
\E[\hf'_j] 
&= \textstyle{\frac{\E[\sum_{i=1}^{n} y'_i[j]] - nq + \E[\sum_{i=1}^{n'} y'_i[j]] - n'q}{(n+n')(1-2q)}} \\
&= \textstyle{\frac{n}{n+n'} \E[\hf_j] + \frac{\E[\sum_{i=1}^{n'} y'_i[j]] - n'q}{(n+n')(1-2q)},}
\end{align*}
and therefore, 
\begin{align*}
G(\bmy') 
&= \textstyle{\sum_{j \in \calT} \E[\Delta \hf_j]} \\
&= \textstyle{\sum_{j \in \calT} \E[\hf'_j - \hf_j]} \\
&= \textstyle{\sum_{j \in \calT} 
\left( \frac{\E[\sum_{i=1}^{n'} y'_i[j]] - n'q}{(n+n')(1-2q)} - \lambda \E[\hf_j] \right)} \\
&= \textstyle{\sum_{j \in \calT} 
\left( \frac{\E[\sum_{i=1}^{n'} y'_i[j]] - n'q}{(n+n')(1-2q)} \right) - \lambda f_T}~ \text{(as $\E[\hf_j] = f_j$)}.
\end{align*}
This is maximized when all the fake users set $y'_i[j] = 1$ for all target items $j \in \calT$. 
Thus, we have
\begin{align*}
\GMGA 
= \textstyle{\max_{\bmy'} G(\bmy')} 
&= \textstyle{\frac{n'|\calT|(1 - q)}{(n+n')(1-2q)} - \lambda f_T} \\
&= \textstyle{\frac{\lambda |\calT|(1 - q)}{1-2q} - \lambda f_T} \\
&= \textstyle{\frac{\lambda q |\calT|}{1-2q} + \lambda (|\calT| - f_T)} \\
&= \textstyle{\lambda (|\calT| - f_T) + \frac{\lambda |\calT|}{e^{\epsilon_L / 2} - 1}.}
\end{align*}
\end{proof}
By Proposition~\ref{prop:GMGA_RAPPOR} and  (\ref{eq:epsilon_L_general}), when $\epsilon_L$ is close to $0$, $\GMGA$ of \RAPS{} is 
\begin{align*}
\GMGA 
&= \textstyle{\lambda (|\calT| - f_T) + \frac{\lambda |\calT|}{e^{\epsilon_L / 2} - 1}} \\
&\approx \textstyle{\lambda (|\calT| - f_T) + \frac{2 \lambda |\calT|}{\epsilon_L}} \\
&\approx \textstyle{\lambda (|\calT| - f_T) + \frac{8 \lambda |\calT|\sqrt{\log(4/\delta)}}{\epsilon \sqrt{n}}.}
\end{align*}

}

\section{Existing Multi-Message Shuffle Protocols}
\label{sec:existing_mechanisms}

\subsection{\BC{}~\cite{Balcer_ITC20}}
\label{sub:BC20}
\noindent{\textbf{Protocol.}}~~In 
\BC{}, 
each user $u_i$ first sends $x_i$ to the shuffler. 
Then, for each item, 
user $u_i$ sends one dummy value to the shuffler with probability $q_1 \in [\frac{1}{2},1)$. 

\smallskip{}
\noindent{\textbf{Function $g$.}}~~If $q_1 = 1 - \frac{200}{\epsilon^2 n}\log\frac{4}{\delta}$ and $\epsilon \in (0,2]$, then \BC{} provides $(\epsilon, \delta)$-DP
(see Theorem 12 in~\cite{Balcer_ITC20}). 
Thus, $\epsilon$ can be expressed as $\epsilon = g(n,\delta)$, where 
$g(n,\delta) = \sqrt{\frac{200\log(4/\delta)}{(1-q_1)n}}$. 

\smallskip{}
\noindent{\textbf{Expected $l_2$ Loss.}}~~Let $a_{i,j}$ (resp.~$b_{i,j}$) $\in \{0,1\}$ be the number of input (resp.~dummy) values of user $u_i$ in the $j$-th item. 
Then, the unbiased estimate $\hf_j$ in \cite{Balcer_ITC20} is 
$\hf_j = \frac{1}{n}\sum_{i=1}^n \left( a_{i,j} + b_{i,j} \right) - q_1$. 
Since $q_1 = 1 - \frac{200}{\epsilon^2 n}\log\frac{4}{\delta} \geq \frac{1}{2}$, 
\begin{align*}
\textstyle{\E\left[ \sum_{i=1}^d (\hf_i - f_i)^2 \right]} 
= \textstyle{\frac{dnq_1(1-q_1)}{n^2}} 
\geq \textstyle{\frac{100d\log(4/\delta)}{\epsilon^2 n^2}.}
\end{align*}

\smallskip{}
\noindent{\textbf{Communication Cost.}}~~Since $q_1 \geq \frac{1}{2}$, we have 
$C_{tot} 
= 2\alpha n(1+dq_1) 
\geq 2\alpha n(1+\frac{d}{2})$. 

\smallskip{}
\noindent{\textbf{$\GMGA$.}}~~Assume that each fake user performs the following attack: (i) send one input value for an item randomly selected from target items $\calT$; (ii) send a dummy value for each target item in $\calT$; (iii) send a dummy value for each non-target item with probability $q_1$. 
This attack maximizes the overall gain.
Moreover, when $|\calT| \ll d$ or $n' \ll n$, the expected number of 
messages 
after poisoning is almost the same as that before poisoning, i.e., 
$n(1+dq_1) + n'(1+|\calT|+(d-|\calT|)q_1) \approx (n+n')(1+dq_1)$. 
Thus, this attack can avoid detection based on the number of 
messages. 
For $i \in \calT$, we have 
\begin{align*}
\E[\hf'_i] 
&= \textstyle{\frac{1}{n+n'}\left( n f_i + nq_1 + \frac{n'}{|\calT|} + n' \right) - q_1}\\
&= \textstyle{\frac{n}{n+n'}f_i + \frac{n'}{n+n'}\left(\frac{1}{|\calT|} + 1 - q_1 \right)}.
\end{align*}
Thus, by using $\lambda = \frac{n'}{n+n'}$ and $f_T = \sum_{i \in \calT} f_i$, we have 
\begin{align*}
\GMGA 
= \textstyle{\sum_{i \in \calT} (\E[\hf'_i - \hf_i])} 
= \textstyle{\lambda (1 - f_T + \frac{200|\calT|}{\epsilon^2 n}\log\frac{4}{\delta}).}
\end{align*}

\subsection{\CM{}~\cite{Cheu_SP22}}
\label{sub:CM22}
\noindent{\textbf{Protocol.}}~~In \CM{}, 
each user $u_i$ first maps $x_i \in [d]$ to 
a vector 
$y_{i,1} \in \{0,1\}^d$ with a $1$ in the $x_i$-th element and $0$'s elsewhere. 
Then, user $u_i$ adds $\xi \in \nats$ dummy values $y_{i,2}, \ldots, y_{i,\xi+1}$. 
Each dummy value is a $d$-dim zero vector; i.e., $0^d$. 
User $u_i$ flips each bit of 
$y_{i,1}, \ldots, y_{i,\xi+1}$ 
with probability 
$q_2$, 
where 
$q_2 = \frac{1}{e^{\epsilon_L/2}+1} (\leq \frac{1}{2})$ and $\epsilon_L \in \nnreals$. 
User $u_i$ sends $\tilde{y}_{i,1}, \ldots, \tilde{y}_{i,\xi+1}$, the flipped versions of $y_{i,1}, \ldots, y_{i,\xi+1}$, to the shuffler. 

\smallskip{}
\noindent{\textbf{Function $g$.}}~~If $q_2(1-q_2) \geq \frac{33}{5n\xi}(\frac{e^\epsilon+1}{e^\epsilon-1})^2 \log\frac{4}{\delta}$, then \CM{} provides $(\epsilon, \delta)$-DP
(see Claim III.1 in~\cite{Cheu_SP22}). 
Thus, $\epsilon$ can be expressed as $\epsilon = g(n,\delta)$, where 
$g(n,\delta) = \log (1 + \frac{2}{c_0-1})$ and $c_0 = \sqrt{\frac{5n\xi q_2(1-q_2)}{33 \log(4/\delta)}}$.

\smallskip{}
\noindent{\textbf{Expected $l_2$ Loss.}}~~There are $n(\xi+1)$ $d$-dim vectors in $\tilde{y}_{1,1}, \ldots, \tilde{y}_{n,\xi+1}$. 
For $i\in[n(\xi+1)]$ and $j\in[d]$, let $z_{i,j} \in \{0,1\}$ be the $j$-th element of the $i$-th vector. 
Then, the unbiased estimate $\hf_j$ in \cite{Cheu_SP22} is 
$\hf_j = \frac{1}{n(1-2q_2)} \{(\sum_{i=1}^{n(\xi + 1)} z_{i,j}) - n(\xi+1)q_2\}$. 
Thus, we have 
\begin{align*}
\textstyle{\V[\hf_i]} 
= \textstyle{\frac{1}{n^2(1-2q_2)^2}} n(\xi+1)q_2(1-q_2) 
= \textstyle{\frac{(\xi+1)q_2(1-q_2)}{n(1-2q_2)^2}}.
\end{align*}
When $\epsilon$ is close to $0$ (i.e., $e^\epsilon \approx \epsilon + 1$), we have $q_2(1-q_2) \geq \frac{33}{5n\xi}(\frac{e^\epsilon+1}{e^\epsilon-1})^2 \log\frac{4}{\delta} \approx \frac{132}{5 \epsilon^2 n \xi}\log\frac{4}{\delta}$. 
Thus, we have 
\begin{align*}
\textstyle{\E\left[ \sum_{i=1}^d (\hf_i - f_i)^2 \right]} 
= \textstyle{\frac{d(\xi+1)q_2(1-q_2)}{n(1-2q_2)^2}} 
\geq \textstyle{\frac{132d \log(4/\delta)}{5 \epsilon^2 n^2}.}
\end{align*}

\smallskip{}
\noindent{\textbf{Communication Cost.}}~~When $\epsilon$ is close to $0$, $\xi$ can be written as 
$\xi \approx \frac{132}{5q_2(1-q_2)n \epsilon^2} \log\frac{4}{\delta} \geq \frac{528}{5n\epsilon^2} \log\frac{4}{\delta}$. 
Thus, 
$C_{tot} 
= 2\alpha n(\xi+1) 
\geq 2\alpha (n + \frac{528}{5\epsilon^2} \log\frac{4}{\delta})$. 

\smallskip{}
\noindent{\textbf{$\GMGA$.}}~~Assume that each fake user $u_i$ ($n+1 \leq i \leq n+n'$) performs the following attack: (i) send $\tilde{y}_{i,1} \in \{0,1\}^d$ such that all elements in $\calT$ are 1's and there are $dq_2 - |\calT|$ 1's in $[d]\setminus\calT$; (ii) send $\tilde{y}_{i,2}, \ldots, \tilde{y}_{i,\xi+1} \in \{0,1\}^d$ honestly (i.e., flip each bit from $0$ with probability $q_2$). 
The expected number of 1's in noisy values ($= (n+n')dq_2(\xi+1)$) is the same before and after poisoning. 
Thus, this attack avoids detection based on the number of 1's or noisy values. 
The number of messages ($=(n+n')(\xi+1)$) is also the same. 
For $j \in \calT$, 
\begin{align*}
&\E[\hf'_j] 
= \textstyle{\frac{\sum_{i=1}^{(n+n')(\xi + 1)} \E[z_{i,j}] - (n+n')(\xi+1)q_2}{(n+n')(1-2q_2)}}\\
&= \textstyle{\frac{n \E[\hf_j]}{n+n'} + \frac{n' + n' \xi q_2 - n'(\xi+1)q_2}{(n+n')(1-2q_2)}}
= \textstyle{\frac{n f_j}{n+n'} + \frac{n'(1 - q_2)}{(n+n')(1-2q_2)}.} 
\end{align*}
Thus, by using $\lambda = \frac{n'}{n+n'}$ and $f_T = \sum_{i \in \calT} f_i$, we have
\begin{align*}
\GMGA 
\geq 
\textstyle{\frac{\lambda(1 - q_2)|\calT|}{1-2q_2} - \lambda f_T}
\geq \textstyle{\lambda (|\calT| - f_T).}
\end{align*}
The first inequality holds because the above attack may not maximize the overall gain. 

\subsection{\LWY{}~\cite{Luo_CCS22}}
\label{sub:LWY22}
\noindent{\textbf{Protocol.}}~~In \LWY{}, 
each user $u_i$ first sends $x_i$ to the shuffler. 
Then, user $u_i$ adds one dummy value randomly selected from $[d]$ to the shuffler with probability $q_3 \in [0,1]$. 

\smallskip{}
\noindent{\textbf{Function $g$.}}~~If $q_3 = \frac{32d \log(2/\delta)}{\epsilon^2 n}$ and $\epsilon \in (0,3]$, then \LWY{} provides $(\epsilon, \delta)$-DP 
(see Lemma 3.2 in~\cite{Luo_CCS22}). 
Thus, 
$\epsilon$ can be expressed as $\epsilon = g(n,\delta)$, where 
$g(n,\delta) = \sqrt{\frac{32d\log(2/\delta)}{q_3 n}}$. 

\smallskip{}
\noindent{\textbf{Expected $l_2$ Loss.}}~~Let $a_{i,j}$ (resp.~$b_{i,j}$) $\in \{0,1\}$ be the number of input (resp.~dummy) values of user $u_i$ in the $j$-th item. 
Then, the unbiased estimate $\hf_j$ in \cite{Luo_CCS22} is 
$\hf_j = \frac{1}{n}\sum_{i=1}^n \left( a_{i,j} + b_{i,j} \right) - \frac{q_3}{d}$. 
Since $q_3 = \frac{32d \log(2/\delta)}{\epsilon^2 n} \ll d$, 
\begin{align*}
\textstyle{\E\left[ \sum_{i=1}^d (\hf_i - f_i)^2 \right]} 
= \textstyle{\frac{dn \frac{q_3}{d}(1 - \frac{q_3}{d})}{n^2}} 
\approx \textstyle{\frac{32d \log(2/\delta)}{\epsilon^2 n^2}}.
\end{align*}

\smallskip{}
\noindent{\textbf{Communication Cost.}}~~$C_{tot} 
= 2\alpha n(1+q_3) 
= 2\alpha n(1+\frac{32d \log(2/\delta)}{\epsilon^2 n})$.

\smallskip{}
\noindent{\textbf{$\GMGA$.}}~~Assume that each fake user performs the following attack: (i) send one input value for an item randomly selected from target items $\calT$; (ii) send a dummy value for an item randomly selected from target items 
$\calT$ with probability $q_3$. 
Then, the expected number of messages after poisoning ($= (n+n')(1+q_3)$) is the same as that before poisoning. 
Thus, this attack avoids detection. 
For $i \in \calT$, we have 
\begin{align*}
\E[\hf'_i] 
&= \textstyle{\frac{1}{n+n'}\left( n f_i + \frac{nq_3}{d} + \frac{n'}{|\calT|} + \frac{n'q_3}{|\calT|} \right) - \frac{q_3}{d}}\\
&= \textstyle{\frac{n}{n+n'}f_i + \frac{n'}{n+n'}\left(\frac{1+q_3}{|\calT|} - \frac{q_3}{d} \right)}.
\end{align*}
Thus, by using $\lambda = \frac{n'}{n+n'}$ and $f_T = \sum_{i \in \calT} f_i$, we have 
\begin{align*}
&\GMGA 
\geq 
\textstyle{\lambda \sum_{i \in \calT} \left(\frac{1+q_3}{|\calT|} - \frac{q_3}{d}  - f_i\right)} \\
&\text{(as the above attack may not maximize the overall gain)} \\
&= \textstyle{\lambda (1 - f_T) + \frac{32 \lambda (d-|\calT|) \log(2/\delta)}{\epsilon^2 n}.}
\end{align*}

\arxiv{
\section{Remark on the Use of Public-key Encryption}

In our protocol, we assume that messages between users and the shuffler and those between the shuffler and the data collector are encrypted by using a public-key encryption (PKE) scheme.
We briefly remark on how DP of our protocol relates to the security of the underlying PKE scheme.

First, if the PKE scheme satisfies indistinguishability under chosen-plaintext attacks (or IND-CPA security), which is the most standard security requirements for PKE \cite{Katz_07}, our protocol satisfies computational DP \cite{MPRV09} against the shuffler, that is, DP holds as long as the shuffler is computationally bounded.
More specifically, our protocol provides a slightly stronger guarantee that for any two neighboring tuples $(x_i)_{i\in[n]}$, $(x_i')_{i\in[n]}$ of users' inputs, the view of the shuffler during the execution of the protocol with 
one tuple 
is computationally indistinguishable from that with 
the other tuple. 
This can be seen from a straightforward reduction to the IND-CPA security of the PKE scheme:
Given an algorithm $\mathcal{A}$ that distinguishes between the two views of the shuffler, we can construct an algorithm $\mathcal{A}'$ that breaks the IND-CPA security as follows:
$\mathcal{A}'$ chooses $n-1$ inputs $x_1,\ldots,x_{n-1}$ and calls an encryption oracle to obtain ciphertexts $(\mathsf{ct}_1,\ldots,\mathsf{ct}_{n-1})$ of them.
$\mathcal{A}'$ then chooses a pair of messages $x_n$, $x_n'$ and sends them to a challenger.
Given a challenge ciphertext $\mathsf{ct}^*$, $\mathcal{A}'$ gives $(\mathsf{ct}_1,\ldots,\mathsf{ct}_{n-1},\mathsf{ct}^*)$ to $\mathcal{A}$ and outputs the bit output by $\mathcal{A}$.
If $\mathcal{A}$ distinguishes neighboring inputs with non-negligible advantage, then $\mathcal{A}'$ breaks the IND-CPA security with the same advantage.

Second, DP against the data collector holds statistically, independent of the computational security of PKE.
The view of the data collector during the execution of $\mathcal{S}_{\mathcal{D},\beta}$ consists of a pair $(\mathsf{pk},\mathsf{sk})$ of public and secret keys generated by himself, and ciphertexts of messages $(\tx_{\pi(1)},\ldots,\tx_{\pi(\tn^*)})$ computed by the shuffler.
If the dummy distribution $\mathcal{D}$ is appropriately chosen, the distribution of the messages $(\tx_{\pi(1)},\ldots,\tx_{\pi(\tn^*)})$ satisfies DP.
Since the key pair $(\mathsf{pk},\mathsf{sk})$ is independently generated, the view of the data collector can be simulated from the messages $(\tx_{\pi(1)},\ldots,\tx_{\pi(\tn^*)})$ only.
Statistical DP is then guaranteed due to the post-processing property.
}

\newpage % The Meta-Review should at least start on a new column

\section{Meta-Review}

The following meta-review was prepared by the program committee for the 2025
IEEE Symposium on Security and Privacy (S\&P) as part of the review process as
detailed in the call for papers.

\subsection{Summary}
This paper considers an extension of the shuffle model of differential privacy (DP) where the shuffler has two additional abilities: downsampling real records, and generating dummy records (the augmented shuffle model). They propose a general algorithm to compute the relative frequencies of values where no noise is added to individual input records, making the method robust to data poisoning and collusion. 

\subsection{Scientific Contributions}
\begin{itemize}
\item Provides a Valuable Step Forward in an Established Field
\end{itemize}

\subsection{Reasons for Acceptance}
\begin{enumerate}
\item This paper provides a valuable step forward in an established field. The shuffle model of DP has seen a lot of interest recently, but existing solutions in the shuffle model have been shown to be vulnerable to data poisoning and to collusion between the data collector and users. This paper introduces an extended shuffle model, where the shuffler is given additional but realistic capabilities. The authors show that, in this extended shuffle model, they are able to design methods that address these shortcomings of shuffle DP, while also improving the utility of the algorithms.
\end{enumerate}

\subsection{Noteworthy Concerns} % Exclude if your meta-review does not have noteworthy concerns
\begin{enumerate} % Enumerate environment is not necessary if there is only one
\item The method introduced in the paper can be hard to scale. The computational complexity scales linearly with the cartesian product of the possible values of all attributes; running the algorithm may be impractical when there are many attributes/attribute values. The method also increases the communication costs to the server in a way that might not be practical in real-world deployments.
\item The utility (i.e., estimated frequencies) depends on the choice of the dummy count distribution $\mathcal{D}$; poor choices of $\mathcal{D}$ can lead to poor performance in this sense. The authors have provided suggestions as to how this distribution can be chosen.
\end{enumerate}

\end{document}